\documentclass[preprint]{elsarticle}
\pdfoutput=1

\usepackage[utf8]{inputenc}
\usepackage{url}

\usepackage{hyperref}

\makeatletter
\providecommand{\doi}[1]{%
  \begingroup
  \let\bibinfo\@secondoftwo
  \urlstyle{rm}%
  \href{http://dx.doi.org/#1}{%
    doi:\discretionary{}{}{}%
    \nolinkurl{#1}%
  }%
  \endgroup
}
\makeatother

\usepackage{wrapfig}
\usepackage{subfig}
\usepackage{xspace}

\usepackage{amsthm,amsmath,amssymb}
\newtheorem{observation}{Observation}
\newtheorem{proposition}{Proposition}
\newtheorem{theorem}{Theorem}
\newtheorem{corollary}{Corollary}
\newtheorem{lemma}{Lemma}

\newcommand{\gamename}{\textsc}

\newcommand{\fpt}{\textsf{\small FPT}\xspace}
\newcommand{\np}{\textsf{\small NP}\xspace}
\newcommand{\p}{\textsf{\small P}\xspace}
\newcommand{\pspace}{\textsf{\small PSPACE}\xspace}
\newcommand{\wone}{\textsf{W[1]}\xspace}
\newcommand{\GG}{\textsc{gg}\xspace}
\newcommand{\hex}{\textsc{hex}\xspace}
\newcommand{\Hex}{\textsc{Hex}\xspace}
\newcommand{\GHex}{\textsc{generalized hex}\xspace}
\newcommand{\SHex}{\textsc{short hex}\xspace}
\newcommand{\sGHex}{\textsc{short generalized hex}\xspace}
\newcommand{\SGHex}{\textsc{Short generalized hex}\xspace}
\newcommand{\Hav}{\textsc{havannah}\xspace}
\newcommand{\Sl}{\textsc{Slither}\xspace}
\newcommand{\slither}{\textsc{slither}\xspace}
\newcommand{\Twixt}{\textsc{Twixt}\xspace}
\newcommand{\twixt}{\textsc{twixt}\xspace}

\usepackage{microtype}
\usepackage{nicefrac}

\usepackage{tikz}
\usetikzlibrary{fit,arrows,shapes}

\tikzset{
  hexsid/.style={red,rounded corners,very thick},
  slitherw/.style={draw,circle,fill=white},
  slitherb/.style={draw,circle,fill=black},
}

\newcommand{\s}{0.5}
\newcommand{\bor}{0.1}

\newcommand{\possw}[2]{\draw        ({(#1 + 0.25) * \s}, {(#2 + 0.25) * \s}) circle ({0.1 * \s}) {};}
\newcommand{\possb}[2]{\draw [fill] ({(#1 - 0.25) * \s}, {(#2 + 0.25) * \s}) circle ({0.1 * \s}) {};}
\newcommand{\sboard}[2]{
  \foreach \i in {1,...,#1} {\foreach \j in {1,...,#2} {\node (\i-\j) at ({\i * \s},{\j * \s}) {};} }
  \foreach \i in {0,...,#1} {
    \draw ({(\i+1) * \s - \s/2},{\s - \s/2}) -- ({(\i+1) * \s - \s/2},{(#2 + 1) * \s - \s/2});
  }
  \foreach \j in {0,...,#2} {
    \draw ({\s - \s/2},{(\j+1) * \s - \s/2}) -- ({(#1 +1) * \s - \s/2},{(\j+1) * \s - \s/2});
  }
}
\newcommand{\sposition}[2]{
  \foreach \k in {#1} { \node[slitherw] at (\k) {}; }
  \foreach \k in {#2} { \node[slitherb] at (\k) {}; }
}
\newcommand{\spositionp}[4]{
  \foreach \i/\j in {#3} {
    \pgfmathtruncatemacro{\ii}{\i + #1} \pgfmathtruncatemacro{\jj}{\j + #2}
    \node [slitherw] at (\ii-\jj) {};
  }
  \foreach \i/\j in {#4} {
    \pgfmathtruncatemacro{\ii}{\i + #1} \pgfmathtruncatemacro{\jj}{\j + #2}
    \node [slitherb] at (\ii-\jj) {};
  }
}
\newcommand{\slithercell}[2]{
  \spositionp{#1}{#2}{1/1,1/2,1/3,1/4,1/5,1/6,1/7,1/8,2/1,2/2,2/3,2/5,2/6,2/7,2/8,3/1,3/2,3/7,3/8,5/2,5/6,5/7,6/5,6/6,7/3,7/5,8/3,8/4,8/5}{}
  \spositionp{#1}{#2}{}{1/6,2/4,2/6,3/3,3/4,4/1,4/2,4/3,4/7,4/8,5/1,5/8,6/1,6/2,6/7,6/8,7/1,7/2,7/4,7/6,7/7,7/8,8/1,8/2,8/6,8/7,8/8}
}
\newcommand{\slitherecell}[2]{
  \slithercell{#1}{#2}
  \spositionp{#1}{#2}{4/4}{5/5}
}

\newcommand{\sedges}[2]{
  \draw [fill] ({\s / 2 - \bor},{\s / 2}) rectangle ({\s / 2},{(#2 + \s)/ 2}) {};
  \draw [fill] ({(#1 + \s)/ 2},{\s / 2}) rectangle ({(#1 + \s)/ 2 + \bor},{(#2 + \s)/ 2}) {};
  \draw [fill=white] ({\s / 2},{\s / 2 - \bor}) rectangle ({(#1 + \s)/ 2},{\s / 2}) {};
  \draw [fill=white] ({\s / 2},{(#2 + \s) / 2}) rectangle ({(#1 + \s)/ 2},{(#2 + \s) / 2 + \bor}) {};
}
\newcommand{\hboard}[2]{
  \foreach \i in {1,...,#1} {\foreach \j in {1,...,#2} {\node (\i-\j) at ({\i * \s},{\j * \s}) {};} }
  \begin{scope}[xshift=-0.25cm, yshift=-0.25cm]
  \foreach \i in {1,...,#1} {
    \foreach \j in {1,...,#2} {
      \draw ({(\i - 0.125) * \s}, {(\j + 0.125) * \s}) -- ({\i * \s}, {(\j + 0.25) * \s}) -- ({\i * \s}, {(\j + 0.75) * \s}) -- ({(\i + 0.25) * \s}, {(\j + 1) * \s}) -- ({(\i + 0.75) * \s}, {(\j + 1) * \s}) -- ({(\i + 0.875) * \s}, {(\j + 1.125) * \s}) -- ({(\i + 1) * \s}, {(\j + 1) * \s}) -- ({(\i + 1.125) * \s}, {(\j + 0.875) * \s}) -- ({(\i + 1) * \s}, {(\j + 0.75) * \s}) -- ({(\i + 1) * \s}, {(\j + 0.25) * \s}) -- ({(\i + 0.75) * \s}, {(\j + 0) * \s}) -- ({(\i + 0.25) * \s}, {(\j + 0) * \s}) -- ({(\i + 0.125) * \s}, {(\j - 0.125) * \s}) -- ({(\i - 0.125) * \s}, {(\j + 0.125) * \s});
    }
  }
  \end{scope}
}
\newcommand{\twboardr}[2]{
  \foreach \i in {1,...,#1} { \foreach \j in {1,...,#2} { \coordinate (\i-\j) at (\i,\j) {} ;}}

  \pgfmathsetmacro{\ii}{#1 - 0.5} \pgfmathsetmacro{\jj}{#2 - 0.5} \pgfmathsetmacro{\kk}{#2 - 0.4}
  \draw [line width=2pt] (1.5,1.5) -- (\ii,1.5); \draw [line width=2pt] (1.5,\jj) -- (\ii,\jj);
  \draw [double distance=1.8pt, very thin] (1.5,1.4) -- (1.5,\kk); \draw [double distance=1.8pt, very thin] (\ii,1.4) -- (\ii,\kk);

  \pgfmathtruncatemacro{\n}{#1 - 1} \pgfmathtruncatemacro{\m}{#2 - 1}
  \foreach \i in {2,...,\n} {
    \foreach \j in {2,...,\m} {
      \node [tw-empty] at (\i,\j) {} ;
    }
  }
  \foreach \i in {2,...,\n} {
    \node [tw-empty] at (\i, 1) {} ;
    \node [tw-empty] at (\i,#2) {} ;
  }
  \foreach \i in {2,...,\m} {
    \node [tw-empty] at ( 1,\i) {} ;
    \node [tw-empty] at (#1,\i) {} ;
  }
}

\newcommand{\twboard}[2]{
  \foreach \i in {1,...,#1} { \foreach \j in {1,...,#2} { \coordinate (\i-\j) at (\i,\j) {} ;}}

  \pgfmathsetmacro{\ii}{#1 - 0.5} \pgfmathsetmacro{\jj}{#2 - 0.6} \pgfmathsetmacro{\kk}{#2 - 0.4}
  \draw [very thin] (1.5,1.4) -- (\ii,1.4);
  \draw [very thin] (1.5,1.55) -- (\ii,1.55);
  \draw [very thin] (1.5,\jj) -- (\ii,\jj);
  \pgfmathsetmacro{\jjj}{\jj + 0.15}
  \draw [very thin] (1.5,\jjj) -- (\ii,\jjj);
  \draw [line width=2pt] (1.5,1.4) -- (1.5,\kk); \draw [line width=2pt] (\ii,1.4) -- (\ii,\kk);

  \pgfmathtruncatemacro{\n}{#1 - 1} \pgfmathtruncatemacro{\m}{#2 - 1}
  \foreach \i in {2,...,\n} {
    \foreach \j in {2,...,\m} {
      \node [tw-empty] at (\i,\j) {} ;
    }
  }
  \foreach \i in {2,...,\n} {
    \node [tw-empty] at (\i, 1) {} ;
    \node [tw-empty] at (\i,#2) {} ;
  }
  \foreach \i in {2,...,\m} {
    \node [tw-empty] at ( 1,\i) {} ;
    \node [tw-empty] at (#1,\i) {} ;
  }
}

\definecolor{light-gray}{gray}{0.92}
\tikzset{
  tw-white/.style={draw,circle,fill=white,inner sep=2},
  tw-black/.style={draw,circle,fill=black,inner sep=2},
  tw-empty/.style={draw,circle,fill=black,inner sep=0.2},
  hav-white/.style={draw,circle,fill=light-gray,minimum size=3.5mm,inner sep=0},
  hav-black/.style={draw,circle,fill=black,minimum size=3.5mm,inner sep=0,text=white},
  hav-empty/.style={regular polygon,regular polygon sides=6,draw,minimum size=5.85mm,inner sep=0},
  hav-white-c/.style={draw,dashed,circle,fill=white,inner sep=2},
  hav-black-c/.style={draw,circle,fill=gray,inner sep=2},
}
\newcommand{\havcoordinate}[4]{
  \foreach \i in {#1,...,#2} {
    \foreach \j in {#3,...,#4} {
      \coordinate (\i-\j) at ({(\i+1) * 1.732},{(\j)}) {};
    }
  }
  \foreach \i in {#1,...,#2} {
    \foreach \j in {#3,...,#4} {
      \coordinate (\i--\j) at ({(\i+1+1/2) * 1.732},{(\j+1/2)}) {};
    }
  }
}

\newcommand{\twoone}[4]{
  \foreach \i/\j in {1/1,1/2,1/3,1/5,2/0,2/6,3/0,3/6,4/1,4/2,4/3,4/5} {
    \pgfmathtruncatemacro{\ii}{\i + #1}
    \pgfmathtruncatemacro{\jj}{\j + #2}
    \node [hav-#3] at (\ii-\jj) {};
  }
  \foreach \i/\j in {0/3,1/0,1/5,2/5,3/0,3/5,4/3} {
    \pgfmathtruncatemacro{\ii}{\i + #1}
    \pgfmathtruncatemacro{\jj}{\j + #2}
    \node [hav-#3] at (\ii--\jj) {};
  }
  \foreach \i/\j in {2/1,2/3,3/1,3/3} {
    \pgfmathtruncatemacro{\ii}{\i + #1}
    \pgfmathtruncatemacro{\jj}{\j + #2}
    \node [hav-#4] at (\ii-\jj) {};
  }
  \foreach \i/\j in {1/1,2/0,2/4,3/1} {
    \pgfmathtruncatemacro{\ii}{\i + #1}
    \pgfmathtruncatemacro{\jj}{\j + #2}
    \node [hav-#4] at (\ii--\jj) {};
  }
}
\newcommand{\twooneb}[4]{
  \foreach \i/\j in {1/1,1/2,1/3,1/5,2/0,2/6,3/0,3/6,4/1,4/2,4/3,4/5} {
    \pgfmathtruncatemacro{\ii}{\i + #1}
    \pgfmathtruncatemacro{\jj}{6-\j + #2}
    \node [hav-#3] at (\ii-\jj) {};
  }
  \foreach \i/\j in {0/3,1/0,1/5,2/5,3/0,3/5,4/3} {
    \pgfmathtruncatemacro{\ii}{\i + #1}
    \pgfmathtruncatemacro{\jj}{5-\j + #2}
    \node [hav-#3] at (\ii--\jj) {};
  }
  \foreach \i/\j in {2/1,2/3,3/1,3/3} {
    \pgfmathtruncatemacro{\ii}{\i + #1}
    \pgfmathtruncatemacro{\jj}{6-\j + #2}
    \node [hav-#4] at (\ii-\jj) {};
  }
  \foreach \i/\j in {1/1,2/0,2/4,3/1} {
    \pgfmathtruncatemacro{\ii}{\i + #1}
    \pgfmathtruncatemacro{\jj}{5-\j + #2}
    \node [hav-#4] at (\ii--\jj) {};
  }
}
\newcommand{\twoonec}[4]{
  \foreach \i/\j in {0/2,0/3,0/4,1/0,1/5,2/6,3/1,3/5} {
    \pgfmathtruncatemacro{\ii}{\i + #1}
    \pgfmathtruncatemacro{\jj}{\j + #2}
    \node [hav-#3] at (\ii--\jj) {};
  }
  \foreach \i/\j in {1/1,1/2,1/6,2/0,2/6,3/0,3/1,3/6,4/2,4/3,4/4} {
    \pgfmathtruncatemacro{\ii}{\i + #1}
    \pgfmathtruncatemacro{\jj}{\j + #2}
    \node [hav-#3] at (\ii-\jj) {};
  }
  \foreach \i/\j in {1/2,2/2,2/5,3/2,3/3,3/4} {
    \pgfmathtruncatemacro{\ii}{\i + #1}
    \pgfmathtruncatemacro{\jj}{\j + #2}
    \node [hav-#4] at (\ii--\jj) {};
  }
  \foreach \i/\j in {2/4,3/5} {
    \pgfmathtruncatemacro{\ii}{\i + #1}
    \pgfmathtruncatemacro{\jj}{\j + #2}
    \node [hav-#4] at (\ii-\jj) {};
  }
}

\newcommand{\startingvertex}[2]{
  \foreach \i/\j in {1/3,2/1,2/3,3/2} {
    \pgfmathtruncatemacro{\ii}{\i + #1}
    \pgfmathtruncatemacro{\jj}{\j + #2}
    \node[hav-white] at (\ii-\jj) {};
  }
  \foreach \i/\j in {1/1,2/0,2/2} {
    \pgfmathtruncatemacro{\ii}{\i + #1}
    \pgfmathtruncatemacro{\jj}{\j + #2}
    \node[hav-white] at (\ii--\jj) {};
  }
  \foreach \i/\j in {2/2} {
    \pgfmathtruncatemacro{\ii}{\i + #1}
    \pgfmathtruncatemacro{\jj}{\j + #2}
    \node[hav-black] at (\ii-\jj) {};
  }
  \foreach \i/\j in {0/3} {
    \pgfmathtruncatemacro{\ii}{\i + #1}
    \pgfmathtruncatemacro{\jj}{\j + #2}
    \node[hav-black] at (\ii--\jj) {};
  }
}

\newcommand{\alphabet}{A,B,C,D,E,F,G,H,I}

\title{On the Complexity of Connection Games}

\author{Édouard Bonnet}
\address{\textup{\url{edouard.bonnet@lamsade.dauphine.fr}}\\\textsc{Sztaki}, Hungarian Academy of Sciences}
\author{Florian Jamain}
\address{\textup{\url{florian.jamain@lamsade.dauphine.fr}}\\\textsc{Lamsade}, Université Paris-Dauphine}
\author{Abdallah Saffidine}
\address{\textup{\url{abdallahs@cse.unsw.edu.au}}\\\textsc{Cse}, The University of New South Wales}

\begin{document}
\begin{abstract}
  In this paper, we study three connection games among the most widely played: \gamename{havannah}, \gamename{twixt}, and \gamename{slither}.
  We show that determining the outcome of an arbitrary input position is \pspace-complete in all three cases.
  Our reductions are based on the popular graph problem \gamename{generalized geography} and on \gamename{hex} itself.
  We also consider the complexity of generalizations of \gamename{hex} parameterized by the length of the solution and establish that while \sGHex is \wone-hard, \SHex is \fpt.
  Finally, we prove that the ultra-weak solution to the empty starting position in \gamename{hex} cannot be fully adapted to any of these three games.
\end{abstract}

\maketitle

\section{Introduction}

Since its independent inventions in 1942 and 1948 by the poet and mathematician Piet Hein and the economist and mathematician John Nash, the game of \gamename{hex} has acquired a special spot in the heart of abstract game aficionados.
Its purity and depth has lead Jack van Rijswijck to conclude his PhD thesis with the following hyperbole~\cite{vanRijswijck2006}.
\begin{quote}
  \gamename{Hex} has a Platonic existence, independent of human thought.
  If ever we find an extraterrestrial civilization at all, they will know \gamename{hex}, without any doubt.
\end{quote}
\gamename{Hex} not only exerts a fascination on players, but it is the root of the field of \emph{connection games} which is being actively explored by game designers and researchers alike~\cite{Browne2005}.

A connection game is a kind of abstract strategy game in which players try to make a specific type of connection with their pieces~\cite{Browne2005}.
In many connection games, the goal is to connect two opposite sides of a board.
In these games, players take turns placing and/or moving pieces until they connect the two sides of the board.
\gamename{Hex}, \gamename{y}, and \gamename{twixt} are typical examples of this type of game.
However, a connection game can also involve completing a loop (\gamename{havannah}), connecting all the pieces of a color (\gamename{lines of action}), or forbidden patterns (\gamename{slither}).

The focus of research on abstract strategy games, and on connection games in particular, includes the design and programming of strong artificial players and solvers~\cite{ArnesonHH2010,Ewalds2012}, as well as theoretical considerations on aspects specific to connection games such a virtual connections and inferior cells~\cite{vanRijswijck2006,Henderson2010,Steane2012,Steane2013}.

Developing an artificial player for a strategy game typically involves adapting standard techniques from the game search literature, in particular the classical Alpha-Beta algorithm~\cite{Anshelevich2002} or the more recent Monte Carlo Tree Search paradigm~\cite{BrownePWLCRTPSC2012,ArnesonHH2010}.
These algorithms which explore an exponentially large game tree are meaningful when optimal polynomial time algorithms are impossible or unlikely.
For instance, tree search algorithms would not be used for \gamename{nim} and \gamename{Shannon's edge switching game} which can be played optimally and solved in polynomial time~\cite{BrunoW1970}.

\emph{Computational complexity} is a theoretical tool used to gain formal intuition on families of problems.
It can indicate that a problem is unlikely to be solvable in polynomial time and that exponential algorithms might be the best bet.
The complexity class \pspace{} comprises those problems that can be solved on a Turing machine using an amount of space polynomial in the size of the input.
The prototypical example of a \pspace{}-complete problem is the Quantified Boolean Formula problem (\textsc{qbf}) which can be seen as a generalization of \textsc{sat} allowing for variables to be both existentially and universally quantified.
Proving that a game is \pspace{}-hard shows that a variety of intricate problems can be encoded via positions of this game.
Additionally, it is widely believed in complexity theory that if a problem is \pspace-hard, then it admits no polynomial time algorithms.

For this reason, studying the computational complexity of games is also a popular research topic.
The complexity class of \gamename{chess} and \gamename{go} was determined shortly after the very definition of the appropriate classes and other popular games have been classified since then~\cite{DemaineH2009,HearnDemaine2009}.
More recently, we studied the complexity of trick-taking card games which notably include \gamename{bridge}, \gamename{skat}, \gamename{tarot}, and \gamename{whist}~\cite{BonnetJamainSaffidine2013IJCAI}.

Connection games have received less attention.
Besides Even and Tarjan's proof that \gamename{Shannon's vertex switching game} is \pspace-complete~\cite{EvenTarjan1976} and Reisch's proof that \gamename{hex} is \pspace-complete~\cite{Reisch1981}, the only complexity results on connection games that we know of are the \pspace-completeness of virtual connection detection~\cite{Kiefer2003} in \gamename{hex}, the \np-completeness of dominated cell detection in \gamename{Shannon's vertex switching game}~\cite{BjornssonHJvR2007}, as well as an unpublished note showing that a problem related to \gamename{twixt} is \np-complete~\cite{MazzoniW1997}.\footnote{For a summary in English of Reisch's reduction, see Maarup's thesis~\cite{Maarup2005}.}

The games that we study in this paper rank among the most notable connection games.
The game of \gamename{twixt} by Alex Randolph~\cite{Whitehill2006} was first commercialized in the 1960s and was short-listed for the prestigious \emph{Spiel des Jahres} award in 1979.
Notwithstanding \hex, Christian Freeling's \Hav is the connection game that has attracted the most interest from researchers and programmers.
This may be attributed to Freeling challenging programmers to build an AI able to beat him in at least one game in a ten-game match before 2012.
Of course, Christian Freeling could not have foreseen the advent of the Monte Carlo Tree Search algorithm and the prize money of the \emph{Havannah Challenge 2012} was awarded after a man-machine match in October 2012.\footnote{See the press release at \url{http://mindsports.nl/index.php/arena/havannah/641}.}
\Sl is a more recent addition to the connection game family as it was designed by Corey Clark in 2010.
Nevertheless, the introduction of forbidden patterns mechanics in this game has already proved quite influential in connection game design.\footnote{See for example the discussion at \url{https://www.boardgamegeek.com/thread/1081423/new-connection-games-slither-restriction}.}
The two-player game \gamename{slither} that we study in this article should not be confused with the Japanese single-player puzzle \gamename{slither link} which has been the object of independent papers~\cite{Yato2000,YoshinakaSKTIM2012}.

The first two games were the main topic of multiple master's theses and research articles~\cite{Huber1983,MazzoniW1997,Moesker2009,UiterwijkM2009,TeytaudT2009,Lorentz2011,Ewalds2012} and all three gave rise to competitive play.
High-level online competitive play takes place on \url{www.littlegolem.net}.
Finally, live competitive play can also be observed between human players at the Mind Sports Olympiads where an international \gamename{twixt} championship has taken place every year since 1997, as well as between \gamename{havannah} computer players at the ICGA Computer Olympiad since 2009.\footnote{See \url{www.boardability.com/game.php?id=twixt} and \url{www.grappa.univ-lille3.fr/icga/game.php?id=37} for details.}

Our main contributions in this paper are diverse.\footnote{This article is based on two conference papers but improves upon them and extends our previous work significantly~\cite{BonnetJamainSaffidine2013CG,BonnetJS2015}.
  The \pspace-completeness proof for \Hav has been simplified greatly.
  The parameterized complexity results are new.
  The discussion on ultra-weak solutions has been detailed and extended to \gamename{twixt} and \Hav.}
\begin{enumerate}
\item We establish that \Hav, \gamename{twixt}, and \gamename{slither} are \pspace-complete.
\item We study the parameterized complexity of \hex with the length of a solution as parameter and we show that while \sGHex is \wone-hard, \SHex is \fpt.
\item We provide a formal proof that planar \slither does not admit draws, as claimed by its designer, but that it heavily depends on the topology of the board.
\item We demonstrate that John Nash's ultra-weak solution approach for \hex can only be adapted to \gamename{twixt}, \Hav, and \gamename{slither} to a limited extent.
\end{enumerate}

To the best of our knowledge, our result on \SHex constitute the first tractability result in the field of Connection Games since \citeauthor{BrunoW1970}'s result on the Shannon Switching Game in \citeyear{BrunoW1970}~\cite{BrunoW1970}.
The other games we focus on each add to or change important mechanics of \hex, the quintessential connection game.
Indeed, from a connection perspective, \gamename{twixt} is played on a non-planar graph\footnote{Indeed, planar graphs have average degree smaller than $6$, whereas a large board of \gamename{twixt} has average degree close to $8$.}; \Hav adds non-edge related winning conditions; and \gamename{slither} allows moving previously placed pieces.
Together, these contributions make for a better understanding of the design decisions in connection games and their theoretical implications.



The paper is organized as follows.
After describing the rules of \GG and \hex, we prove that \gamename{twixt} and \slither are \pspace-complete in Section~\ref{sec:twixt-slither} and that \Hav is \pspace-complete in Section~\ref{sec:complexity-havannah}.
We then address the parameterized complexity of \GHex and \hex (Section~\ref{sec:parameterized-complexity-hex}).
After having extensively considered the computational complexity of solving arbitrary positions of these games, Section~\ref{sec:ultra-weak} moves on to the problem of solving their starting (empty) positions.
A discussion and final remarks conclude the paper.

\section{Previous Work}
A staple in proofs of \pspace-hardness for two-player games, \gamename{generalized geography} (\GG{}) is one of the first two-player games to have been proved \pspace-complete~\cite{Schaefer1978}.
This problem has been used to show the intractability of games as different as \gamename{hex}~\cite{Reisch1981}, \gamename{othello}~\cite{IwataK1994}, \gamename{amazons}~\cite{FurtakKUB2005}, \gamename{bridge}~\cite{BonnetJamainSaffidine2013IJCAI}, and many more~\cite{HearnDemaine2009}.
Our \pspace-hardness proof for \Hav is a reduction from \GG and the results for \gamename{twixt} and \slither also rely on \gamename{generalized geography}, albeit indirectly since we reduce from \gamename{hex}.

\GG is a game played on a graph.
In this section, we recall the rules of \GG and some of the assumptions that can be made on the input graph while preserving \pspace-completeness.
We also present the rules of \hex.

\subsection{Generalized Geography}
\label{sec:gg}

In \GG{}, players take turns moving a token from vertex to vertex in a given graph.
If the token is on a vertex $v$, then it can be moved to a vertex $v'$ neighboring $v$ provided $v'$ has not been visited yet.
A player wins when it is their opponent's turn and the opponent has no legal moves.
An instance of \GG{} is a graph $G$ and an initial vertex $v_0$, and asks whether the first player has a winning strategy in the corresponding game.

We denote by $P(v)$ the set of predecessors of the vertex $v$ in $G$, and $S(v)$ the set of successors of $v$.
A vertex with in-degree $i$ and out-degree $o$ is called $(i,o)$-vertex.
The degree of a vertex is the sum of the in-degree and the out-degree, and the degree of $G$ is the maximal degree among all vertices of $G$.
If $V$ is the set of vertices of $G$ and $V'$ is a subset of vertices, then $G[V\setminus V']$ is the induced subgraph of $G$ where vertices belonging to $V'$ have been removed.

Lichtenstein and Sipser have proved that the game remained \pspace-hard even if $G$ was assumed to be bipartite, planar, and of degree at most 3~\cite{LichtensteinS1980}.
We will reduce from such a restriction of \GG{} to show that \Hav{} is \pspace-hard.
To limit the number of gadgets\footnote{In designing a reduction from $\Pi_A$ to $\Pi_B$, a gadget is a portion of a $\Pi_B$-instance that encodes a meaningful component of $\Pi_A$, be it a vertex, an edge, a variable, a clause etc.} we need to create, we will also assume a few simplifications detailed below.
Naturally, these simplifications do not impact \pspace-hardness.
An example of a simplified instance of \GG{} can be found in Figure~\ref{fig:gg-instance}.

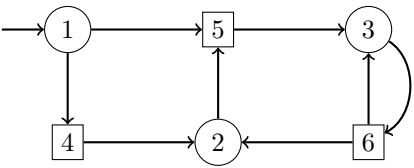
\begin{figure}
  \centering
\begin{tikzpicture}
  \tikzstyle{every node}=[draw,font=\normalsize]

  \node[draw=none] (v0) at (-1,0) {};
  \node[circle] (v1) at (0, 0  ) {1};
  \node[circle] (v2) at (2,-1.5) {2};
  \node[circle] (v3) at (4, 0  ) {3};

  \node[] (v4) at (0,-1.5) {4};
  \node[] (v5) at (2, 0  ) {5};
  \node[] (v6) at (4,-1.5) {6};

  \draw[thick,->] (v0) to (v1);

  \draw[thick,->] (v1) to (v4);
  \draw[thick,->] (v1) to (v5);

  \draw[thick,->] (v2) to (v5);

  \draw[thick,->] (v3) to [bend left=60] (v6);

  \draw[thick,->] (v4) to (v2);

  \draw[thick,->] (v5) to (v3);

  \draw[thick,->] (v6) to (v2);
  \draw[thick,->] (v6) to [bend left=00] (v3);
\end{tikzpicture}
  \caption{Example of an instance of \GG{} with vertex $1$ as initial vertex.}
  \label{fig:gg-instance}
\end{figure}

Let $(G,v_0)$ be an instance of \GG{} with $G$ bipartite, planar, and of degree at most 3.
We can assume that there is no vertex $v$ with out-degree 0 in $G$.
Indeed, if $v_0 \in P(v)$ then $(G,v_0)$ is trivially winning for Player 1.
Else, $(G[V \setminus (\{v\} \cup P(v))],v_0)$ is an equivalent instance, since playing in a predecessor of $v$ is losing.

All edges coming to the initial vertex $v_0$ can be removed to form an equivalent instance.
So, $v_0$ is a $(0,1)$-, a $(0,2)$-, or a $(0,3)$-vertex.
If $S(v_0)=\{v'\}$, then $(G[V \setminus \{v_0\}],v')$ is a strictly smaller instance such that Player 1 is winning in $(G,v_0)$ if and only if Player 1 is losing in $(G[V \setminus \{v_0\}],v')$.
If $S(v_0)=\{v',v'',v'''\}$, then Player 1 is winning in $(G,v_0)$ if and only if Player 1 is losing in at least one of the three instances $(G[V \setminus \{v_0\}],v')$, $(G[V \setminus \{v_0\}],v'')$, and $(G[V \setminus \{v_0\}],v''')$.
In those three instances $v'$, $v''$, and $v''$ are not $(0,3)$-vertices since they had in-degree at least 1 in $G$.
Thus, solving instances with an initial $(0,2)$-vertex is as hard as solving instances with an initial $(0,3)$-vertex.
We can therefore also assume that $v_0$ is a $(0,2)$-vertex.

Finally, we may assume that there exists a bipartition $G_1, G_2$ of the graph such that both $G_1$ and $G_2$ contain a $(1, 2)$ choice vertex.
Indeed, such vertices can always be added to a new component of the graph disconnected from the initial vertex without changing the outcome of the instance.

We call a \GG{} instance with an initial $(0,2)$-vertex, at least one $(1, 2)$-vertex per partition and then only $(1,1)$-, $(1,2)$-, and $(2,1)$-vertices a \emph{simplified} instance.

\subsection{Hex}
In \hex, two players alternate placing a stone of their color (black or white) in an unoccupied cell of a parallelogram board paved by hexagons.
Stones cannot be taken or moved, so the length of any game is bounded by the number of cells.
To win, one player must connect together a specified pair of opposite sides of the parallelogram.
The other player must connect the other pair of opposite sides.
As an example, Figure~\ref{fig:hex-puzzle} reproduces a \hex puzzle from the graph theorist Claude Berge, as cited by \citet{HaywardR2006}.

\newcommand{\bergepuzzle}{
  \havcoordinate{0}{18}{0}{18}
  \foreach \j in {3,...,3} { \node[hav-empty] at (0--\j)  {} ;}
  \foreach \j in {3,...,4} { \node[hav-empty] at  (1-\j)  {} ;}
  \foreach \j in {2,...,4} { \node[hav-empty] at (1--\j)  {} ;}
  \foreach \j in {2,...,5} { \node[hav-empty] at  (2-\j)  {} ;}
  \foreach \j in {1,...,5} { \node[hav-empty] at (2--\j)  {} ;}
  \foreach \j in {2,...,5} { \node[hav-empty] at  (3-\j)  {} ;}
  \foreach \j in {2,...,4} { \node[hav-empty] at (3--\j)  {} ;}
  \foreach \j in {3,...,4} { \node[hav-empty] at  (4-\j)  {} ;}
  \foreach \j in {3,...,3} { \node[hav-empty] at (4--\j)  {} ;}
  \node[hav-white] at (0--1) {}; \node[hav-black] at (0--5) {}; \node[hav-black] at (4--1) {}; \node[hav-white] at (4--5) {};
  \node[hav-white] at (0--3) {};
  \node[hav-black] at (1--3) {};
  \node[hav-black] at (2--1) {};
  \node[hav-black] at (2--3) {};
  \node[hav-white] at (2--4) {};
  \node[hav-white] at (3-2)  {};
  \node[hav-white] at (3-3)  {};
  \node[hav-white] at (3-4)  {};
  \node[hav-black] at (3-5)  {};
  \node[hav-black] at (3--4) {};
}
\begin{figure}
  \centering
  \newcommand{\scala}{0.5}
  \subfloat[Starting position of the puzzle.]{
    \begin{tikzpicture}[>=stealth',scale=\scala,every node/.style={scale=2*\scala}]
      \bergepuzzle
    \end{tikzpicture}
  }\hfill
  \subfloat[A failed attempt by White.]{
    \begin{tikzpicture}[>=stealth',scale=\scala,every node/.style={scale=2*\scala}]
      \bergepuzzle
      \node[hav-white] at (2--5) {1};
      \node[hav-black] at (2--2) {2};
      \node[hav-white] at (1--4) {3};
      \node[hav-black] at (1-4)  {4};
      \node[hav-white] at (2-4)  {5};
      \node[hav-black] at (2-3)  {6};
    \end{tikzpicture}
  }
  \caption{A \hex puzzle by Berge. White to play and win.
  White is trying to connect the bottom left edge to the top right edge and Black is trying to connect the top left edge to the bottom right edge.}
  \label{fig:hex-puzzle}
\end{figure}
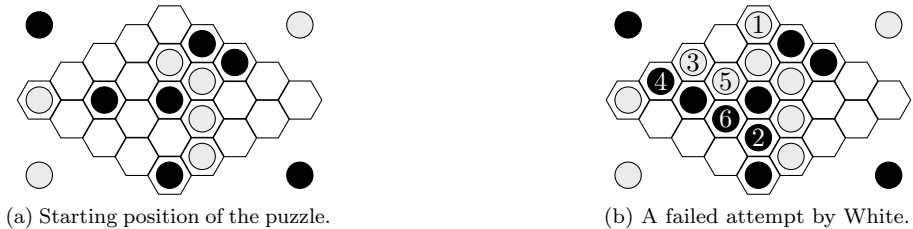

In Section~\ref{subsec:ultra-weak-hex}, we will see that a game cannot be drawn.
We now formalize what it means to \emph{connect} two (opposite) sides of the board.
Two cells are said \emph{neighbors} if they share an edge.
In a parallelogram-shaped hexagonal paving, cells have $6$ neighbors if they are not on an edge of the parallelogram.
On an edge, they may have $2$, $3$, or $4$ neighbors.
A black (resp.~white) group is a maximal connected component in the graph whose vertices are black (resp.~white) stones, and edges represent the \emph{neighbor} relation.
Then, connecting two sides mean having a group with at least one stone on each side.

\section{Complexity of Twixt and Slither}
\label{sec:twixt-slither}
\subsection{Twixt}
Alex Randolph's \gamename{twixt} is one of the most popular connection games.
It was invented around 1960 and was marketed as soon as in 1962~\cite{Huber1983}.
In his book devoted to connection games, Cameron Browne describes \gamename{twixt} as one of the most popular and widely marketed of all connection games~\cite{Browne2005}.
We now briefly describe the rules of \gamename{twixt} and refer to Moesker's master's thesis for an introduction and a mathematical approach to the strategy, and the description of a possible implementation~\cite{Moesker2009}.

\gamename{twixt} is a 2-player connection game played on a \gamename{go}-like board.
At their turn, player White and Black place a pawn of their color in an unoccupied place.
Just as in \gamename{havannah} and \gamename{hex}, pawns cannot be taken, moved, nor removed.
When 2 pawns of the same color are spaced by a knight's move, they are linked by an edge of their color, unless this edge would cross another edge.
At each turn, a player can remove some of their edges to allow for new links.
The goal for player White (resp.~Black) is to link top and bottom (resp.~left and right) sides of the board.
Note that sometimes, a player could have to choose between two possible edges that intersect each other.
The \textit{pencil and paper} version \gamename{twixtpp} where the edges of a same color are allowed to cross is also famous and played online.

As an illustration of the game rules, we reproduce here one of the original \gamename{twixt} puzzles invented by Alex Randolph.
A complete list of puzzles by Randolph supplemented by new puzzles by Alan Hensel can be found on \url{http://www.ibiblio.org/twixtpuzzles/}.
Figure~\ref{fig:twixt-puzzle18} displays puzzle 18 in which White is to play and win.
Observe that only White (resp.~Black) can play beyond the two horizontal white (resp.~vertical black) strips.
The first natural observation, in this puzzle, is that White has three groups of stones, two of which are virtually connected to the top side (White move K1, connecting both groups simultaneously, needs two Black moves to be partially parried), and one is connected to the bottom side (via B13 or F13).
A tentative approach could be to play F5, bringing the top left group closer to the bottom group.
Unfortunately that attempt fails when Black answers F6, shutting the top group completely, and the F5 move is wasted.
The solution to the puzzle is, instead, for White to start with move G7, which is connected to the bottom group.
White threatens both to connect the bottom group to the top left one via F5 and to connect the bottom group to the top right group via I6.
There is no way for Black to prevent both threats at once and White wins the game.

\begin{figure}
  \centering
\begin{tikzpicture}[scale=0.43]
  \begin{scope}[xshift=-0cm,yshift=-0cm,yscale=-1]
    \twboard{13}{13}

    \foreach \j in {1,...,13} {
      \node at (0, \j) {\j} ;
    }
    \node at ( 1,0) {A}; \node at ( 2,0) {B}; \node at ( 3,0) {C}; \node at ( 4,0) {D}; \node at ( 5,0) {E};
    \node at ( 6,0) {F}; \node at ( 7,0) {G}; \node at ( 8,0) {H}; \node at ( 9,0) {I}; \node at (10,0) {J};
    \node at (11,0) {K}; \node at (12,0) {L}; \node at (13,0) {M};

    \draw (4-12) -- (5-10) -- (6-8) -- (5-6);
    \node at (4-12)[tw-white] {}; \node at (5-10)[tw-white] {}; \node at  (6-8)[tw-white] {}; \node at  (5-6)[tw-white] {};

    \draw (7-3) -- (9-2);
    \node at  (7-3)[tw-white] {}; \node at  (9-2)[tw-white] {};

    \draw (11-5) -- (12-3);
    \node at (11-5)[tw-white] {}; \node at (12-3)[tw-white] {};

    \draw (4-5) -- (6-4);
    \node at  (4-5)[tw-black] {}; \node at  (6-4)[tw-black] {};

    \draw (8-6) -- (7-4) -- (9-3);
    \node at  (7-4)[tw-black] {}; \node at  (8-6)[tw-black] {}; \node at  (9-3)[tw-black] {};

    \draw (8-9) -- (10-8) -- (11-6) -- (12-4);
    \node at  (8-9)[tw-black] {}; \node at (10-8)[tw-black] {}; \node at (11-6)[tw-black] {}; \node at (12-4)[tw-black] {};
  \end{scope}
\end{tikzpicture}
  \caption{\Twixt puzzle 18 by Alex Randolph.
  White to play and win.}
  \label{fig:twixt-puzzle18}
\end{figure}
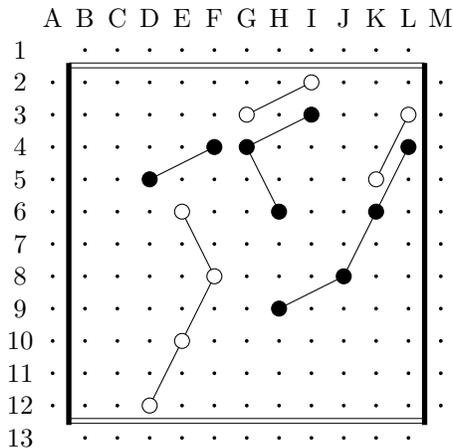

As the length of a game of \twixt is polynomially bounded, exploring the whole tree can be done with polynomial space using a minimax algorithm.
Therefore \gamename{twixt} is in \pspace.

Mazzoni and Watkins have shown that \textsc{3-sat} could be reduced to single-player \twixt, thus showing \np-completeness of the variant~\cite{MazzoniW1997}.
While it might be possible to try and adapt their work and obtain a reduction from \textsc{3-qbf} to standard two-player \twixt, we propose a simpler approach based on \hex.
The \pspace-completeness of \hex \cite{Reisch1981} has already been used to show the \pspace-completeness of \gamename{amazons}, a well-known territory game~\cite{FurtakKUB2005}.

We now present how we construct from an instance $G$ of \hex an instance $\phi(G)$ of \twixt.
We can represent a cell of \hex by the $9\times 9$ \twixt gadgets displayed in Figure~\ref{fig:twixt-cell}.
Let $n$ be the size of a side of $G$, Figure~\ref{fig:twixt-board} shows how a \twixt board can be paved by $n^2$ \twixt cell gadgets to create a \hex board.

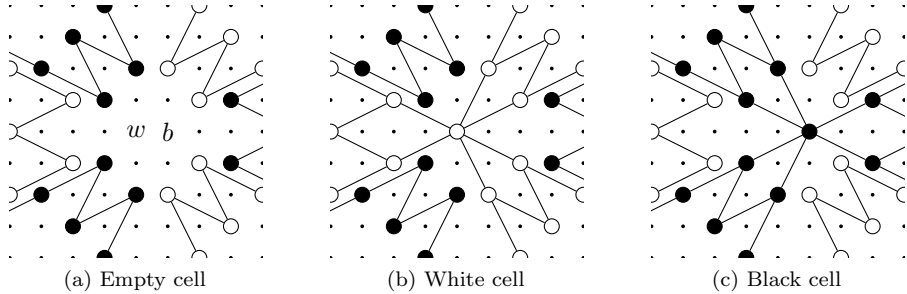
\begin{figure}
  \centering
  \subfloat[Empty cell]{
    \centering
    \begin{tikzpicture}[scale=0.42,rotate=-90]
      \clip (1, 1) rectangle (9, 9);
\foreach[count=\ci] \i in \alphabet {
  \foreach \j in {0,...,10} {
    \node (\i\j) [tw-empty] at ({\ci},{\j}) {} ;
  }
}
\foreach \i in {B0, A4, C2, C5, C10, D4, D8, F4, F8, G2, G5, G10, H0, I4} {
  \node (\i) at (\i)[tw-black] {};
}
\foreach \i in {A7, C1, C6, C9, D3, D7, F3, F7, G1, G6, G9, I7} {
  \node (\i) at (\i)[tw-white] {};
}
\node (B3) at (B3)[tw-black] {}; \node (H3) at (H3)[tw-black] {}; \node (E10) at (E10)[tw-black] {};
\node (B8) at (B8)[tw-white] {}; \node (H8) at (H8)[tw-white] {}; \node (E1) at (E1)[tw-white] {};
\draw (C5) -- (B3) -- (D4); \draw (F4) -- (H3) -- (G5); \draw (D8) -- (E10) -- (F8); \draw (C6) -- (B8) -- (D7); \draw (F7) -- (H8) -- (G6); \draw (D3) -- (E1) -- (F3);

\draw (A4) -- (C5);
\draw (A7) -- (C6);
\draw (B0) -- (C2);
\draw (H0) -- (G2);
\draw (C1) -- (D3);
\draw (C2) -- (D4);
\draw (C9) -- (D7);
\draw (C10) -- (D8);
\draw (F3) -- (G1);
\draw (F4) -- (G2);
\draw (F7) -- (G9);
\draw (F8) -- (G10);
\draw (G5) -- (I4);
\draw (G6) -- (I7);
      \node (E6) at (E6)[fill=white] {$b$};
      \node (E5) at (E5)[fill=white] {$w$};
    \end{tikzpicture}
    \label{fig:tw-empty}
  }
  \hfill
  \subfloat[White cell]{
    \centering
    \begin{tikzpicture}[scale=0.42,rotate=-90]
      \clip (1, 1) rectangle (9, 9);
\foreach[count=\ci] \i in \alphabet {
  \foreach \j in {0,...,10} {
    \node (\i\j) [tw-empty] at ({\ci},{\j}) {} ;
  }
}
\foreach \i in {B0, A4, C2, C5, C10, D4, D8, F4, F8, G2, G5, G10, H0, I4} {
  \node (\i) at (\i)[tw-black] {};
}
\foreach \i in {A7, C1, C6, C9, D3, D7, F3, F7, G1, G6, G9, I7} {
  \node (\i) at (\i)[tw-white] {};
}
\node (B3) at (B3)[tw-black] {}; \node (H3) at (H3)[tw-black] {}; \node (E10) at (E10)[tw-black] {};
\node (B8) at (B8)[tw-white] {}; \node (H8) at (H8)[tw-white] {}; \node (E1) at (E1)[tw-white] {};
\draw (C5) -- (B3) -- (D4); \draw (F4) -- (H3) -- (G5); \draw (D8) -- (E10) -- (F8); \draw (C6) -- (B8) -- (D7); \draw (F7) -- (H8) -- (G6); \draw (D3) -- (E1) -- (F3);

\draw (A4) -- (C5);
\draw (A7) -- (C6);
\draw (B0) -- (C2);
\draw (H0) -- (G2);
\draw (C1) -- (D3);
\draw (C2) -- (D4);
\draw (C9) -- (D7);
\draw (C10) -- (D8);
\draw (F3) -- (G1);
\draw (F4) -- (G2);
\draw (F7) -- (G9);
\draw (F8) -- (G10);
\draw (G5) -- (I4);
\draw (G6) -- (I7);
\draw (E5) -- (C6);
\draw (E5) -- (D3);
\draw (E5) -- (D7);
\draw (E5) -- (F3);
\draw (E5) -- (F7);
\draw (E5) -- (G6);

\node at (E5)[tw-white] {};
    \end{tikzpicture}
    \label{fig:tw-white}
  }
  \hfill
  \subfloat[Black cell]{
    \centering
    \begin{tikzpicture}[scale=0.42,rotate=-90]
      \clip (1, 1) rectangle (9, 9);
\foreach[count=\ci] \i in \alphabet {
  \foreach \j in {0,...,10} {
    \node (\i\j) [tw-empty] at ({\ci},{\j}) {} ;
  }
}
\foreach \i in {B0, A4, C2, C5, C10, D4, D8, F4, F8, G2, G5, G10, H0, I4} {
  \node (\i) at (\i)[tw-black] {};
}
\foreach \i in {A7, C1, C6, C9, D3, D7, F3, F7, G1, G6, G9, I7} {
  \node (\i) at (\i)[tw-white] {};
}
\node (B3) at (B3)[tw-black] {}; \node (H3) at (H3)[tw-black] {}; \node (E10) at (E10)[tw-black] {};
\node (B8) at (B8)[tw-white] {}; \node (H8) at (H8)[tw-white] {}; \node (E1) at (E1)[tw-white] {};
\draw (C5) -- (B3) -- (D4); \draw (F4) -- (H3) -- (G5); \draw (D8) -- (E10) -- (F8); \draw (C6) -- (B8) -- (D7); \draw (F7) -- (H8) -- (G6); \draw (D3) -- (E1) -- (F3);

\draw (A4) -- (C5);
\draw (A7) -- (C6);
\draw (B0) -- (C2);
\draw (H0) -- (G2);
\draw (C1) -- (D3);
\draw (C2) -- (D4);
\draw (C9) -- (D7);
\draw (C10) -- (D8);
\draw (F3) -- (G1);
\draw (F4) -- (G2);
\draw (F7) -- (G9);
\draw (F8) -- (G10);
\draw (G5) -- (I4);
\draw (G6) -- (I7);
\draw (E6) -- (C5);
\draw (E6) -- (D4);
\draw (E6) -- (D8);
\draw (E6) -- (F4);
\draw (E6) -- (F8);
\draw (E6) -- (G5);

\node at (E6)[tw-black] {};
    \end{tikzpicture}
    \label{fig:tw-black}
  }
  \caption{Basic gadgets representing \hex cells in \twixt.}
  \label{fig:twixt-cell}
\end{figure}

It is not hard to see from Figure~\ref{fig:tw-empty} that in each gadget of Figure~\ref{fig:twixt-board}, move $w$ (resp.~$b$) is dominating for White (resp.~Black).
That is, playing $w$ is as good for White as any other move of the gadget.
We can also see that the moves that are not part of any gadget in Figure~\ref{fig:twixt-board} are dominated for both players.
As a result, if player Black (resp.~White) has a winning strategy in $G$, then player Black has a winning strategy in $\phi(G)$.
Thus, $G$ is won by Black if and only if $\phi(G)$ is won by Black.
Therefore determining the winner in \twixt is at least as hard as in \hex, leading to the desired result.

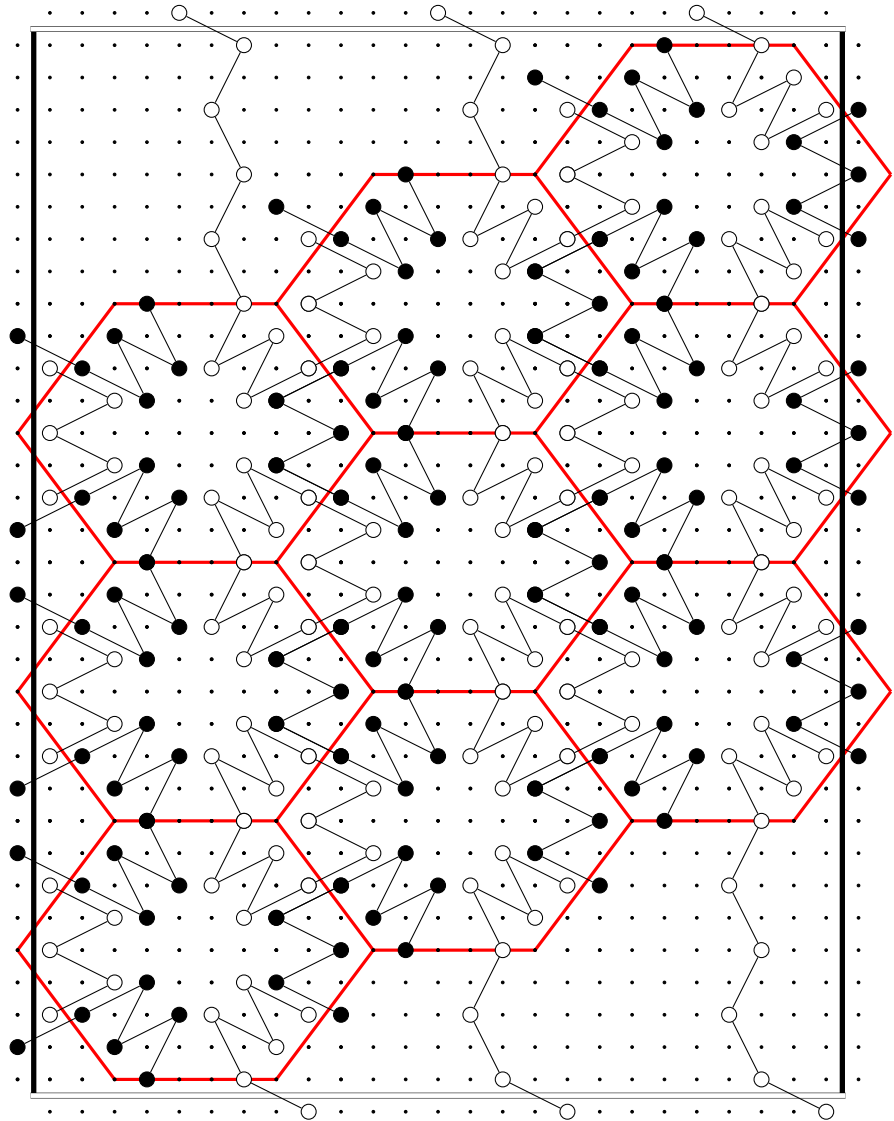
\begin{figure}
  \centering
\begin{tikzpicture}[scale=0.43]
  \begin{scope}[scale=1.0,xshift=30mm,yshift= 10mm]
    \draw[hexsid] ( 0,  0) -- ( 5,  0); \draw[hexsid] ( 0,  8) -- ( 5,  8); \draw[hexsid] ( 0, 16) -- ( 5, 16); \draw[hexsid] ( 0, 24) -- ( 5, 24);
    \draw[hexsid] ( 8,  4) -- (13,  4); \draw[hexsid] ( 8, 12) -- (13, 12); \draw[hexsid] ( 8, 20) -- (13, 20); \draw[hexsid] ( 8, 28) -- (13, 28);
    \draw[hexsid] (16,  8) -- (21,  8); \draw[hexsid] (16, 16) -- (21, 16); \draw[hexsid] (16, 24) -- (21, 24); \draw[hexsid] (16, 32) -- (21, 32);

    \draw[hexsid] (-3,  4) -- ( 0,  8); \draw[hexsid] (-3, 12) -- ( 0, 16); \draw[hexsid] (-3, 20) -- ( 0, 24);
    \draw[hexsid] ( 5,  0) -- ( 8,  4); \draw[hexsid] ( 5,  8) -- ( 8, 12); \draw[hexsid] ( 5, 16) -- ( 8, 20); \draw[hexsid] ( 5, 24) -- ( 8, 28);
    \draw[hexsid] (13,  4) -- (16,  8); \draw[hexsid] (13, 12) -- (16, 16); \draw[hexsid] (13, 20) -- (16, 24); \draw[hexsid] (13, 28) -- (16, 32);
    \draw[hexsid] (21,  8) -- (24, 12); \draw[hexsid] (21, 16) -- (24, 20); \draw[hexsid] (21, 24) -- (24, 28);

    \draw[hexsid] (-3,  4) -- ( 0,  0); \draw[hexsid] (-3, 12) -- ( 0,  8); \draw[hexsid] (-3, 20) -- ( 0, 16);
    \draw[hexsid] ( 5,  8) -- ( 8,  4); \draw[hexsid] ( 5, 16) -- ( 8, 12); \draw[hexsid] ( 5, 24) -- ( 8, 20);
    \draw[hexsid] (13, 12) -- (16,  8); \draw[hexsid] (13, 20) -- (16, 16); \draw[hexsid] (13, 28) -- (16, 24);
    \draw[hexsid] (21, 16) -- (24, 12); \draw[hexsid] (21, 24) -- (24, 20); \draw[hexsid] (21, 32) -- (24, 28);
  \end{scope}
  \begin{scope}[rotate=-90,xscale=-1]
    \begin{scope}[xshift=-1cm,yshift=-1cm]
      \twboardr{35}{27}

      \node (1-10) at (1-10)[tw-white] {};
      \draw (1-10) -- (2-8);

      \node (1-18) at (1-18)[tw-white] {}; \node (2-16) at (2-16)[tw-white] {}; \node (4-15) at (4-15)[tw-white] {};
      \draw (1-18) -- (2-16) -- (4-15) -- (6-16);

      \node (1-26) at (1-26)[tw-white] {}; \node (2-24) at (2-24)[tw-white] {}; \node (4-23) at (4-23)[tw-white] {}; \node (6-24) at (6-24)[tw-white] {}; \node (8-23) at (8-23)[tw-white] {};
      \draw (1-26) -- (2-24) -- (4-23) -- (6-24) -- (8-23) -- (10-24);

      \node (35-22) at (35-22)[tw-white] {};
      \draw (35-22) -- (34-24);

      \node (32-15) at (32-15)[tw-white] {}; \node (34-16) at (34-16)[tw-white] {}; \node (35-14) at (35-14)[tw-white] {};
      \draw (35-14) -- (34-16) -- (32-15) -- (30-16);

      \node (28-7) at (28-7)[tw-white] {}; \node (30-8) at (30-8)[tw-white] {}; \node (32-7) at (32-7)[tw-white] {}; \node (34-8) at (34-8)[tw-white] {}; \node (35-6) at (35-6)[tw-white] {};
      \draw (35-6) -- (34-8) -- (32-7) -- (30-8) -- (28-7) -- (26-8);

    \end{scope}

    \begin{scope}[xshift= 8cm,yshift=16cm] 
\foreach[count=\ci] \i in \alphabet {
  \foreach \j in {0,...,10} {
    \node (\i\j) [tw-empty] at ({\ci},{\j}) {} ;
  }
}
\foreach \i in {B0, A4, C2, C5, C10, D4, D8, F4, F8, G2, G5, G10, H0, I4} {
  \node (\i) at (\i)[tw-black] {};
}
\foreach \i in {A7, C1, C6, C9, D3, D7, F3, F7, G1, G6, G9, I7} {
  \node (\i) at (\i)[tw-white] {};
}
\node (B3) at (B3)[tw-black] {}; \node (H3) at (H3)[tw-black] {}; \node (E10) at (E10)[tw-black] {};
\node (B8) at (B8)[tw-white] {}; \node (H8) at (H8)[tw-white] {}; \node (E1) at (E1)[tw-white] {};
\draw (C5) -- (B3) -- (D4); \draw (F4) -- (H3) -- (G5); \draw (D8) -- (E10) -- (F8); \draw (C6) -- (B8) -- (D7); \draw (F7) -- (H8) -- (G6); \draw (D3) -- (E1) -- (F3);

\draw (A4) -- (C5);
\draw (A7) -- (C6);
\draw (B0) -- (C2);
\draw (H0) -- (G2);
\draw (C1) -- (D3);
\draw (C2) -- (D4);
\draw (C9) -- (D7);
\draw (C10) -- (D8);
\draw (F3) -- (G1);
\draw (F4) -- (G2);
\draw (F7) -- (G9);
\draw (F8) -- (G10);
\draw (G5) -- (I4);
\draw (G6) -- (I7);
 \end{scope}
    \begin{scope}[xshift=16cm,yshift=16cm] 
\foreach[count=\ci] \i in \alphabet {
  \foreach \j in {0,...,10} {
    \node (\i\j) [tw-empty] at ({\ci},{\j}) {} ;
  }
}
\foreach \i in {B0, A4, C2, C5, C10, D4, D8, F4, F8, G2, G5, G10, H0, I4} {
  \node (\i) at (\i)[tw-black] {};
}
\foreach \i in {A7, C1, C6, C9, D3, D7, F3, F7, G1, G6, G9, I7} {
  \node (\i) at (\i)[tw-white] {};
}
\node (B3) at (B3)[tw-black] {}; \node (H3) at (H3)[tw-black] {}; \node (E10) at (E10)[tw-black] {};
\node (B8) at (B8)[tw-white] {}; \node (H8) at (H8)[tw-white] {}; \node (E1) at (E1)[tw-white] {};
\draw (C5) -- (B3) -- (D4); \draw (F4) -- (H3) -- (G5); \draw (D8) -- (E10) -- (F8); \draw (C6) -- (B8) -- (D7); \draw (F7) -- (H8) -- (G6); \draw (D3) -- (E1) -- (F3);

\draw (A4) -- (C5);
\draw (A7) -- (C6);
\draw (B0) -- (C2);
\draw (H0) -- (G2);
\draw (C1) -- (D3);
\draw (C2) -- (D4);
\draw (C9) -- (D7);
\draw (C10) -- (D8);
\draw (F3) -- (G1);
\draw (F4) -- (G2);
\draw (F7) -- (G9);
\draw (F8) -- (G10);
\draw (G5) -- (I4);
\draw (G6) -- (I7);
 \end{scope}
    \begin{scope}[xshift=24cm,yshift=16cm] 
\foreach[count=\ci] \i in \alphabet {
  \foreach \j in {0,...,10} {
    \node (\i\j) [tw-empty] at ({\ci},{\j}) {} ;
  }
}
\foreach \i in {B0, A4, C2, C5, C10, D4, D8, F4, F8, G2, G5, G10, H0, I4} {
  \node (\i) at (\i)[tw-black] {};
}
\foreach \i in {A7, C1, C6, C9, D3, D7, F3, F7, G1, G6, G9, I7} {
  \node (\i) at (\i)[tw-white] {};
}
\node (B3) at (B3)[tw-black] {}; \node (H3) at (H3)[tw-black] {}; \node (E10) at (E10)[tw-black] {};
\node (B8) at (B8)[tw-white] {}; \node (H8) at (H8)[tw-white] {}; \node (E1) at (E1)[tw-white] {};
\draw (C5) -- (B3) -- (D4); \draw (F4) -- (H3) -- (G5); \draw (D8) -- (E10) -- (F8); \draw (C6) -- (B8) -- (D7); \draw (F7) -- (H8) -- (G6); \draw (D3) -- (E1) -- (F3);

\draw (A4) -- (C5);
\draw (A7) -- (C6);
\draw (B0) -- (C2);
\draw (H0) -- (G2);
\draw (C1) -- (D3);
\draw (C2) -- (D4);
\draw (C9) -- (D7);
\draw (C10) -- (D8);
\draw (F3) -- (G1);
\draw (F4) -- (G2);
\draw (F7) -- (G9);
\draw (F8) -- (G10);
\draw (G5) -- (I4);
\draw (G6) -- (I7);
 \end{scope}
    \begin{scope}[xshift= 4cm,yshift= 8cm] 
\foreach[count=\ci] \i in \alphabet {
  \foreach \j in {0,...,10} {
    \node (\i\j) [tw-empty] at ({\ci},{\j}) {} ;
  }
}
\foreach \i in {B0, A4, C2, C5, C10, D4, D8, F4, F8, G2, G5, G10, H0, I4} {
  \node (\i) at (\i)[tw-black] {};
}
\foreach \i in {A7, C1, C6, C9, D3, D7, F3, F7, G1, G6, G9, I7} {
  \node (\i) at (\i)[tw-white] {};
}
\node (B3) at (B3)[tw-black] {}; \node (H3) at (H3)[tw-black] {}; \node (E10) at (E10)[tw-black] {};
\node (B8) at (B8)[tw-white] {}; \node (H8) at (H8)[tw-white] {}; \node (E1) at (E1)[tw-white] {};
\draw (C5) -- (B3) -- (D4); \draw (F4) -- (H3) -- (G5); \draw (D8) -- (E10) -- (F8); \draw (C6) -- (B8) -- (D7); \draw (F7) -- (H8) -- (G6); \draw (D3) -- (E1) -- (F3);

\draw (A4) -- (C5);
\draw (A7) -- (C6);
\draw (B0) -- (C2);
\draw (H0) -- (G2);
\draw (C1) -- (D3);
\draw (C2) -- (D4);
\draw (C9) -- (D7);
\draw (C10) -- (D8);
\draw (F3) -- (G1);
\draw (F4) -- (G2);
\draw (F7) -- (G9);
\draw (F8) -- (G10);
\draw (G5) -- (I4);
\draw (G6) -- (I7);
 \end{scope}
    \begin{scope}[xshift=12cm,yshift= 8cm] 
\foreach[count=\ci] \i in \alphabet {
  \foreach \j in {0,...,10} {
    \node (\i\j) [tw-empty] at ({\ci},{\j}) {} ;
  }
}
\foreach \i in {B0, A4, C2, C5, C10, D4, D8, F4, F8, G2, G5, G10, H0, I4} {
  \node (\i) at (\i)[tw-black] {};
}
\foreach \i in {A7, C1, C6, C9, D3, D7, F3, F7, G1, G6, G9, I7} {
  \node (\i) at (\i)[tw-white] {};
}
\node (B3) at (B3)[tw-black] {}; \node (H3) at (H3)[tw-black] {}; \node (E10) at (E10)[tw-black] {};
\node (B8) at (B8)[tw-white] {}; \node (H8) at (H8)[tw-white] {}; \node (E1) at (E1)[tw-white] {};
\draw (C5) -- (B3) -- (D4); \draw (F4) -- (H3) -- (G5); \draw (D8) -- (E10) -- (F8); \draw (C6) -- (B8) -- (D7); \draw (F7) -- (H8) -- (G6); \draw (D3) -- (E1) -- (F3);

\draw (A4) -- (C5);
\draw (A7) -- (C6);
\draw (B0) -- (C2);
\draw (H0) -- (G2);
\draw (C1) -- (D3);
\draw (C2) -- (D4);
\draw (C9) -- (D7);
\draw (C10) -- (D8);
\draw (F3) -- (G1);
\draw (F4) -- (G2);
\draw (F7) -- (G9);
\draw (F8) -- (G10);
\draw (G5) -- (I4);
\draw (G6) -- (I7);
 \end{scope}
    \begin{scope}[xshift=20cm,yshift= 8cm] 
\foreach[count=\ci] \i in \alphabet {
  \foreach \j in {0,...,10} {
    \node (\i\j) [tw-empty] at ({\ci},{\j}) {} ;
  }
}
\foreach \i in {B0, A4, C2, C5, C10, D4, D8, F4, F8, G2, G5, G10, H0, I4} {
  \node (\i) at (\i)[tw-black] {};
}
\foreach \i in {A7, C1, C6, C9, D3, D7, F3, F7, G1, G6, G9, I7} {
  \node (\i) at (\i)[tw-white] {};
}
\node (B3) at (B3)[tw-black] {}; \node (H3) at (H3)[tw-black] {}; \node (E10) at (E10)[tw-black] {};
\node (B8) at (B8)[tw-white] {}; \node (H8) at (H8)[tw-white] {}; \node (E1) at (E1)[tw-white] {};
\draw (C5) -- (B3) -- (D4); \draw (F4) -- (H3) -- (G5); \draw (D8) -- (E10) -- (F8); \draw (C6) -- (B8) -- (D7); \draw (F7) -- (H8) -- (G6); \draw (D3) -- (E1) -- (F3);

\draw (A4) -- (C5);
\draw (A7) -- (C6);
\draw (B0) -- (C2);
\draw (H0) -- (G2);
\draw (C1) -- (D3);
\draw (C2) -- (D4);
\draw (C9) -- (D7);
\draw (C10) -- (D8);
\draw (F3) -- (G1);
\draw (F4) -- (G2);
\draw (F7) -- (G9);
\draw (F8) -- (G10);
\draw (G5) -- (I4);
\draw (G6) -- (I7);
 \end{scope}
    \begin{scope}[xshift= 0cm,yshift= 0cm] 
\foreach[count=\ci] \i in \alphabet {
  \foreach \j in {0,...,10} {
    \node (\i\j) [tw-empty] at ({\ci},{\j}) {} ;
  }
}
\foreach \i in {B0, A4, C2, C5, C10, D4, D8, F4, F8, G2, G5, G10, H0, I4} {
  \node (\i) at (\i)[tw-black] {};
}
\foreach \i in {A7, C1, C6, C9, D3, D7, F3, F7, G1, G6, G9, I7} {
  \node (\i) at (\i)[tw-white] {};
}
\node (B3) at (B3)[tw-black] {}; \node (H3) at (H3)[tw-black] {}; \node (E10) at (E10)[tw-black] {};
\node (B8) at (B8)[tw-white] {}; \node (H8) at (H8)[tw-white] {}; \node (E1) at (E1)[tw-white] {};
\draw (C5) -- (B3) -- (D4); \draw (F4) -- (H3) -- (G5); \draw (D8) -- (E10) -- (F8); \draw (C6) -- (B8) -- (D7); \draw (F7) -- (H8) -- (G6); \draw (D3) -- (E1) -- (F3);

\draw (A4) -- (C5);
\draw (A7) -- (C6);
\draw (B0) -- (C2);
\draw (H0) -- (G2);
\draw (C1) -- (D3);
\draw (C2) -- (D4);
\draw (C9) -- (D7);
\draw (C10) -- (D8);
\draw (F3) -- (G1);
\draw (F4) -- (G2);
\draw (F7) -- (G9);
\draw (F8) -- (G10);
\draw (G5) -- (I4);
\draw (G6) -- (I7);
 \end{scope}
    \begin{scope}[xshift= 8cm,yshift= 0cm] 
\foreach[count=\ci] \i in \alphabet {
  \foreach \j in {0,...,10} {
    \node (\i\j) [tw-empty] at ({\ci},{\j}) {} ;
  }
}
\foreach \i in {B0, A4, C2, C5, C10, D4, D8, F4, F8, G2, G5, G10, H0, I4} {
  \node (\i) at (\i)[tw-black] {};
}
\foreach \i in {A7, C1, C6, C9, D3, D7, F3, F7, G1, G6, G9, I7} {
  \node (\i) at (\i)[tw-white] {};
}
\node (B3) at (B3)[tw-black] {}; \node (H3) at (H3)[tw-black] {}; \node (E10) at (E10)[tw-black] {};
\node (B8) at (B8)[tw-white] {}; \node (H8) at (H8)[tw-white] {}; \node (E1) at (E1)[tw-white] {};
\draw (C5) -- (B3) -- (D4); \draw (F4) -- (H3) -- (G5); \draw (D8) -- (E10) -- (F8); \draw (C6) -- (B8) -- (D7); \draw (F7) -- (H8) -- (G6); \draw (D3) -- (E1) -- (F3);

\draw (A4) -- (C5);
\draw (A7) -- (C6);
\draw (B0) -- (C2);
\draw (H0) -- (G2);
\draw (C1) -- (D3);
\draw (C2) -- (D4);
\draw (C9) -- (D7);
\draw (C10) -- (D8);
\draw (F3) -- (G1);
\draw (F4) -- (G2);
\draw (F7) -- (G9);
\draw (F8) -- (G10);
\draw (G5) -- (I4);
\draw (G6) -- (I7);
 \end{scope}
    \begin{scope}[xshift=16cm,yshift= 0cm] 
\foreach[count=\ci] \i in \alphabet {
  \foreach \j in {0,...,10} {
    \node (\i\j) [tw-empty] at ({\ci},{\j}) {} ;
  }
}
\foreach \i in {B0, A4, C2, C5, C10, D4, D8, F4, F8, G2, G5, G10, H0, I4} {
  \node (\i) at (\i)[tw-black] {};
}
\foreach \i in {A7, C1, C6, C9, D3, D7, F3, F7, G1, G6, G9, I7} {
  \node (\i) at (\i)[tw-white] {};
}
\node (B3) at (B3)[tw-black] {}; \node (H3) at (H3)[tw-black] {}; \node (E10) at (E10)[tw-black] {};
\node (B8) at (B8)[tw-white] {}; \node (H8) at (H8)[tw-white] {}; \node (E1) at (E1)[tw-white] {};
\draw (C5) -- (B3) -- (D4); \draw (F4) -- (H3) -- (G5); \draw (D8) -- (E10) -- (F8); \draw (C6) -- (B8) -- (D7); \draw (F7) -- (H8) -- (G6); \draw (D3) -- (E1) -- (F3);

\draw (A4) -- (C5);
\draw (A7) -- (C6);
\draw (B0) -- (C2);
\draw (H0) -- (G2);
\draw (C1) -- (D3);
\draw (C2) -- (D4);
\draw (C9) -- (D7);
\draw (C10) -- (D8);
\draw (F3) -- (G1);
\draw (F4) -- (G2);
\draw (F7) -- (G9);
\draw (F8) -- (G10);
\draw (G5) -- (I4);
\draw (G6) -- (I7);
 \end{scope}

    \node at ( 9,  9) [tw-white] {}; \node at (17,  9) [tw-white] {}; \node at (25,  9) [tw-white] {}; \node at (13, 17) [tw-white] {}; \node at (21, 17) [tw-white] {}; \node at (29, 17) [tw-white] {};
  \end{scope}
\end{tikzpicture}
  \caption{Empty $3\times 3$ \hex board reduced to a \twixt board.}
  \label{fig:twixt-board}
\end{figure}

\begin{theorem}
  \Twixt is \pspace{}-complete.
\end{theorem}

Observe that the proposed reduction holds both for the classic version of \twixt as well as for the \emph{pencil and paper} version \gamename{twixtpp}.
Indeed, the reduction does not require the losing player to remove any edge, so it also proves that \gamename{twixtpp} is \pspace-hard.

\subsection{Slither}
Invented in 2010 by Corey Clark, \gamename{slither} is relatively new connection game with an increasing popularity among online board game players.
Unlike \gamename{hex} and \gamename{havannah} which are played on a hexagonally-paved board, \gamename{slither} is played on a grid and each player is trying to connect a pair of opposite edges corresponding to their color by constructing connected groups of stones.
Whereas moves in most other connection games only involve putting down a new element on the board, moves in \gamename{slither} also allow relocating previously played stones.
Another important difference between usual connections games and \gamename{slither} is that some stone configurations are forbidden in the latter.
Namely, a player is not allowed to play a stone diagonally adjacent to a pre-existing stone of their color unless one of their already placed stones would be mutually adjacent.

\subsubsection{Rules}

\gamename{Slither} is a two-player game starting on an empty $n$ by $n$ grid (or board).
Let us call the players Black (or $B$) and White (or $W$).
Black and White alternate \emph{moves}.
Before stating what a move consists of, and what the winning conditions are, we introduce some useful definitions.

\begin{figure}
  \centering
  \begin{tikzpicture}
    \sboard{7}{7}
    \sedges{7}{7}
    \sposition{2-2,3-2,3-3,2-5,6-2,5-5,6-6}{5-2,5-3,6-3,2-6,3-5}
    \draw[opacity=0.7,very thick] (0.75,2.25) -- (1.75,3.25) ;
    \draw[opacity=0.7,very thick] (0.75,3.25) -- (1.75,2.25) ;
    \draw[opacity=0.7,very thick] (2.25,2.25) -- (3.25,3.25) ;
    \draw[opacity=0.7,very thick] (2.25,3.25) -- (3.25,2.25) ;
  \end{tikzpicture}
  \caption{Some examples of allowed and forbidden configurations.
    Forbidden configurations are crossed.}
  \label{fig:slither-diagonal}
\end{figure}
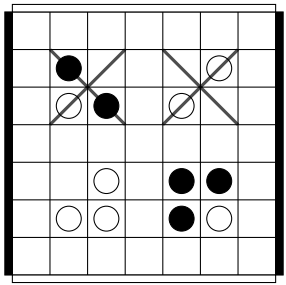

As the game proceeds, \emph{squares} of the board can be empty, or contain a black stone, or contain a white stone.
We refer to black (resp.~white) stones as the stones of player $B$ (resp.~$W$).
We say that two squares of the board are \emph{adjacent} if they are in the same row and adjacent columns, or in the same column and adjacent rows.
They are \emph{king-adjacent} if a chess king can move from one square to the other, and \emph{diagonally-adjacent} if they are king-adjacent but not adjacent.
Two stones are \emph{adjacent} (resp.~\emph{king-adjacent}, \emph{diagonally-adjacent}) if they are in adjacent (resp.~king-adjacent, diagonally-adjacent) squares.
For $P \in \{W,B\}$, let $G_P$ be the graph whose vertices are the stones of player $P$ placed in the board, and the edges encode the \emph{adjacent} relation.
That is, two vertices are linked by an edge if and only if they represent adjacent stones.
Like in the game of \gamename{go}, a \emph{group} for player $P$ is a maximal connected component in $G_P$.

A \emph{move} for player $P$ consists of an optional relocation of an existing stone of $P$ on a king-adjacent empty grid square, followed by a mandatory placement a stone of $P$ into an empty grid square (Figure~\ref{fig:slither-def-and-rules}).
For a move to be legal, the resulting position may not have two diagonally-adjacent stones of $P$ that do not also have an orthogonally-adjacent stone in common (Figure~\ref{fig:slither-diagonal}).
In what follows, we refer to this restrictive rule as the \emph{diagonal rule}.

Black wins if they form a group with at least one stone in the first and in the last column.
White wins if they form a group with at least one stone in the first and in the last row.
Informally, Black wants to connect left-right and White wants to connect top-bottom (Figure~\ref{fig:slither-winning}).

\begin{figure}[t]
  \centering
  \subfloat[Black to play and win in one move.]{
    \hspace{5mm}
    \label{fig:slither-legal}
    \begin{tikzpicture}
      \sboard{7}{7}
      \sedges{7}{7}
      \sposition{2-4,4-4,5-4,6-7,7-1,7-2,7-3,7-4,6-6,6-7,2-6,2-7,3-6,3-5,3-4}{2-5,1-3,2-3,2-2,2-1,3-1,4-1,4-2,4-3,5-2,6-2,6-3,6-4,5-6}
    \end{tikzpicture}
    \hspace{7mm}
  }
  \hfill
  \subfloat[The winning move for Black.]{
    \hspace{4mm}
    \label{fig:slither-winning}
    \begin{tikzpicture}[>=stealth']
      \sboard{7}{7}
      \sedges{7}{7}
      \node[draw,fill,circle,opacity=0.25] (u) at (5-6) {};
      \draw[->] (5-6) -- (6-5) ;
      \node[draw,fill,circle,opacity=0.7] (v) at (6-5) {};
      \node[draw,fill,circle,opacity=0.7] (w) at (7-5) {};
      \sposition{2-4,4-4,5-4,6-7,7-1,7-2,7-3,7-4,6-6,6-7,2-6,2-7,3-6,3-5,3-4}{2-5,1-3,2-3,2-2,2-1,3-1,4-1,4-2,4-3,5-2,6-2,6-3,6-4}
      \draw[very thick] (3.75,2.75) -- (2.75,2.75) -- (2.75,1.25) -- (2.25,1.25) -- (2.25,1.75) -- (1.75,1.75) -- (1.75,0.75) -- (1.25,0.75) -- (1.25,1.75) -- (0.25,1.75) ;
      \draw[very thick] (3.75,2.25) -- (3.25,2.25) -- (3.25,0.75) -- (2.25,0.75) -- (2.25,0.25) -- (0.75,0.25) -- (0.75,1.25) -- (0.25,1.25) ;
    \end{tikzpicture}
    \hspace{4mm}
  }
  \caption{Illustration of a move and a winning group on a $7\times 7$ \slither board.}
  \label{fig:slither-def-and-rules}
\end{figure}
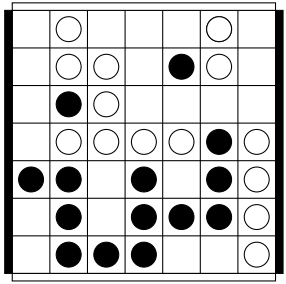
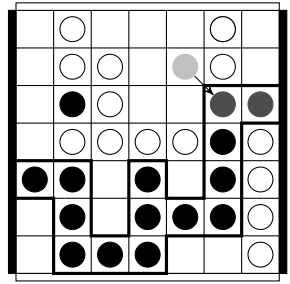

As in most connection games, a \emph{swap} rule is usually implemented.
That is, after the first move, the second player can decide either to play themself a move and the game goes on normally, or to become first player with that very same move.

\subsubsection{Computational complexity}
Here, we show that deciding if one player has a winning strategy from a given position is intractable.
Again, we present a reduction from \gamename{hex}.
A hexagonal cell of \gamename{hex} is encoded by the gadgets depicted in Figure~\ref{fig:slither-cell}.
More precisely, an empty cell (resp.~a cell containing a black stone, resp.~containing a white stone) is transformed into the portion of position of Figure~\ref{fig:slither-empty-cell} (resp.~Figure~\ref{fig:slither-black-cell}, resp.~Figure~\ref{fig:slither-white-cell}).

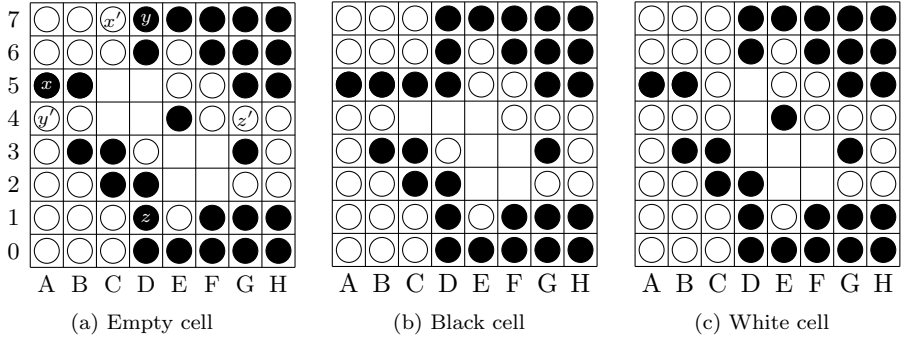
\begin{figure}
  \centering
  \subfloat[Empty cell]{\label{fig:slither-empty-cell} \begin{tikzpicture}[scale=0.88,font=\scriptsize] \sboard{8}{8}
\slithercell{0}{0}

\begin{scope}[font=\small]
\node at (0.5,0) {A}; \node at (1  ,0) {B}; \node at (1.5,0) {C}; \node at (2  ,0) {D}; \node at (2.5,0) {E}; \node at (3  ,0) {F}; \node at (3.5,0) {G}; \node at (4  ,0) {H};
\node at (0,0.5) {0}; \node at (0,1  ) {1}; \node at (0,1.5) {2}; \node at (0,2  ) {3}; \node at (0,2.5) {4}; \node at (0,3  ) {5}; \node at (0,3.5) {6}; \node at (0,4  ) {7};
\end{scope}

\node[text=white] at (1-6) {$x$}; \node at (3-8) {$x'$}; \node[text=white] at (4-8) {$y$}; \node at (1-5) {$y'$}; \node[text=white] at (4-2) {$z$}; \node at (7-5) {$z'$};

\sposition{4-4}{5-5}
 \end{tikzpicture}}
  \hfill
  \subfloat[Black cell]{\label{fig:slither-black-cell} \begin{tikzpicture}[scale=0.88] \sboard{8}{8}
\slithercell{0}{0}

\node at (0.5,0) {A}; \node at (1  ,0) {B}; \node at (1.5,0) {C}; \node at (2  ,0) {D}; \node at (2.5,0) {E}; \node at (3  ,0) {F}; \node at (3.5,0) {G}; \node at (4  ,0) {H};


\sposition{4-4}{3-6,4-6}
 \end{tikzpicture}}
  \hfill
  \subfloat[White cell]{\label{fig:slither-white-cell} \begin{tikzpicture}[scale=0.88] \sboard{8}{8}
\slithercell{0}{0}

\node at (0.5,0) {A}; \node at (1  ,0) {B}; \node at (1.5,0) {C}; \node at (2  ,0) {D}; \node at (2.5,0) {E}; \node at (3  ,0) {F}; \node at (3.5,0) {G}; \node at (4  ,0) {H};


\sposition{3-6,3-5}{5-5}
 \end{tikzpicture}}
  \caption{Basic gadgets representing \hex cells in \slither.}
  \label{fig:slither-cell}
\end{figure}

\begin{observation}\label{obs:empty-cell}
  When a player only places a stone in any empty square of the empty cell gadget, they create a configuration which is forbidden by the diagonal rule.
  Therefore, they need to also move one of their stones in the same cell gadget.
\end{observation}

\begin{lemma}\label{lem:black-cell}
  In a black cell $\mathcal C$ (Figure~\ref{fig:slither-black-cell}), White cannot prevent Black from having a group containing stones in $x$, $y$, and $z$.
\end{lemma}

\begin{proof}
Black stones on $x$ and $y$ are already in the same group.
Because of the diagonal rule, White cannot place a stone in cell $\mathcal C$ and move a stone in another cell gadget.
By Observation~\ref{obs:empty-cell} and the previous remark, if White moves a stone in cell $\mathcal C$, but decides to place a stone in another cell gadget, they can only do so in a white cell, which turns out to be useless.
Thus, White might as well place a stone \emph{and} move in cell $\mathcal C$.
After their move, White should occupy square C4; otherwise, Black places a stone on C4 and thereby connects their group containing $z$ to their group containing $x$ and $y$.

There are three ways for $W$ to occupy square C4: (1) move stone on B4 to C4, (2) move stone on D3 to C4, or (3) place a new stone on C4.
The first option cannot be extended into a legal move.
Indeed, the diagonal rule would impose that a stone is placed on D4, to connect the two diagonally-adjacent white stones on C4 and D3.
But then white stones on D4 and E5 would form a forbidden configuration.
In the second option, White cannot place a stone on D3 nor on D4, because of the diagonal rule.
And Black's next move would consist of moving the stone on C2 to D3, and placing a stone on D4, which connects $z$ to $x$ and $y$.
Finally, in the third option, White is forced to move their stone in D3 to a square other than D4.
And Black connects in the same manner.
\end{proof}

If one disregards the eight stones in A6, A7, B6, B7, G0, G1, H0, and H1, the diagonal A7-H0 is an axis of anti-symmetry for the empty cell gadget, in the sense that the image of a black (resp.~white) stone is a white (resp.~black) stone.
Due to this likeness, the following holds similarly.
\begin{lemma}\label{lem:white-cell}
  In a white cell (Figure~\ref{fig:slither-white-cell}), Black cannot prevent White from having a group containing stones on $x'$, $y'$, and $z'$.
\end{lemma}

The cell gadgets are glued together and attached to the edges of the board as described in Figure~\ref{fig:slither-hex}.
The empty squares that do not belong to any gadgets can be filled with stones of any color without affecting neither player's optimal strategy.
For convenience, we do not represent these stones.

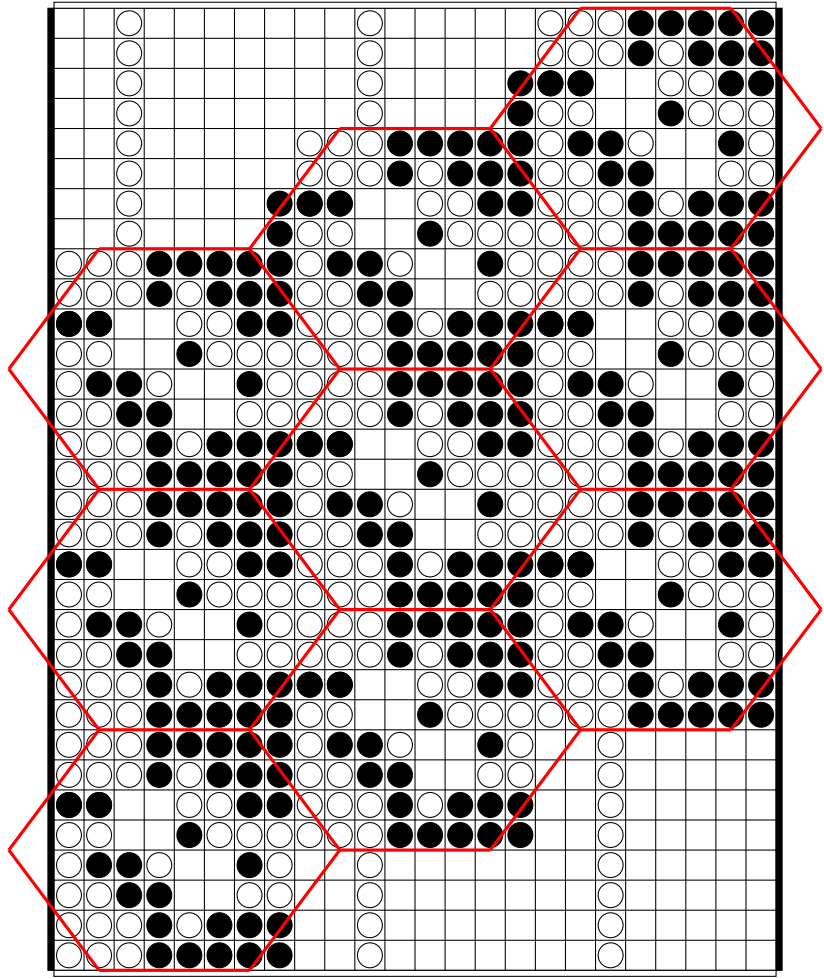
\begin{figure}
  \centering
\begin{tikzpicture}[scale=0.8]
  \sboard{24}{32}
  \sedges{24}{32}
  \slitherecell{ 0}{16} \slitherecell{ 8}{20} \slitherecell{16}{24}
  \slitherecell{ 0}{ 8} \slitherecell{ 8}{12} \slitherecell{16}{16}
  \slitherecell{ 0}{ 0} \slitherecell{ 8}{ 4} \slitherecell{16}{ 8}

  \foreach \j in {1,...,4} { \sposition{11-\j}{}}
  \foreach \j in {1,...,8} { \sposition{19-\j}{}}

  \foreach \j in {25,...,32} { \sposition{3-\j}{}}
  \foreach \j in {29,...,32} { \sposition{11-\j}{}}

  \sposition{}{8-25,8-26,16-29,16-30}

  \begin{scope}[scale=0.5,yshift=0.5cm]
    \draw[hexsid] ( 2, 0) -- ( 7, 0); \draw[hexsid] ( 2,  8) -- ( 7,  8); \draw[hexsid] ( 2, 16) -- ( 7, 16); \draw[hexsid] ( 2, 24) -- ( 7, 24);
    \draw[hexsid] (10, 4) -- (15, 4); \draw[hexsid] (10, 12) -- (15, 12); \draw[hexsid] (10, 20) -- (15, 20); \draw[hexsid] (10, 28) -- (15, 28);
    \draw[hexsid] (18,  8) -- (23,  8); \draw[hexsid] (18, 16) -- (23, 16); \draw[hexsid] (18, 24) -- (23, 24); \draw[hexsid] (18, 32) -- (23, 32);

    \draw[hexsid] (-1,  4) -- ( 2,  8); \draw[hexsid] (-1, 12) -- ( 2, 16); \draw[hexsid] (-1, 20) -- ( 2, 24);
    \draw[hexsid] ( 7,  0) -- (10,  4); \draw[hexsid] ( 7,  8) -- (10, 12); \draw[hexsid] ( 7, 16) -- (10, 20); \draw[hexsid] ( 7, 24) -- (10, 28);
    \draw[hexsid] (15,  4) -- (18,  8); \draw[hexsid] (15, 12) -- (18, 16); \draw[hexsid] (15, 20) -- (18, 24); \draw[hexsid] (15, 28) -- (18, 32);
    \draw[hexsid] (23,  8) -- (26, 12); \draw[hexsid] (23, 16) -- (26, 20); \draw[hexsid] (23, 24) -- (26, 28);

    \draw[hexsid] (-1,  4) -- ( 2,  0); \draw[hexsid] (-1, 12) -- ( 2,  8); \draw[hexsid] (-1, 20) -- ( 2, 16);
    \draw[hexsid] ( 7,  8) -- (10,  4); \draw[hexsid] ( 7, 16) -- (10, 12); \draw[hexsid] ( 7, 24) -- (10, 20);
    \draw[hexsid] (15, 12) -- (18,  8); \draw[hexsid] (15, 20) -- (18, 16); \draw[hexsid] (15, 28) -- (18, 24);
    \draw[hexsid] (23, 16) -- (26, 12); \draw[hexsid] (23, 24) -- (26, 20); \draw[hexsid] (23, 32) -- (26, 28);
  \end{scope}
\end{tikzpicture}
  \caption{Empty $3\times 3$ \gamename{hex} board reduced to a \gamename{slither} position.}
  \label{fig:slither-hex}
\end{figure}

The following observation is outlined by Figure~\ref{fig:slither-hex}.

\begin{observation}\label{obs:opt}
  When playing in a empty cell gadget, the best Black can do is to connect $x$, $y$, and $z$, and the best White can do is to connect $x'$, $y'$, and $z'$.
\end{observation}

From the empty cell gadget, Black can move stone E4 to D5 and place a stone on C5, resulting in the black cell configuration.
By Lemma~\ref{lem:black-cell} and Observation~\ref{obs:opt}, it is the optimal play within this cell.
Similarly, the optimal play for White in a given empty cell, is to move stone D3 to C4 and place a stone on C5, yielding the white cell.
Thus, having chosen the cell gadget where to play, the optimal move is to connect six paths going from this cell to the six adjacent cells in a hexagonal paving (Figure~\ref{fig:slither-hex}).
Hence, the built \gamename{slither} position simulates a game of \gamename{hex}, and so, \gamename{slither} is as hard as \gamename{hex}.

\begin{theorem}
  It is \pspace-complete to decide which player has a winning strategy from a given \gamename{slither} position.
\end{theorem}
\begin{proof}
The membership of this problem to \pspace{} boils down to noticing that the length of a game is bounded from above by the number of empty squares.
Indeed, at each move, one stone is added to the board.
Thus, a minimax depth-first search uses a polynomial amount of space.

The gadgets in Figure~\ref{fig:slither-cell} and their assembly as in Figure~\ref{fig:slither-hex} provide a reduction from \hex, a \pspace-hard problem.
\end{proof}

\section{Complexity of Havannah}
\label{sec:complexity-havannah}
\gamename{havannah} is a 2-player connection game played on a hexagonal board paved by hexagons.
White and Black place a stone of their color in turn in an unoccupied cell.
Stones cannot be taken, moved nor removed.
Two cells are neighbors if they share an edge.
A group is a connected component of stones of the same color via the neighbor relation.
A player wins if they realize one of the three following different structures: a circular group, called \emph{ring}, with at least one cell, possibly empty, inside; a group linking two corners of the board, called \emph{bridge}; or a group linking three edges of the board, called \emph{fork}.
With respect to the \emph{fork} winning condition, corner cells are not considered part of any edge.
Figure~\ref{fig:hav-winning} presents a board in which all three winning conditions are met by different groups of stones.

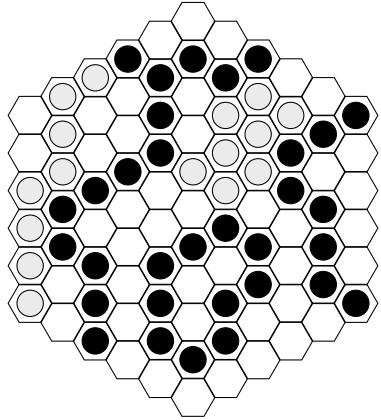
\begin{figure}
  \centering
\newcommand{\scala}{0.5}
\begin{tikzpicture}[>=stealth',scale=\scala,every node/.style={scale=2*\scala}]

  \havcoordinate{0}{18}{0}{18}
  \foreach \j in {4,...,9}  { \node[hav-empty] at  (0-\j)  {} ;}
  \foreach \j in {3,...,9}  { \node[hav-empty] at (0--\j)  {} ;}
  \foreach \j in {3,...,10} { \node[hav-empty] at  (1-\j)  {} ;}
  \foreach \j in {2,...,10} { \node[hav-empty] at (1--\j)  {} ;}
  \foreach \j in {2,...,11} { \node[hav-empty] at  (2-\j)  {} ;}
  \foreach \j in {1,...,11} { \node[hav-empty] at (2--\j)  {} ;}
  \foreach \j in {2,...,11} { \node[hav-empty] at  (3-\j)  {} ;}
  \foreach \j in {2,...,10} { \node[hav-empty] at (3--\j)  {} ;}
  \foreach \j in {3,...,10} { \node[hav-empty] at  (4-\j)  {} ;}
  \foreach \j in {3,...,9}  { \node[hav-empty] at (4--\j)  {} ;}
  \foreach \j in {4,...,9}  { \node[hav-empty] at  (5-\j)  {} ;}

  \foreach \cell in {0-4,0-5,0-6,0-7,0--7,0--8,0--9,1-10} {
    \node[hav-white] at  (\cell) {};
  }
  \foreach \cell in {3-7,2--7,3-8,3-9,3--9,4-9,3--8,3--7} {
    \node[hav-white] at  (\cell) {};
  }
  \foreach \cell in {2-3,2-4,2-5,2--5,3-6,3--5,3--4,3-4,3-3,2--2} {
    \node[hav-black] at  (\cell) {};
  }
  \foreach \cell in {1-3,1-4,1-5,0--5,0--6,1-7,1--7,2-8,2-9,2-10,1--10,2--10,3-10,3--10} {
    \node[hav-black] at  (\cell) {};
  }
  \foreach \cell in {5-4,4--4,4-7,4-8,4--5,4--6,4--8,5-9} {
    \node[hav-black] at  (\cell) {};
  }
\end{tikzpicture}
  \caption{\Hav winning conditions on a size 6 board.
    Each black group represents a winning pattern: from left to right, \emph{fork}, \emph{ring}, and \emph{bridge}.
    White groups are not winning, the left group does not constitute a fork because the corner is not part of any edge and so the group only connects two edges; the right group does not constitute a ring because it would need to enclose a non-empty surface.}
  \label{fig:hav-winning}
\end{figure}

\subsection{Preliminary results}
As the length of a game of \Hav is polynomially bounded, exploring the whole game tree with Depth-First Search can be done in polynomial space, so \Hav is in \pspace.

In our reduction, the \Hav board is large enough that the gadgets are far from the edges and the corners.
Additionally, the gadgets feature ring threats that are short enough that the bridges and forks winning conditions do not have any influence.
Before starting the reduction, we define threats and make two observations that will prove useful in the course of the reduction.

We call \emph{simple threat} a move which threatens to realize a ring on the next move.
There are only two kinds of answers to a simple threat: either win on the spot or defend by placing a stone in the cell creating this very threat.

\begin{lemma}
  \label{lem:simple-threat}
  If a player is not threatened, playing a simple threat forces the opponent to answer on the cell of the threat.
\end{lemma}
\begin{proof}
  Placing a stone on the cell of the threat wins the game against any other move of the opponent.
\end{proof}

A \emph{double threat} is a move which threatens to realize a ring on the next move on at least two different cells.
We will use \emph{threat} as a generic term to encompass both simple and double threats.
To be more concise, we will denote by $a_1; a_2; \dots; a_p$ a sequence of moves starting with one player playing move $a_1$, the opponent answering $a_2$, and so on.
We will specify which player starts the sequence if it is not clear from the context.
\begin{lemma}
  \label{lem:double-threat}
  If a player is not threatened, playing a double threat is winning.
\end{lemma}
\begin{proof}
 The player is not threatened, so their opponent cannot win next turn.
 Let $u$ and $v$ be two cells of the double threat.
 If the opponent plays in $u$, the player wins by playing in $v$.
 If the opponent plays anywhere else, the player wins by playing in $u$.
\end{proof}

A second-order threat, or \emph{2-threat}, is a move which threatens to realize a double threat on the next move.
That is, after a 2-threat is played by a player $P$, there is a set of cells $\{x, a, b\}$ such that both $P$ playing $\{x, a\}$ as well as $P$ playing $\{x, b\}$ would complete a ring.
We call the three cells involved in a 2-threat the \emph{carrier} of the threat, and in particular the cell $x$ is called the \emph{exit} of the threat.
For example, in Figure~\ref{fig:hav-11}, move $n$ is a 2-threat for Black with carrier $x$, $a$, and $b$.
Indeed, if Black follows up with $x$, then Black obtains a ring by playing $a$ or by playing $b$.

\begin{lemma}
  \label{lem:2-threat}
  If a player makes a 2-threat, then any opponent move that does not belong to the following list is losing.
  1) Immediately winning move 2) Simple threat 3) Move inside the carrier of the threat.
\end{lemma}

In the following subsections we propose gadgets that encode the different parts of a \emph{simplified} instance of \GG{}, as defined in Section~\ref{sec:gg}.
These gadgets have starting points and ending points.
The gadgets are assembled so that the ending point of a gadget coincides with the starting point of the next one.
The resulting instance of \gamename{havannah} is such that both players must enter in the gadgets by a starting point and leave it by an ending point otherwise they lose.
Wires and curves will enable us to encode the edges of the input graph.
While figures and lemmas are mostly presented from White's point of view, all the gadgets and lemmas work exactly the same way with colors reversed.

\subsection{Vertex gadgets}
Recall from Section~\ref{sec:gg} that simplified \GG{} instances only feature four types of vertices: $(1,1)$-, $(2, 1)$, $(0, 2)$-, and $(1,2)$-vertices.
Figure~\ref{fig:hav-gadgets} depicts how these four types of vertices are mapped into \Hav gadgets.
The diagrams illustrate the White version of the gadgets, corresponding to vertices belonging to the first player in \GG.
Each gadget can naturally by rotated by multiples of 60$^\circ$ and the colors can be swapped to obtain the Black gadgets.

In Figure~\ref{fig:hav-12} and~\ref{fig:hav-02}, after Black plays $n$, White can elect to play $c_1$ or $c_2$.
In the former case, Black is forced to reply $c_2$ which forces White to e\emph{x}it the gadget via $x_1$.
In the latter case, Black replies $c_1$ and then White exits the gadget via $x_2$.

\begin{figure}
  \centering
  \subfloat[$(1,1)$-vertex gadget.
    \label{fig:hav-11}
  ]{
    \centering
\begin{tikzpicture}[>=stealth',scale=0.5]
  \havcoordinate{1}{4}{1}{6}
  \foreach \i in {1,...,4} {
    \foreach \j in {1,...,6} {
      \node [hav-empty] at (\i-\j) {};
    }
  }
  \foreach \i in {1,...,4} {
    \foreach \j in {1,...,6} {
      \node [hav-empty] at (\i--\j) {};
    }
  }

  \foreach \i in {1-5,1--4,2-4,2--1,2--2,3-4,3--1,3--3,4-1,4-3,4--2} {
    \node at (\i)[hav-black] {};
  }
  \foreach \i in {1-4,1-6,1--3,1--5,2-3,2-5,2--4,3-3,3--2,4-2} {
    \node at (\i)[hav-white] {};
  }
  \draw[thick,->] (1--6) -- (2--5);
  \draw[thick,->] (3--4) -- (4--3);

  \node at (2--3) {$n$};
  \node at (3-1) {$b$};
  \node at (3-2) {$a$};
  \node at (4--1) {$x$};
\end{tikzpicture}
  }
  \hfill
  \subfloat[$(2,1)$-vertex gadget, also known as \emph{join}.
    \label{fig:hav-21}
  ]{
    \centering
\begin{tikzpicture}[>=stealth',scale=0.50]
  \havcoordinate{0}{6}{-1}{9}
  \foreach \i in {1,...,6} {
    \foreach \j in {0,...,9} {
      \node [hav-empty] at (\i-\j) {};
    }
  }
  \foreach \i in {0,...,6} {
    \foreach \j in {-1,...,8} {
      \node [hav-empty] at (\i--\j) {};
    }
  }

  \twoone{1}{2}{black}{white}

  \node at (1--6)  {$n_1$};
  \node at (2-6)   {$n_1'$};
  \node at (5--6)  {$n_2$};
  \node at (5-6)   {$n_2'$};
  \node at (3---1) {$x$};

  \foreach \i in {0--7,1-7,3-0,3-1,4-0,4-1,6-7,6--7} {
    \node [hav-black] at (\i) {};
  }
  \foreach \i in {0--6,0--8,1-6,1-8,1--7,3--0,3--1,5--7,6-6,6-8,6--6,6--8} {
    \node [hav-white] at (\i) {};
  }
  \draw[thick,->] (1-9) -- (2-8);
  \draw[thick,->] (6-9) -- (5-8);
  \draw[thick,->] (2--1) -- (2---1);
\end{tikzpicture}
  }\\
  \subfloat[$(0,2)$-vertex gadget, the starting vertex.
  ]{
    \centering
\begin{tikzpicture}[>=stealth',scale=0.5]
  \havcoordinate{0}{7}{0}{11}
  \foreach \i in {0,...,4} {
    \foreach \j in {1,...,10} {
      \node[hav-empty] at (\i-\j) {};
    }
  }
  \foreach \i in {0,...,3} {
    \foreach \j in {1,...,10} {
      \node[hav-empty] at (\i--\j) {};
    }
  }

  \startingvertex{1}{3}
  \foreach \i/\j in {0/1,0/2,0/3,0/4,0/5,1/1,1/2,1/3,1/5,1/8, 
                     2/9,3/8,3/10 
  } {
    \pgfmathtruncatemacro{\ii}{\i}
    \pgfmathtruncatemacro{\jj}{\j}
    \node[hav-black] at (\ii--\jj) {};
  }
  \foreach \i/\j in {2/4, 
                     2/7,2/9,3/7,3/8,3/10,4/9 
  } {
    \pgfmathtruncatemacro{\ii}{\i}
    \pgfmathtruncatemacro{\jj}{\j}
    \node[hav-black] at (\ii-\jj) {};
  }
  \foreach \i/\j in {1/4, 
                     2/7,2/8,3/9 
  } {
    \pgfmathtruncatemacro{\ii}{\i}
    \pgfmathtruncatemacro{\jj}{\j}
    \node[hav-white] at (\ii--\jj) {};
  }
  \foreach \i/\j in {1/2,1/3,1/4, 
                     3/9 
  } {
    \pgfmathtruncatemacro{\ii}{\i}
    \pgfmathtruncatemacro{\jj}{\j}
    \node[hav-white] at (\ii-\jj) {};
  }
  \draw[thick,->] (0-4) -- (0-2);
  \draw[thick,->] (2-10) -- (3-11);
  \node at (2--5) {$n$};
  \node at (2-5)  {$c_1$};
  \node at (2--6) {$c_2$};
  \node at (1-1)  {$x_2$};
  \node at (4-10) {$x_1$};
  \node at (2-8)  {$a_1$};
  \node at (1--7) {$b_1$};
  \node at (1-5)  {$a_2$};
  \node at (1-6)  {$b_2$};
  \node at (3--4) {$e$};
  \node at (4-4) {$f$};
\end{tikzpicture}
    \label{fig:hav-02}
  }
  \hfill
  \subfloat[$(1,2)$-vertex gadget, also known as \emph{split}.
  ]{
    \centering
\begin{tikzpicture}[>=stealth',scale=0.5]
  \havcoordinate{0}{7}{0}{10}
  \foreach \i in {3,...,7} {
    \foreach \j in {1,...,10} {
      \node [hav-empty] at (\i-\j) {};
    }
  }
  \foreach \i in {2,...,6} {
    \foreach \j in {1,...,10} {
      \node [hav-empty] at (\i--\j) {};
    }
  }

  \foreach \i/\j in {5/3,6/2,6/4, 
                     5/6,5/7,5/8,5/9, 
                     3/2,3/5,4/3,4/6 
  } {
    \pgfmathtruncatemacro{\ii}{\i}
    \pgfmathtruncatemacro{\jj}{\j}
    \node[hav-white] at (\ii--\jj) {};
  }
  \foreach \i/\j in {5/5,6/3,6/5,7/2,7/4, 
                     3/2,4/3,4/5,4/7 
  } {
    \pgfmathtruncatemacro{\ii}{\i}
    \pgfmathtruncatemacro{\jj}{\j}
    \node[hav-white] at (\ii-\jj) {};
  }
  \foreach \i/\j in {6/3, 
                     4/5,4/7,5/5,
                     2/2,3/1,3/4,4/2 
  } {
    \pgfmathtruncatemacro{\ii}{\i}
    \pgfmathtruncatemacro{\jj}{\j}
    \node[hav-black] at (\ii--\jj) {};
  }
  \foreach \i/\j in {6/4,7/3, 
                     5/7,5/9,5/10,6/6,6/7,6/8,6/9,6/10, 
                     3/1,3/3,4/2,4/4,5/3,5/4 
  } {
    \pgfmathtruncatemacro{\ii}{\i}
    \pgfmathtruncatemacro{\jj}{\j}
    \node[hav-black] at (\ii-\jj) {};
  }

  \draw[thick,->] (5-2) -- (3--0);
  \draw[thick,->] (6--6) -- (6--9);
  \draw[thick,->] (7-1) -- (6-2);

  \node at (5--4)  {$n$};
  \node at (4--4)  {$c_1$};
  \node at (5-6)   {$c_2$};
  \node at (5--10) {$x_1$};
  \node at (2--1)  {$x_2$};
  \node at (5-8)   {$a_1$};
  \node at (4--8)  {$b_1$};
  \node at (3--3)  {$a_2$};
  \node at (3-4)   {$b_2$};
  \node at (4-6)   {$e$};
  \node at (3--6)  {$f$};
\end{tikzpicture}
    \label{fig:hav-12}
  }
  \caption{White gadgets before being used.
    Black e\emph{n}ters the gadget by playing $n$ (either $n_1$ or $n_2$ in Figure~\ref{fig:hav-21}).
    In each case, optimal play leads White to e\emph{x}it the gadget via $x$ (either $x_1$ or $x_2$ in the choice gadgets Figure~\ref{fig:hav-12} and~\ref{fig:hav-02}).
    In both choice gadgets, exactly one of the sequences $c_1; c_2$ or $c_2; c_1$ is played between the entry and the exit of the gadget, up to White's decision.
  }
  \label{fig:hav-gadgets}
\end{figure}
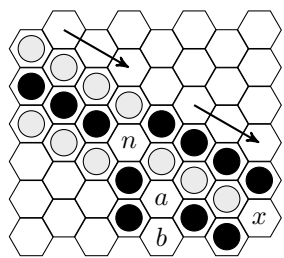
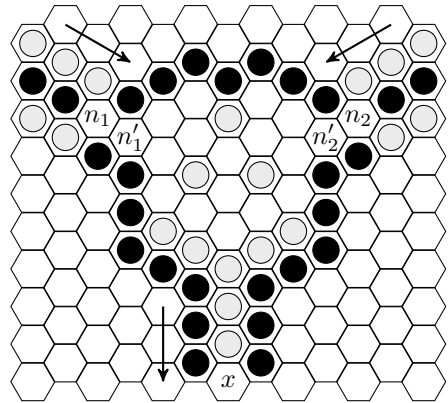
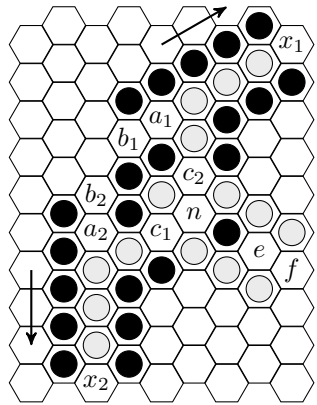
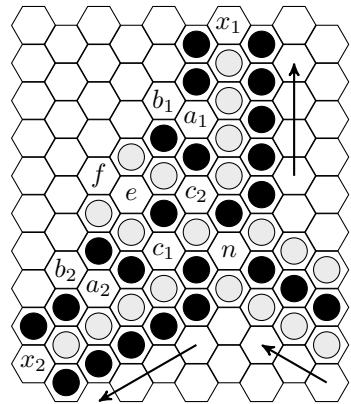






Let $(G, v_0)$ be a simplified instance of \GG{}. 
$G$ being bipartite, we denote by $V_1$ the side of the partition containing $v_0$, and $V_2$ the other side.
Player 1 moves the token from vertices of $V_1$ to vertices of $V_2$ and player 2 moves the token from $V_2$ to $V_1$.
We denote by $\phi$ the reduction from \GG{} to \Hav{}.
Each of Player 1's vertices is encoded by an instance of the corresponding gadget of Figure~\ref{fig:hav-gadgets}, and each of Player 2's vertices is encoded similarly with colors reversed.
Wires and curves are used to connect the gadgets.
As an example, we provide the reduction from the \GG{} instance from Figure~\ref{fig:gg-instance} in Figure~\ref{fig:hav-reduction}.

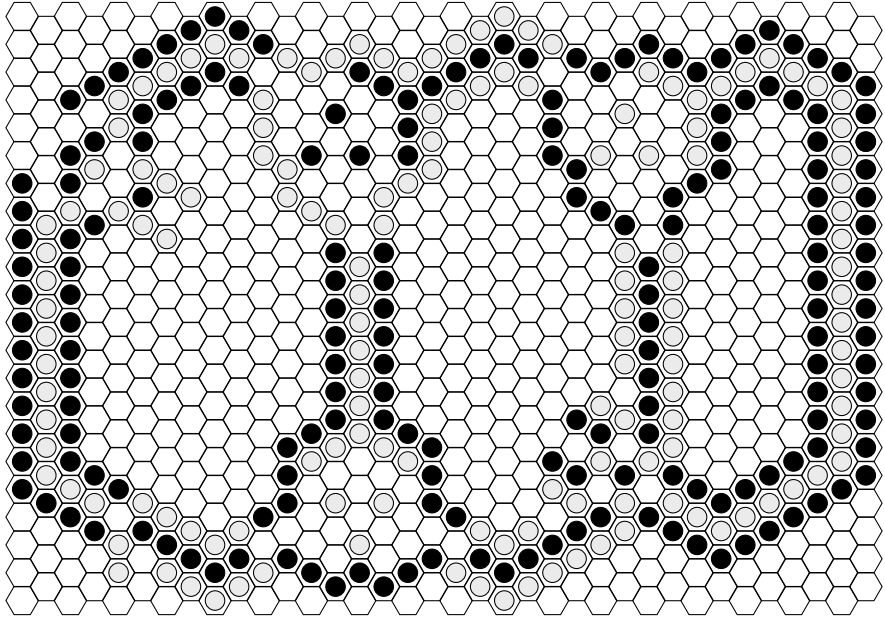
\begin{figure}
  \centering
\newcommand{\offseti}{0}
\newcommand{\offsetj}{0}
\newcommand{\scal}{0.37}
\begin{tikzpicture}[>=stealth',scale=\scal,every node/.style={scale=2*\scal}]

  \havcoordinate{0}{18}{5}{26}
  \foreach \i in {1,...,18} {
    \foreach \j in {6,...,26} {
      \node[hav-empty] at (\i-\j) {} ;
    }
  }
  \foreach \i in {0,...,17} {
    \foreach \j in {5,...,26} {
      \node[hav-empty] at  (\i--\j) {} ;
    }
  }

  \renewcommand{\offseti}{0}
  \renewcommand{\offsetj}{15}
  \startingvertex{\offseti+1}{\offsetj+3}
  \foreach \i/\j in {0/0,0/1,0/2,0/3,0/4,0/5,1/0,1/1,1/2,1/3,1/5,1/8, 
                     2/9,3/8,3/10 
  } {
    \pgfmathtruncatemacro{\ii}{\i + \offseti}
    \pgfmathtruncatemacro{\jj}{\j + \offsetj}
    \node[hav-black] at (\ii--\jj) {};
  }
  \foreach \i/\j in {2/4, 
                     2/7,2/9,3/7,3/8,3/10,4/9,4/11 
  } {
    \pgfmathtruncatemacro{\ii}{\i + \offseti}
    \pgfmathtruncatemacro{\jj}{\j + \offsetj}
    \node[hav-black] at (\ii-\jj) {};
  }
  \foreach \i/\j in {1/4, 
                     2/7,2/8,3/9 
  } {
    \pgfmathtruncatemacro{\ii}{\i + \offseti}
    \pgfmathtruncatemacro{\jj}{\j + \offsetj}
    \node[hav-white] at (\ii--\jj) {};
  }
  \foreach \i/\j in {1/0,1/1,1/2,1/3,1/4, 
                     3/9,4/10 
  } {
    \pgfmathtruncatemacro{\ii}{\i + \offseti}
    \pgfmathtruncatemacro{\jj}{\j + \offsetj}
    \node[hav-white] at (\ii-\jj) {};
  }

  \renewcommand{\offseti}{0}
  \renewcommand{\offsetj}{4}
  \foreach \i/\j in {0/5,0/6,0/7,0/8,0/9,0/10,1/4,1/6,1/7,1/8,1/9,1/10, 
                     2/5,3/3
  } {
    \pgfmathtruncatemacro{\ii}{\i + \offseti}
    \pgfmathtruncatemacro{\jj}{\j + \offsetj}
    \node[hav-black] at (\ii--\jj) {};
  }
  \foreach \i/\j in {1/5,2/4,2/6,3/4 
  } {
    \pgfmathtruncatemacro{\ii}{\i + \offseti}
    \pgfmathtruncatemacro{\jj}{\j + \offsetj}
    \node[hav-black] at (\ii-\jj) {};
  }
  \foreach \i/\j in {
                     1/5,2/2,2/3,3/2,3/4 
  } {
    \pgfmathtruncatemacro{\ii}{\i + \offseti}
    \pgfmathtruncatemacro{\jj}{\j + \offsetj}
    \node[hav-white] at (\ii--\jj) {};
  }
  \foreach \i/\j in {1/6,1/7,1/8,1/9,1/10, 
                     2/5,3/5 
  } {
    \pgfmathtruncatemacro{\ii}{\i + \offseti}
    \pgfmathtruncatemacro{\jj}{\j + \offsetj}
    \node[hav-white] at (\ii-\jj) {};
  }

  \renewcommand{\offseti}{2}
  \renewcommand{\offsetj}{1}
  \twooneb{\offseti+3}{\offsetj+5}{black}{white}
  \foreach \i/\j in {7/5,8/4,8/6, 
                     2/4,2/6,3/5,
                     5/11,5/12,5/13,5/14 
  } {
    \pgfmathtruncatemacro{\ii}{\i + \offseti}
    \pgfmathtruncatemacro{\jj}{\j + \offsetj}
    \node[hav-white] at (\ii--\jj) {};
  }
  \foreach \i/\j in {8/5,8/7, 
                     2/5,2/7,3/5,3/7
  } {
    \pgfmathtruncatemacro{\ii}{\i + \offseti}
    \pgfmathtruncatemacro{\jj}{\j + \offsetj}
    \node[hav-white] at (\ii-\jj) {};
  }
  \foreach \i/\j in {8/5, 
                     2/5
  } {
    \pgfmathtruncatemacro{\ii}{\i + \offseti}
    \pgfmathtruncatemacro{\jj}{\j + \offsetj}
    \node[hav-black] at (\ii--\jj) {};
  }
  \foreach \i/\j in {8/6, 
                     2/6,3/6, 
                     5/12,5/13,5/14,6/12,6/13,6/14 
  } {
    \pgfmathtruncatemacro{\ii}{\i + \offseti}
    \pgfmathtruncatemacro{\jj}{\j + \offsetj}
    \node[hav-black] at (\ii-\jj) {};
  }

  \renewcommand{\offseti}{2}
  \renewcommand{\offsetj}{12}
  \twoonec{\offseti+3}{\offsetj+7}{white}{black}
  \foreach \i/\j in {
                     2/12,2/14,3/13, 
                     7/12,8/13 
  } {
    \pgfmathtruncatemacro{\ii}{\i + \offseti}
    \pgfmathtruncatemacro{\jj}{\j + \offsetj}
    \node[hav-black] at (\ii--\jj) {};
  }
  \foreach \i/\j in {5/4,5/5,5/6,6/4,6/5,6/6, 
                     3/12,3/14, 
                     7/12,8/13 
  } {
    \pgfmathtruncatemacro{\ii}{\i + \offseti}
    \pgfmathtruncatemacro{\jj}{\j + \offsetj}
    \node[hav-black] at (\ii-\jj) {};
  }
  \foreach \i/\j in {5/4,5/5, 
                     2/13, 
                     7/11,7/13,8/12,8/14 
  } {
    \pgfmathtruncatemacro{\ii}{\i + \offseti}
    \pgfmathtruncatemacro{\jj}{\j + \offsetj}
    \node[hav-white] at (\ii--\jj) {};
  }
  \foreach \i/\j in {
                     3/13,
                     7/13,8/12,8/14 
  } {
    \pgfmathtruncatemacro{\ii}{\i + \offseti}
    \pgfmathtruncatemacro{\jj}{\j + \offsetj}
    \node[hav-white] at (\ii-\jj) {};
  }

  \renewcommand{\offseti}{7}
  \renewcommand{\offsetj}{11}
  \twoonec{\offseti+4}{\offsetj+8}{black}{white}
  \foreach \i/\j in { 
                     4/14, 
                     8/13,9/13,10/5,10/6,10/7,10/8,10/9,10/10,10/11,10/12 
  } {
    \pgfmathtruncatemacro{\ii}{\i + \offseti}
    \pgfmathtruncatemacro{\jj}{\j + \offsetj}
    \node[hav-white] at (\ii--\jj) {};
  }
  \foreach \i/\j in {6/5,6/6,6/7,7/5,7/6,7/7, 
                     4/13,4/15, 
                     8/13,9/14,10/13 
  } {
    \pgfmathtruncatemacro{\ii}{\i + \offseti}
    \pgfmathtruncatemacro{\jj}{\j + \offsetj}
    \node[hav-white] at (\ii-\jj) {};
  }
  \foreach \i/\j in {6/5,6/6, 
                     8/12,8/14,9/12,9/14,10/13 
  } {
    \pgfmathtruncatemacro{\ii}{\i + \offseti}
    \pgfmathtruncatemacro{\jj}{\j + \offsetj}
    \node[hav-black] at (\ii--\jj) {};
  }
  \foreach \i/\j in {
                     4/14, 
                     8/14,9/13,9/15,10/5,10/6,10/7,10/8,10/9,10/10,10/11,10/12,10/14,11/5,11/6,11/7,11/8,11/9,11/10,11/11,11/11,11/12,11/13 
  } {
    \pgfmathtruncatemacro{\ii}{\i + \offseti}
    \pgfmathtruncatemacro{\jj}{\j + \offsetj}
    \node[hav-black] at (\ii-\jj) {};
  }

  \renewcommand{\offseti}{8}
  \renewcommand{\offsetj}{5}
  \pgfmathtruncatemacro{\ii}{4 + \offseti}
  \pgfmathtruncatemacro{\iii}{5 + \offseti}
  \pgfmathtruncatemacro{\jj}{5 + \offsetj}
  \node[hav-black] at (\iii-\jj) {};
  \node[hav-white] at (\ii--\jj) {};

  \foreach \i/\j in {5/3,6/2,6/4,7/2,7/4,8/3,8/5,9/4, 
                     5/6,5/7,5/8,5/9,5/10, 
                     3/2,3/5,4/3,4/6 
  } {
    \pgfmathtruncatemacro{\ii}{\i + \offseti}
    \pgfmathtruncatemacro{\jj}{\j + \offsetj}
    \node[hav-black] at (\ii--\jj) {};
  }
  \foreach \i/\j in {6/3,6/5,7/2,7/4,8/3,8/5,9/4,9/6,9/7,9/8,9/9,9/10,10/5,10/6,10/7,10/8,10/9,10/10, 
                     3/2,4/3,4/5,4/7 
  } {
    \pgfmathtruncatemacro{\ii}{\i + \offseti}
    \pgfmathtruncatemacro{\jj}{\j + \offsetj}
    \node[hav-black] at (\ii-\jj) {};
  }
  \foreach \i/\j in {6/3,7/3,8/4,9/5,9/6,9/7,9/8,9/9,9/10, 
                     4/7,5/5,
                     3/1,3/4,4/2 
  } {
    \pgfmathtruncatemacro{\ii}{\i + \offseti}
    \pgfmathtruncatemacro{\jj}{\j + \offsetj}
    \node[hav-white] at (\ii--\jj) {};
  }
  \foreach \i/\j in {6/4,7/3,8/4,9/5, 
                     5/7,5/9,5/10,6/6,6/7,6/8,6/9,6/10, 
                     3/1,3/3,4/2,4/4,5/3,5/4 
  } {
    \pgfmathtruncatemacro{\ii}{\i + \offseti}
    \pgfmathtruncatemacro{\jj}{\j + \offsetj}
    \node[hav-white] at (\ii-\jj) {};
  }
\end{tikzpicture}
  \caption{\gamename{Havannah} gadgets representing the \GG{} instance from Figure~\ref{fig:gg-instance}.
  Black to play first.}
  \label{fig:hav-reduction}
\end{figure}



\subsection{Properties}
We now establish a few properties that will help us prove that Player 1 wins the \GG instance if and only if White wins the corresponding \Hav position.
The main idea of the reduction is that optimal play in the \Hav position consists of a sequence of White defending a 2-threat by making a 2-threat which Black defends by making their own 2-threat and so on and so forth.
Under optimal play, the current 2-threat is in the gadget that corresponds to the token vertex in the \GG instance.

\begin{lemma}
  \label{lem:hav-force-win}
  If White plays the 2-threat in one of their $(2, 1)$-vertex gadgets (Figure~\ref{fig:hav-21-network}) and Black has no winning sequence of simple threats elsewhere on the board, then White can force a win.
\end{lemma}
\begin{proof}
  With move 1 in Figure~\ref{fig:hav-21-network}, Player White makes four distinct 2-threats with respective carriers: $\{k,l,m\}$, $\{k,l,n\}$, $\{k,m,n\}$, and $\{l,m,n\}$.
  The intersection of all these carriers is empty, so no single Black move is inside the carrier of all threats.
  By Lemma~\ref{lem:2-threat}, the only non-losing Black moves would be those satisfying criterion 1) or 2).
  Since we assumed that Black had no winning sequence of simple threats, no move satisfies criterion 1) nor criterion 2).
  Therefore, White can force a win.
\end{proof}

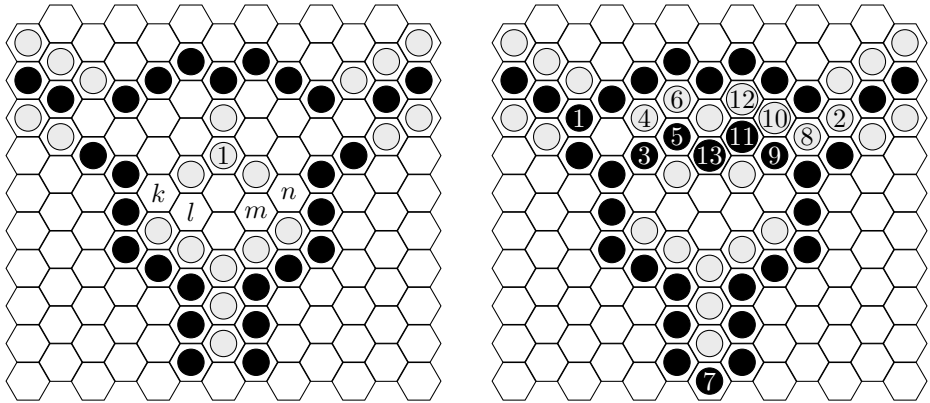
\begin{figure}
  \centering
  \subfloat[White's 2-threats network in the $(2,1)$ gadget.
    Black cannot defend against this threat locally and would need to find a winning sequence of simple threats elsewhere on the board.
  ]{
    \centering
    \hspace{-2mm}
\begin{tikzpicture}[>=stealth',scale=0.50]
  \havcoordinate{0}{6}{-1}{9}
  \foreach \i in {1,...,6} {
    \foreach \j in {0,...,9} {
      \node [hav-empty] at (\i-\j) {};
    }
  }
  \foreach \i in {0,...,6} {
    \foreach \j in {-1,...,8} {
      \node [hav-empty] at (\i--\j) {};
    }
  }

  \twoone{1}{2}{black}{white}

  \node at (2--4) {$k$};
  \node at (3-4)  {$l$};
  \node at (4-4)  {$m$};
  \node at (4--4) {$n$};

  \foreach \i in {0--7,1-7,3-0,3-1,4-0,4-1} {
    \node [hav-black] at (\i) {};
  }
  \foreach \i in {0--6,0--8,1-6,1-8,1--7,3--0,3--1} {
    \node [hav-white] at (\i) {};
  }
  \foreach \i in {6-7,6--7} {
    \node [hav-black] at (\i) {};
  }
  \foreach \i in {5--7,6-6,6-8,6--6,6--8} {
    \node [hav-white] at (\i) {};
  }

  \node [hav-white] at (3--5) {1};
\end{tikzpicture}
    \hspace{-4mm}
    \label{fig:hav-21-network}
  }
  \hfill
  \subfloat[
    Main winning line for Black when White does not exit the gadget with $x$ in Figure~\ref{fig:hav-21}, assuming Black had entered from $n_1$ in Figure~\ref{fig:hav-21}.
    Move 2 and 8 are interchangeable for White and are both losing.]{
    \centering
    \hspace{-2mm}
\begin{tikzpicture}[>=stealth',scale=0.50]
  \havcoordinate{0}{6}{-1}{9}
  \foreach \i in {1,...,6} {
    \foreach \j in {0,...,9} {
      \node [hav-empty] at (\i-\j) {};
    }
  }
  \foreach \i in {0,...,6} {
    \foreach \j in {-1,...,8} {
      \node [hav-empty] at (\i--\j) {};
    }
  }

  \twoone{1}{2}{black}{white}

  \node [hav-black] at (1--6) {$1$};
  \node [hav-white] at (5--6)  {$2$};
  \node [hav-black] at (2--5) {$3$};
  \node [hav-white] at (2--6)  {$4$};
  \node [hav-black] at (3-6) {$5$};
  \node [hav-white] at (3-7) {$6$};
  \node [hav-black] at (3---1) {$7$};
  \node [hav-white] at (5-6)   {$8$};
  \node [hav-black] at (4--5) {$9$};
  \node [hav-white] at (4--6)   {$10$};
  \node [hav-black] at (4-6)  {$11$};
  \node [hav-white] at (4-7)   {$12$};
  \node [hav-black] at (3--5) {$13$};

  \foreach \i in {0--7,1-7,3-0,3-1,4-0,4-1,6-7,6--7} {
    \node [hav-black] at (\i) {};
  }
  \foreach \i in {0--6,0--8,1-6,1-8,1--7,3--0,3--1,5--7,6-6,6-8,6--6,6--8} {
    \node [hav-white] at (\i) {};
  }
\end{tikzpicture}
    \hspace{-4mm}
    \label{fig:hav-21-threats}
  }
  \caption{Possible developments on a White $(2,1)$-vertex gadget.}
\end{figure}

This 2-threat network remains valid even if the $(2, 1)$ gadget has been entered and exited before move 1 is played.
Note also that there is no preemptive way for Black to ``break'' the network: If Black were to play in 1, then White could play the cell just beneath 1 and create a winning 2-threat network.
Similarly, move $k$ to $n$ from Black before 1 is played would let White create another 2-threat network in the same gadget.

The dual result is true for Black's $(2, 1)$ gadget.


\begin{corollary}
  \label{cor:no-waste}
  A move by a player $P$ that does maintain the existence of a sequence of winning threats for $P$ is a losing move.
\end{corollary}
\begin{proof}
  Because we are reducing from a \emph{simplified} \GG instance, we know that each player has at least one $(2, 1)$-vertex.
  Suppose that a player makes a move that leaves the board without a potential sequence of winning threat for them, then the opponent can make the winning 2-threat network in one of their $(2, 1)$-vertices and force a win by Lemma~\ref{lem:hav-force-win}.
\end{proof}

Corollary~\ref{cor:no-waste} ensures that not only should the players defend against the opponent's 2-threats, but that their defense should create 2-threats for the opponent to defend.


Before showing that all the gadgets behave as expected, let us notice that the exchange $3; 4; 5; 6$ in Figure~\ref{fig:hav-21-threats} can be played in any $(2, 1)$ vertex previously entered from $n_1$, using Figure~\ref{fig:hav-21}'s notation.
This exchange (or the symmetrical one $9; 10; 11; 12$ when the gadget has been entered from $n_2$) does not harm Black and the White moves are forced, so we will assume that the exchange always takes place as soon as available, even if White correctly exits the gadget in $x$.


\begin{lemma}
  \label{lem:hav-21}
  If Black enters a White $(2,1)$-vertex gadget via $n_1$ or $n_2$ (Figure~\ref{fig:hav-21}), and White does not have a threat, then White is forced to play in the exit point $x$.
\end{lemma}
\begin{proof}
  When Black plays in $n_1$, it creates a 2-threat with carrier $\{x, n_2, n_2'\}$.
  Suppose that White plays in $n_2$, then Black has a winning sequence of simple threats as displayed in Figure~\ref{fig:hav-21-threats}.
  If instead White plays in $n_2'$, then the same winning sequence is available for Black, simply swapping White move 2 for White move 8 in Figure~\ref{fig:hav-21-threats}.
  Invoking Lemma~\ref{lem:2-threat}, we know that White needs to play in the carrier, so we can conclude that White needs to play in $x$.

  The proof when Black plays in $n_2$ is symmetrical.
\end{proof}

\begin{lemma}
\label{lem:hav-11}
  If Black enters a White $(1, 1)$-vertex gadget in $n$ (Figure~\ref{fig:hav-11}), and White does not have a threat, then White is forced to play in the exit point $x$.
\end{lemma}
\begin{proof}
  When Black plays in $n_1$, it creates a 2-threat with carrier $\{x, a, b\}$.
  White playing in $a$ or in $b$ does not create any new threat for Black, so by Corollary~\ref{cor:no-waste} these moves are losing.
  Invoking Lemma~\ref{lem:2-threat}, we know that White needs to play in the carrier, so we can conclude that White needs to play in $x$.
\end{proof}

\begin{lemma}
  \label{lem:hav-12}
  If Black enters a White $(1,2)$-vertex gadget in $n$ (Figure~\ref{fig:hav-12}), and neither player has a threat available elsewhere on the board, then optimal play dictates that at most two sequences of moves are possible, $c_1; c_2; x_1$ and $c_2; c_1; x_2$, up to White's choice.
\end{lemma}
\begin{proof}
  Let us first show that after Black $n$ and White $c_1$, the sequence $c_2; x_1$ is forced.
  White's $c_1$ creates a 2-threat with carrier $\{c_2, e, f\}$.
  Black playing in $e$ or in $f$ does not create any new threat for White, so by Corollary~\ref{cor:no-waste} these moves are losing.
  Invoking Lemma~\ref{lem:2-threat}, we know that Black needs to play in the carrier, so we can conclude that Black needs to play in $c_2$.
  Black's $c_2$ creates a 2-threat with carrier $\{x_1, a_1, b_1\}$.
  White's playing $a_1$ or $b_1$ does not create any new threat for Black, so by the same reasoning, we obtain that White is forced to play $x_1$.

  The same style of argumentation shows that if White starts with $c_2$, then Black needs to reply $c_1$ and $x_2$ is forced for White.

  Now suppose that after Black $n$, White plays a non-threat elsewhere on the board.
  Black plays in $c_1$, forcing White $c_2$.
  Black then plays $x_2$ creating a double threat in $a_2$ and $b_2$ and White has lost.
\end{proof}

\begin{lemma}
  \label{lem:hav-02}
  When Black initiates in a White $(0,2)$-vertex gadget in $n$ (Figure~\ref{fig:hav-02}), if neither player has a threat available elsewhere on the board, then optimal play dictates that at most two sequences of moves are possible, $c_1; c_2; x_1$ and $c_2; c_1; x_2$, up to White's choice.
\end{lemma}
\begin{proof}
  The proof is virtually identical to that of Lemma~\ref{lem:hav-12}.
\end{proof}

From now on, whenever we consider choice gadgets, we simply write that after Black enters in $n$, White plays in $x_1$ or $x_2$ to exit.
That is, we omit the intermediate sequences $c_1; c_2$ or $c_2; c_1$ and leave them implicit.

\begin{lemma}
  \label{lem:hav-reentering}
  If Black has no threat, then Black re-entering a White $(2,1)$ gadget after it was visited before is a losing move (Figure~\ref{fig:hav-21-reentering}).
\end{lemma}
\begin{proof}
  When Black re-enters the gadget, the only new threat it creates is the winning sequence for Black starting with $a; b; c; d; 2$.
  This sequence can be defended against by White playing 2 in this gadget and creating the network of 2-threats seen in Lemma~\ref{lem:hav-force-win}.
\end{proof}

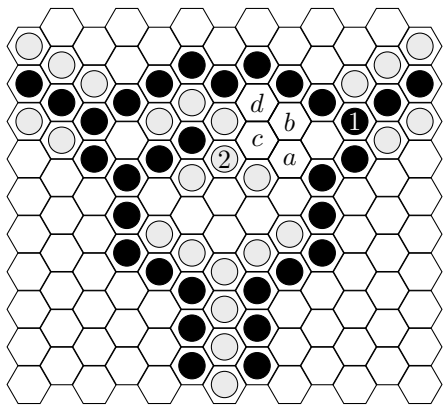
\begin{figure}
  \centering
  \hspace{-2mm}
\begin{tikzpicture}[>=stealth',scale=0.50]
  \havcoordinate{0}{6}{-1}{9}
  \foreach \i in {1,...,6} {
    \foreach \j in {0,...,9} {
      \node [hav-empty] at (\i-\j) {};
    }
  }
  \foreach \i in {0,...,6} {
    \foreach \j in {-1,...,8} {
      \node [hav-empty] at (\i--\j) {};
    }
  }

  \twoone{1}{2}{black}{white}

  \node [hav-black] at (1--6)  {};
  \node [hav-white] at (3---1) {};
  \node [hav-black] at (2--5)  {};
  \node [hav-white] at (2--6)  {};
  \node [hav-black] at (3-6)   {};
  \node [hav-white] at (3-7)   {};

  \node [hav-black] at (5--6) {$1$};
  \node [hav-empty] at (4--5) {$a$};
  \node [hav-empty] at (4--6)   {$b$};
  \node [hav-empty] at (4-6)  {$c$};
  \node [hav-empty] at (4-7)   {$d$};
  \node [hav-white] at (3--5) {$2$};

  \foreach \i in {0--7,1-7,3-0,3-1,4-0,4-1} {
    \node [hav-black] at (\i) {};
  }
  \foreach \i in {0--6,0--8,1-6,1-8,1--7,3--0,3--1} {
    \node [hav-white] at (\i) {};
  }
  \foreach \i in {6-7,6--7} {
    \node [hav-black] at (\i) {};
  }
  \foreach \i in {5--7,6-6,6-8,6--6,6--8} {
    \node [hav-white] at (\i) {};
  }

\end{tikzpicture}
  \hspace{-4mm}
  \caption{Reentering a White $(2,1)$-vertex gadget.
    The moves $n_1$ and $x$ from Figure~\ref{fig:hav-21} were played earlier, and Black reenters the gadget with $n_2$ as move 1.
    Move 2 wins the game for White.}
  \label{fig:hav-21-reentering}
\end{figure}

\begin{theorem}
  \label{hav-main-th}
  \gamename{Havannah} is \pspace{}-complete.
\end{theorem}
\begin{proof}
  We have already mentioned that \Hav{} $\in$ \pspace{} and the gadgets presented in this Section constitute a polynomial time reduction from a \pspace-complete problem.
  We shall now prove that the reduction is sound, that is: player 1 is winning in $(G,v_0)$ if and only if White is winning in $\phi((G,v_0))$ when Black starts.

  The simulation of \GG on \Hav proceeds as follows.
  The initial position has a single 2-threat, in the unique $(0, 2)$ gadget and no potential simple threats otherwise.
  Optimal play according to the prescription of Lemmas~\ref{lem:hav-21} through~\ref{lem:hav-02}, that is, entering a gadget with $n$ (or $n_1$ or $n_2$) and exiting it with $x$ (or $x_1$ or $x_2$) maintains the invariant that the position has a single 2-threat up until a $(2, 1)$ gadget is re-entered, at which point the game ends.

  Black's first move is forced and is entering White's $(0, 2)$ gadget with $n$.
  On the one hand, the \Hav gadget containing the current 2-threat, \emph{i.e.}, the one having most recently been entered, corresponds to the token vertex in \GG.
  On the other hand, a player re-entering a $(2, 1)$ gadget in \Hav loses, so it cannot be part of a winning strategy for that player.

  As a result, we can map any \Hav winning strategy for White to a \GG winning strategy for Player 1 and any \Hav winning strategy for Black to a \GG winning strategy for Player 2.
  Lemmas~\ref{lem:hav-21} through~\ref{lem:hav-02} together with the ``no threat elsewhere invariant'' guarantee that exiting the most recently entered gadget is the optimal \Hav play for both players.
  Since the number of vertices is finite, one player will be forced to re-enter a $(2, 1)$ gadget under optimal play.
  This means that the theoretical outcome of the translated instance is not a draw, so one player has indeed a winning strategy in \Hav and finding it solves \GG.
\end{proof}

\section{Parameterized Complexity of Short Generalized Hex}
\label{sec:parameterized-complexity-hex}
\paragraph{Generalized Hex}
\gamename{Shannon's vertex switching game} is more commonly called \GHex since it generalizes \hex played on a hexagonal tiling to any (potentially non planar) graphs.
A \GHex-instance is a graph $G$ with two specified vertices $s$ and $t$, each containing a white pebble.
White (Player~1) and Black (Player~2) alternate in playing pebbles of their color in unoccupied vertices of the graph.
White wins if they manage to create an $(s,t)$-path where each vertex contains a white pebble.
Black wins if they prevent White from doing so.
When a black pebble is played in a vertex, one can equivalently imagine that this vertex has been removed from the graph.
Any edge having a white pebble at its two endpoints can be equivalently contracted to a single vertex containing a white pebble.
That way, White wins if $s$ and $t$ end up within the same vertex.
In all generality, one can assume that the initial instance already contains white pebbles in some other vertices than $s$ and $t$.

\paragraph{Parameterized complexity}
In a nutshell, parameterized complexity aims at solving hard problems (be it \np-hard, \pspace-hard, etc.) in time $f(k)n^{O(1)}$, called \fpt time, (for fixed-parameter tractable), where $n$ is the size of the instance, $f$ is a computable function, and $k$ is a \emph{parameter} of the problem.
The parameter can take various forms: \emph{size of the solution} for optimization problems, \emph{treewidth} or \emph{maximum degree} for graph problems, \emph{size of the alphabet} for word problems, to name a few.
Assuming the problem we are trying to solve is \np-hard, function $f$ has to be superpolynomial, unless \p $=$ \np.
When such an \fpt algorithm exists, $f$ is usually exponential.
However, if our parameter $k$ is \emph{small} compared to the size of the instance $n$, what we obtain is that the exponential blow-up is limited to the small value $k$.
Some problems, like finding a clique of size $k$ in a graph having $n$ vertices, does not admit an algorithm running in time $f(k)n^{O(1)}$ for any computable function $f$, unless \textsc{3-sat} can be solved in subexponential-time which is believed unlikely.
Phrased with the terminology of parameterized complexity, \textsc{clique} is unlikely to be \fpt parameterized by the size of the solution.
There is a whole hierarchy of problems beyond \fpt: $W[1] \subseteq W[2] \subseteq \ldots \subseteq W[SAT] \subseteq W[P] \subseteq AW[*]$.
For our purpose here, we do not need to define formally any of those classes.
Instead, we just mention that \textsc{clique} is $W[1]$-complete for \fpt reductions, where an \fpt reduction may blow-up the size of the instance $n$ only polynomially but the new parameter can be any computable function of the old parameter.
A problem $\Pi$ is $W[1]$-hard with respect to $\kappa$ if there is an \fpt reduction from \textsc{clique} parameterized by the size of a solution $k$ to $\Pi$ parameterized by $\kappa$.
It is highly suspected that \fpt $\neq W[1]$.
Therefore, showing that a problem is $W[1]$-hard with respect to some parameter gives a strong evidence that there is no \fpt algorithm.
The class $AW[*]$ is a superset of $W[1]$ which encompasses hard parameterized problems with alternation.
We now describe the basic $AW[*]$-complete problem (again, $\Pi$ is $AW[*]$-hard if there is an \fpt reduction from this problem to $\Pi$).
Given a Boolean formula $\phi$ and a partition $V_1, V_2, \ldots, V_k$ of its variables, two players Existential and Universal alternate in choosing which single variable of $V_i$ is set to true (all the other variables of $V_i$ are set to false).
Existential chooses in $V_1$, then Universal chooses in $V_2$, and so on.
Existential wins iff $\phi$ is satisfied by the resulting assignment.
This problem is $AW[*]$-complete parameterized by the number $k$ of partite sets.

As an easy exercise, one can check that if there is an \fpt reduction from a problem $A$ to a problem $B$ and $B$ admits an \fpt algorithm, then so does $A$.
Two recent and extensive textbooks on parameterized complexity present these concepts in far greater detail~\cite{Cygan15,DowneyF13}.

\paragraph{Short games}
In their introductory book to parameterized complexity \cite{DowneyF99}, Downey and Fellows devote a small section to \emph{$k$-move games}.
The $k$-move (or \emph{short}) variant of a two-player game consists of deciding if Player $1$ can win in at most $k$ moves.
More formally, if one identifies a strategy to a decision tree, one needs to determine if Player $1$ has a winning strategy of depth at most $k$.
The problem is parameterized by $k$.
The PhD thesis of Allan Scott \cite{ScottThesis} is dedicated to short games.
The problem \textsc{short chess} was shown $AW[*]$-complete \cite{ScottS2008}.
Pursuit-evasion games (also known as cops and robber games) have been studied within the framework of short games \cite{ScottS10} as well as infinite two-player games on graphs \cite{Bjorklund03}.

The parameterized complexity of \sGHex was first asked in the book of Downey and Fellows \cite{DowneyF99} and is still open \cite{FominM2012,DowneyF13}.
We answer partially this question by showing that \sGHex is unlikely to be \fpt.
More precisely, we prove that \sGHex is $W[1]$-hard.
Nevertheless, this problem could be higher on the intractability hierarchy.
It is conjectured $AW[*]$-complete and Scott proved that it belongs to $AW[*]$.
From a purely practical perspective, our partial answer is already satisfactory in asserting that we should not expect to find an \fpt algorithm solving \sGHex{}.

\paragraph{Multicolored clique}
In \textsc{multicolored clique}, one is given a graph where the vertices are partitioned into $k$ sets called \emph{color classes} and the goal is to find a clique of size $k$ that intersects each color class exactly once.
The \textsc{multicolored clique} is also $W[1]$-complete and it is often more convenient to design \fpt reductions from this variant than from \textsc{clique}.

\begin{theorem}
\SGHex is $W[1]$-hard.
\end{theorem}

\begin{proof}
From an instance $G=(V=V_1 \uplus V_2 \uplus \ldots \uplus V_k, E)$ of the $W[1]$-hard \textsc{multicolored clique}, we build an equivalent instance $(H,s,t)$ of \sGHex in the following way.
For each vertex $v \in V$, we add four vertices $v^0$, $v^1$, $v^2$, $v^3$ to $H$ and a path $sv^0v^1v^2v^3t$.
For every $i \in [k]$, we add a vertex $t_i$, and we link $t_i$ to $t$ and to all the vertices $v^2$ such that $v \in V_i$.
Similarly, we add a vertex $s_i$ that we link to $s$ and to all the vertices $v^1$ such that $v \in V_i$.
That finishes the construction of the selector gadgets (see Figure~\ref{fig:selector}).

\begin{figure}
\centering
\begin{tikzpicture}[scale=0.8]

\node[draw,circle,inner sep=-0.25cm] (s) at (0.7,-3) {$s$} ;
\node[draw,circle,inner sep=-0.25cm] (t) at (0.7,7.5) {$t$} ;

\foreach \i in {1,...,3} {
  \node[draw,circle,inner sep=-0.28cm] (t\i) at (5 * \i - 11.5,4.5) {$t_\i$} ;
  \node[draw,circle,inner sep=-0.28cm] (s\i) at (5 * \i - 11.5,0) {$s_\i$} ;
  \node (sv\i) at (5 * \i - 10.2,1.5) {} ;
  \node (ev\i) at (5 * \i - 8.3,1.5) {} ;
  \foreach \j in {1,...,6} {
    \foreach \k in {0,...,3} {
      \node[draw,circle] (u\i\j\k) at (5 * \i - 11 + 0.5 * \j,1.5 * \k) {} ;
    }
    \draw (u\i\j1) -- (s\i) -- (s) -- (u\i\j0) -- (u\i\j1) -- (u\i\j2) -- (u\i\j3) -- (t) -- (t\i) -- (u\i\j2) ;
  }
  \node[draw,rounded corners,rectangle,very thick,fit=(sv\i) (ev\i)] (v\i) {} ;
  \node (tv\i) at (5 * \i - 11, 1.5) {$V_\i$} ;
}
\foreach \k in {0,...,3} {
  \node at (-4,1.5 * \k + 0.55) {$u^\k$} ;
  \node at (0,1.5 * \k + 0.55) {$v^\k$} ;
}

\end{tikzpicture}
\caption{The selector gadgets with $k=3$, and some vertices $u \in V_1$ and $v \in V_2$.}
\label{fig:selector}
\end{figure}
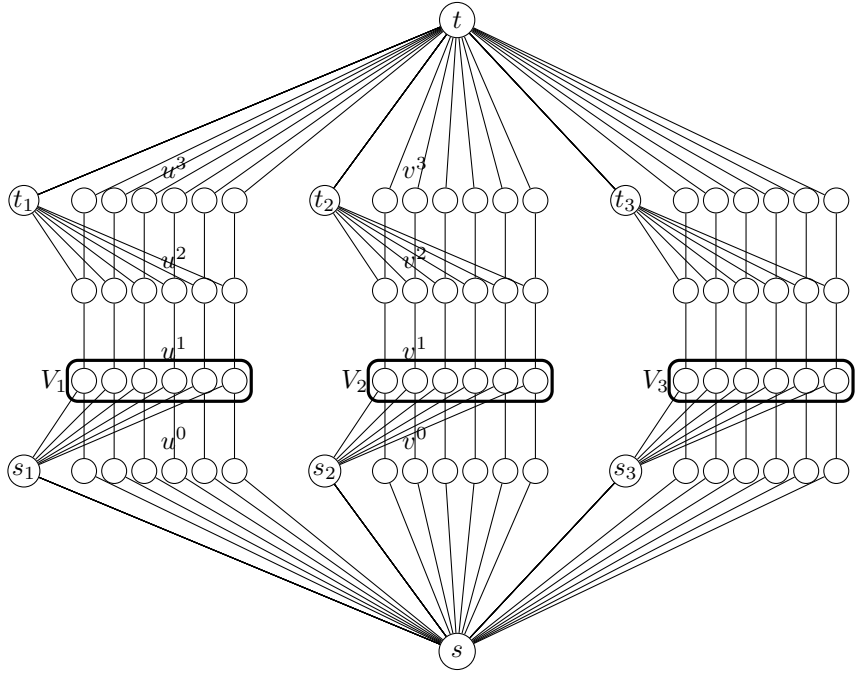

Then, we add two copies of an almost complete binary tree (i.e., all levels except the last are completely filled) with ${k \choose 2}$ leaves, and link both of their roots to a new vertex $r$, itself linked to $t$.
For each leaf of the second to last level, we replace the edge incident to it by two parallel paths of length $2$.
This has the same effect of \emph{weighting} the edge by $2$ (see Figure~\ref{fig:edge-gadget}), so that each leaf is at the \emph{same distance} from the root.
We make the slight abuse of still calling the whole structure a tree.
Within the $b$-th copy of the two binary trees, each leaf represents a distinct set $E_{ij}:=E(V_i,V_j)$ with $i \neq j$, and is denoted by $l_{i,j}(b)$.

For each edge $e=uv \in E$ and $a,b \in [2]$, we add three vertices $w_e(a,b)$, $w_{e,u}(a,b)$ and $w_{e,v}(a,b)$.
We link $w_e(a,b)$ to $w_{e,u}(a,b)$ and $w_{e,v}(a,b)$.
We also link $w_{e,u}(a,b)$ to $u^1$ and $w_{e,v}(a,b)$ to $v^1$.
Finally, for any $b \in [2]$, we link $l_{i,j}(b)$ to all the vertices $w_e(a,b)$ such that $e \in E_{ij}$ and $a \in [2]$ (see Figure~\ref{fig:edge-gadget}).
For each $b \in [2]$, we denote by $T_b$ the $b$-th binary tree together with all the vertices $w_e(a,b)$ and $w_{e,u}(a,b)$, for $a \in [2]$, $e \in E$ and $u \in V$.
We observe that $T_1$ and $T_2$ are disjoint trees.
The tree induced by $T_1$, $T_2$, and $r$ is denoted by $T$, and is seen as rooted in $r$.

This ends the construction of $H$ which has $O(|V(G)|+|E(G)|+k^2)$ vertices.
We ask if there is a winning strategy for White whose depth is at most $q := 2(2k+\lceil \log {k \choose 2} \rceil + 4)-1=O(k)$.
First, if there is a multicolored clique in $G$, we will exhibit a strategy for White that forces a win in at most $q$ moves.
Second, if there is no multicolored clique, we will give a winning strategy for Black. 

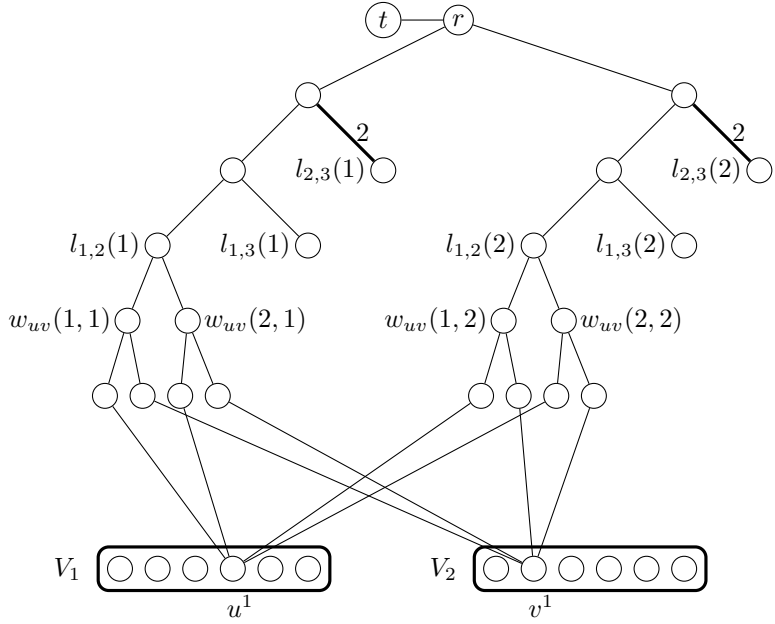
\begin{figure}
\centering
\begin{tikzpicture}[scale=1]
\node[draw,circle,inner sep=-0.25cm] (t) at (0,7) {$t$} ;
\node[draw,circle,inner sep=-0.22cm] (r) at (1,7) {$r$} ;
\draw (r) -- (t) ;

\foreach \b in {1,2} {
  \node[draw,circle] (a1\b) at (-6+5 * \b,6) {} ;
  \node[draw,circle] (l23\b) at (-5+5 * \b,5) {} ;
  \node[draw,circle] (a2\b) at (-7+5 * \b,5) {} ;
  \node[draw,circle] (l12\b) at (-8+5 * \b,4) {} ;
  \node[draw,circle] (l13\b) at (-6+5 * \b,4) {} ;

  \node (tl12\b) at (-8.7+5 * \b,4) {$l_{1,2}(\b)$} ;
  \node (tl13\b) at (-6.7+5 * \b,4) {$l_{1,3}(\b)$} ;
  \node (tl23\b) at (-5.7+5 * \b,5) {$l_{2,3}(\b)$} ;

  \node[draw,circle] (e1\b) at (-8.4+5 * \b,3) {} ;
  \node[draw,circle] (e2\b) at (-7.6+5 * \b,3) {} ;

  \node (te1\b) at (-9.3+5 * \b,3) {$w_{uv}(1,\b)$} ;
  \node (te2\b) at (-6.7+5 * \b,3) {$w_{uv}(2,\b)$} ;

  \node[draw,circle] (u1\b) at (-8.7+5 * \b,2) {} ;
  \node[draw,circle] (v1\b) at (-8.2+5 * \b,2) {} ;
  \node[draw,circle] (u2\b) at (-7.7+5 * \b,2) {} ;
  \node[draw,circle] (v2\b) at (-7.2+5 * \b,2) {} ;

  \draw (u1\b) -- (e1\b) -- (v1\b) ;
  \draw (u2\b) -- (e2\b) -- (v2\b) ;

  \draw (e1\b) -- (l12\b) -- (e2\b) ;

  \draw (r) -- (a1\b) ;
  \draw (a2\b) -- (a1\b) ;
  \draw[very thick] (a1\b) -- node [anchor=west] {$2$} (l23\b) ;
  \draw (l12\b) -- (a2\b) -- (l13\b) ;
}

\foreach \i in {1,2} {
  \foreach \j in {1,...,6} {
    \node[draw,circle] (w\i\j) at (5 * \i - 9 + 0.5 * \j,-0.3) {} ;
  }
  \node[draw,rounded corners,rectangle,very thick,fit=(w\i1) (w\i6)] (v\i) {} ;
  \node (tv\i) at (5 * \i - 9.2, -0.3) {$V_\i$} ;
}

\foreach \b in {1,2} {
  \foreach \a in {1,2} {
    \draw (u\a\b) -- (w14) ;
    \draw (v\a\b) -- (w22) ;
  }
}

\node (ttu1) at (-1.9,-0.8) {$u^1$} ;
\node (ttv1) at (2.1,-0.8) {$v^1$} ;

\end{tikzpicture}
\caption{Edge gadgets for $k=3$ (only one edge $uv \in E_{1,2}$ is represented).}
\label{fig:edge-gadget}
\end{figure}

\paragraph{Normal play}
If there is a multicolored clique $C$ in $G$, then White can win in $q$ moves in $H$.
Here, \emph{move} corresponds to what is usually called a \emph{ply} or a \emph{half-move} in games such as chess (where a move is actually a "move" from White \emph{plus} a "move" from Black).
Let $C=\{c_1,c_2,\ldots,c_k\}$ with $c_i \in V_i$ for each $i \in [k]$.
The first $4k$ moves are the following.
For each $i \in [k]$, White plays in $c_i^1$.
If Black does not answer in $\{s_i,t_i,c_i^0,c_i^2,c_i^3\}$, White plays in $c_i^2$ and wins in four more moves; indeed, $s$ and $c_i^1$ can be connected via $c_i^0$ or $s_i$, while $c_i^2$ and $t_i$ can be connected via $c_i^3$ or $t_i$.
Black plays in $s_i$ and White answers in $c_i^0$.
Then Black is forced to play in $\{c_i^2,c_i^3,t_i\}$.

On move $4k+1$, White plays in $r$ and will play the child in $T$ of their last move in a subtree where Black has \emph{not} played until they play in $l_{i,j}(b)$ for some $i \neq j \in [k]$ and $b \in [2]$.
As $C$ is a clique, there is an edge $e$ between $c_i$ and $c_j$.
By construction, the six vertices $w_e(1,b)$, $w_e(2,b)$, $w_{e,c_i}(1,b)$, $w_{e,c_i}(2,b)$, $w_{e,c_j}(1,b)$, and $w_{e,c_j}(2,b)$ are still empty.
White's next move is $w_e(a,b)$ where $a$ is such that Black has \emph{not} played in $\{w_e(a,b), w_{e,c_i}(a,b), w_{e,c_j}(a,b)\}$.
Then, White's final move is either $w_{e,c_i}(a,b)$ or $w_{e,c_j}(a,b)$, finishing a white path from $t$ all across $T$ to a vertex of the clique to $s$.
Following this strategy, the total number of moves is at most $4k$ in the selector gadgets, plus twice the height of $T$ minus $1$, i.e., $2(\lceil \log {k \choose 2} \rceil + 4)-1$.
This corresponds exactly to $q$ moves. 

\paragraph{Proper defense}
If there is no multicolored clique in $G$, we give a winning strategy for Black: not only White cannot force a win in at most $q$ moves but he cannot win at all.
A \emph{$(y,z)$-cut} $S$ of a graph $J$ is a set of edges of $J$ whose removal makes that $y$ and $z$ are in two different connected components.
We say that $S$ \emph{separates} $y$ and $z$.
The following lemma presents a situation where Black can easily block their opponent.

\begin{lemma}\label{lem:induced-matching}
If there is an $(s,t)$-cut $M$ of $H$ such that no vertex of $M$ contains a white pebble and $M$ is an induced matching, then Black wins, even if it is White to play.
\end{lemma}

\begin{proof}
Whenever White plays a vertex $u$ in $M$, Black answers the other endpoint of the unique edge of $M$ which touches $u$.
\end{proof}

If White plays in some vertex $v^2$, then Black plays in $v^1$ and White has not achieved anything.
If White plays in some vertex $v^0$, $v^3$, $s_i$, or $t_i$, then Black plays $r$ and wins by Lemma~\ref{lem:induced-matching}, since the induced matching $\bigcup_{v \in V} \{v^1v^2\}$ separates $s$ from $t$.
If White plays in a vertex $v^1$ such that $v \in V_i$, Black plays in $s_i$.
Then, White should play in $v^0$ (indeed, it is pointless for White to play in $u^d$ for some $d \in \{0,1\}$ and $u \in V_i$, since Black can answer in $u^{1-d}$) and Black plays in $t_i$.
Now, if White plays in another vertex $u^1$ with $u \in V_i$ (for the same $i$), then Black plays $r$ and wins by Lemma~\ref{lem:induced-matching} with the $(s,t)$-cut $\bigcup_{v \in V \setminus \{u\}} \{v^1v^2\} \cup \{u^2u^3\}$.

So far, we showed that, for each $i \in [k]$, White can play in at most one vertex of the form $v^1$ with $v \in V_i$, and then has to play in $r$.
We distinguish two cases: (a) for each $i \in [k]$, White has played in a vertex $v^1$ with $v \in V_i$ when they play in $r$ (b) there is some $i \in [k]$ such that White has not yet played in a vertex $v^1$ with $v \in V_i$ when they play in $r$.

We start with case (a).
Let $A=\{d_1^1,d_2^1,\ldots,d_k^1\}$ be the vertices which now contain a white pebble, with $d_i^1 \in V_i$ for each $i \in [k]$.
As, $S := \{d_1,d_2,\ldots,d_k\}$ is not a clique in $G$, there are two vertices $d_i$ and $d_j$ such that $d_id_j \notin E$.
Black plays in the root of $T_2$, forcing White to play in the root of $T_1$.
Then, Black will systematically block the path which does \emph{not} go to the leaf $l_{i,j}(1)$.
At this point, $M' := \bigcup_{v \in V \setminus S} \{v^0v^1\} \cup \bigcup_{v \in S} \{v^2v^3\} \cup \bigcup_{a \in [2], xy \in E_{ij}, x \in \{d_i,d_j\}} \{w_{xy}(a,1)w_{xy,x}(a,1)\}$ separates $s$ and $t$.
Also, there is no white pebble in any of the endpoints of $M'$.
Finally, $M'$ is an induced matching.
That observation is based on the key argument that $\bigcup_{a \in [2], xy \in E_{ij}, x \in \{d_i,d_j\}} \{w_{xy}(a,1)w_{xy,x}(a,1)\}$ cannot contain both $w_{xy}(a,1)w_{xy,x}(a,1)$ and $w_{xy}(a,1)w_{xy,y}(a,1)$ since that would imply that $\{x,y\}=\{d_i,d_j\}$, contradicting that $d_id_j \notin E$.
Therefore, Black wins by Lemma~\ref{lem:induced-matching}.


Now, we consider case (b).
Black plays $s_i$ where $i$ is such that White has not played in any vertex $v^1$ with $v \in V_i$ yet.
Now, whenever White plays in $u^1 \in V_i$, Black plays in $u^0$.
This way, White cannot connect any vertices of $V_i$ to $s$.
When White goes back to playing in $T$, Black forces him to the leaves $l_{i,j}(1)$ (and/or $l_{i,j}(2)$) for some arbitrary $j \in [k]$ and wins by Lemma~\ref{lem:induced-matching} similarly to case (a).
\end{proof}

Under the Exponential Time Hypothesis (ETH) that asserts that \textsc{3-sat} has no subexponential algorithm, \textsc{multicolored clique} cannot be solved in time $f(k)(n+m)^{o(k)}$ for any computable function $f$ \cite{Chen06}.
As our reduction linearly preserves the parameter, we have shown that \sGHex cannot be solved in the same running time $f(k)(n+m)^{o(k)}$ unless the ETH fails, where $n$ and $m$ are the number of vertices and edges of the graph.
A simple procedure solving \sGHex in time $O^*(n^k)$ (where $O^*$ suppresses polynomial factors) would consist of expanding all the continuations of length up to $k$ and decide if Player $1$ wins in this game subtree.

A class $\mathcal C$ of graphs has \emph{bounded local treewidth} if there is a function $F_{\mathcal C}$ such that for any $G \in \mathcal C$, $\text{tw}(G) \leqslant F_{\mathcal C}(\text{diam}(G))$ where $\text{tw}(G)$ is the treewidth of $G$ and $\text{diam}(G)$ its diameter.

\begin{theorem}
\SGHex is \fpt in every class of graphs with bounded local treewidth and closed by edge contractions.
\end{theorem}

\begin{proof}
As we said above, we can contract the maximal connected components of vertices containing white pebbles (\emph{white connected components}) into a single vertex containing a white pebble.
Indeed, the condition \emph{there is an $(s,t)$-path in which each vertex contains a white pebble} is equivalent to \emph{$s$ and $t$ are in the same supervertex in the graph obtained by contracting every edge whose two endpoints contain a white pebble}.
Thus, we can now assume that, initially, the white connected components have all size $1$.
Let $W$ be the set of vertices containing initially a white pebble in this new graph.
After at most $k$ moves, the length of a potential $(s,t)$-path with white pebbles is bounded by $2\lceil \frac{k}{2} \rceil + 1 \leqslant k+2$.
Indeed, White plays at most $\lceil \frac{k}{2} \rceil$ new pebbles and the path can contain at most $\lceil \frac{k}{2} \rceil + 1$ initially placed white pebbles (every other two).
It implies that the problem is expressible in first-order logic by the formula $\bigvee_{k' \leqslant k}[\exists v_1 \forall v_2 \exists v_3 \forall v_4 \ldots \exists v_{k'} \bigwedge_{i \neq j \in [k']} v_i \neq v_j \land \bigvee_{l \leqslant k+2}[\exists u_1, u_2, \ldots, u_l \in W \cup \{v_1, v_3, \ldots, v_{k'}\}~u_1=s \land u_l=t \land E(u_1,u_2) \land E(u_2,u_3) \land \ldots \land E(u_{l-1},u_l)]]$.
Frick and Grohe showed that any graph problem expressible in first-order logic can be solved in linear time in classes of graph with local bounded treewidth \cite{Frick01}.
Yet, this algorithm is far from practical since its dependence in the formula is an exponential tower of height the number of quantifier alternations, so in our case height $\Theta(k)$.
\end{proof}

\begin{corollary}
\SGHex is \fpt in planar graphs.
\end{corollary}

\begin{proof}
Planar graphs are closed by edge contractions and have bounded local treewidth \cite{Robertson84}.
\end{proof}

\begin{corollary}
\textsc{Short hex} is \fpt.
\end{corollary}

\begin{proof}
The underlying graph of \gamename{hex} is planar.
\end{proof}

\section{Ultra-weak Solutions of Connection games}
\label{sec:ultra-weak}
In \citeauthor{Allis1994}'s game solving hierarchy~\cite[Section 1.5]{Allis1994}, a game is said to be \emph{ultra-weakly} solved if the game-theoretic value of the initial position has been determined.
\Hex is unusual among strategy games in that an ultra-weak solution is known for any board size.

In this section, we first recall how \hex can be ultra-weakly solved to prove that the starting position is a theoretical first-player win.
We then consider \twixt, \Hav, and \slither, and show that although some of the arguments carry over, it is impossible to ultra-weakly solve these games with the same method.

\subsection{Hex}
\label{subsec:ultra-weak-hex}

When John Nash discovered/invented the game \hex, one of his motivations was to find a non-trivial game with a non-constructive proof that the first player has a winning strategy in the initial position.
The proof has two steps, prove that the initial position is not a draw, then prove that the initial position is not a second-player win.

The outcome of the initial position is not a game-theoretical draw because no \hex position can end in a draw.
This is a well-known fact but the proof is non-trivial.
We refer to \citet{Bogomolny2016} for a discussion of the issue and further references.
\begin{lemma}
  \label{lem:hex-no-draws}
  Draws cannot occur in \gamename{hex}.
\end{lemma}

To derive the second step, Nash developed a \emph{strategy-stealing argument} which can be summed up as follows~\cite{vanRijswijck2006}.
Suppose for a contradiction that the second player has a winning strategy $\sigma$ in the initial position.
Then a winning strategy for the first player can be obtained as: start with a random move, then apply $\sigma$ pretending that the initial move did not occur.
If $\sigma$ ever recommends to play on the location of the initial move, then play another random move and carry on.
Given that having an additional random stone on the board cannot hurt the first player, then we have developed a winning strategy for the first player too.
Since the second and first player cannot both have a winning strategy at the same time, we can conclude that our hypothesis does not hold.
Therefore the initial \hex position is not a second-player win when the swap rule is not used.
\begin{lemma}
  \label{lem:hex-not-second}
  \Hex is not a theoretical second-player win.
\end{lemma}

The strategy-stealing argument applies to \hex for two reasons.
First, it is a symmetrical game, \emph{i.e.}, for any position, there is a map from moves that are legal for the second player to moves that are legal for the first player.
Second, there are no \emph{zugzwangs}: having an additional move never hurts.

This leads to the ultra-weak solution for \hex.
\begin{theorem}
  On any board size, the initial position in \hex is a theoretical first-player win.
\end{theorem}
\begin{proof}
  \Hex is a turn-taking finite two-person game of perfect information with no chance.
  Zermelo's Theorem states that any such game is either a first-player win, a second-player win, or a draw.
  Lemma~\ref{lem:hex-no-draws} and~\ref{lem:hex-not-second} rule out two possibilities and \hex can thus only be a first-player win.
\end{proof}

\subsection{\Twixt and \Hav}

\Twixt and \Hav admit draws in theory but they occur very rarely in practice.
The fact is relatively well-known in the game community, for instance \citet[Section 3.4.5]{Ewalds2012} states that draws are possible for \Hav sizes above 3~\cite[Section 3.4.5]{Ewalds2012}.
Yet, to the best of our knowledge, a proof for arbitrary sizes has not appeared in the related literature.
We give here a general construction for both games illustrated in Figure~\ref{fig:draws}.

\begin{proposition}
  It is possible to reach a drawn position from an empty \twixt board of any size, and from an empty \Hav board of any size strictly larger than 3.
\end{proposition}
\begin{proof}
  For \twixt, any position on size 3 and 4 is a draw, and an example of draw for size 5 is given in Figure~\ref{fig:twixt-draw5}.
  For larger sizes, if both players keep placing pegs in a checkered layout, no connection can ever go past two rows away from any border.
  This construction ensures a draw on board from size 6.
  Figure~\ref{fig:twixt-draw6} and Figure~\ref{fig:twixt-draw12} display the pattern for size 6 and 12 respectively.

  In case of \Hav, if both players keep placing pegs in a striped layout, no winning condition is ever met.
  Figure~\ref{fig:hav-draw4} and \ref{fig:hav-draw6} displays the pattern for size 4 and 6.
  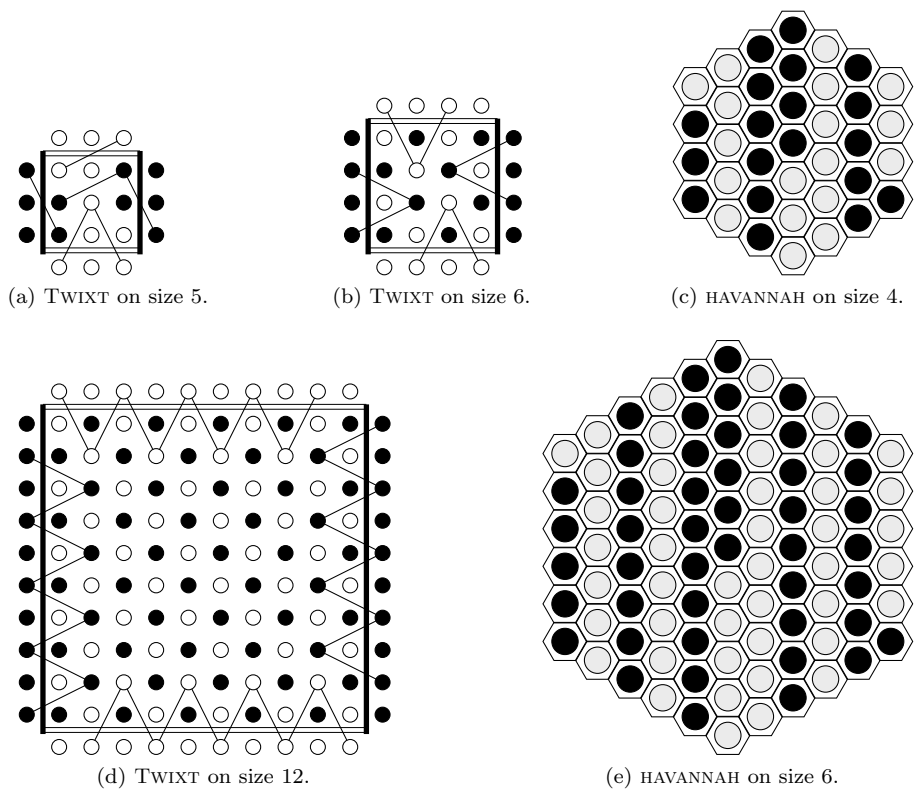
\begin{figure}
    \centering
    \subfloat[\Twixt on size 5.]{\label{fig:twixt-draw5} \begin{tikzpicture}[scale=0.43]
  \begin{scope}[xshift=-0cm,yshift=-0cm,yscale=-1]
    \twboard{5}{5}

    \draw (1, 2) -- (2, 4);
    \draw (2, 2) -- (4, 1);
    \draw (2, 3) -- (4, 2) -- (5, 4);
    \draw (2, 5) -- (3, 3) -- (4, 5);

    \foreach \i in {2,...,4} {
      \node [tw-white] at (\i,1) {}; \node [tw-white] at (\i,5) {}; \node [tw-black] at (1,\i) {}; \node [tw-black] at (5,\i) {};
    }

    \node [tw-black] at (2, 3) {};
    \node [tw-black] at (2, 4) {};
    \node [tw-black] at (4, 2) {};
    \node [tw-black] at (4, 3) {};
    \node [tw-white] at (2, 2) {};
    \node [tw-white] at (3, 2) {};
    \node [tw-white] at (3, 3) {};
    \node [tw-white] at (3, 4) {};
    \node [tw-white] at (4, 4) {};

  \end{scope}
\end{tikzpicture}
\hspace{5mm}}
    \hfill
    \subfloat[\Twixt on size 6.]{\label{fig:twixt-draw6} \begin{tikzpicture}[scale=0.43]
  \begin{scope}[xshift=-0cm,yshift=-0cm,yscale=-1]
    \twboard{6}{6}

    \draw (2, 1) -- (3, 3) -- (4, 1);
    \draw (6, 2) -- (4, 3) -- (6, 4);
    \draw (5, 6) -- (4, 4) -- (3, 6);
    \draw (1, 5) -- (3, 4) -- (1, 3);

    \foreach \i in {2,...,5} {
      \node [tw-white] at (\i,1) {} ;
      \node [tw-white] at (\i,6) {} ;
      \node [tw-black] at (1,\i) {} ;
      \node [tw-black] at (6,\i) {} ;
    }

    \foreach \i in {1,...,2} {
      \foreach \j in {1,...,2} {
        \pgfmathtruncatemacro{\ii}{2 * \i}
        \pgfmathtruncatemacro{\jj}{2 * \j}
        \node [tw-white] at (\ii,\jj) {} ;
      }
    }
    \foreach \i in {1,...,2} {
      \foreach \j in {1,...,2} {
        \pgfmathtruncatemacro{\ii}{2 * \i + 1}
        \pgfmathtruncatemacro{\jj}{2 * \j + 1}
        \node [tw-white] at (\ii,\jj) {} ;
      }
    }
    \foreach \i in {1,...,2} {
      \foreach \j in {1,...,2} {
        \pgfmathtruncatemacro{\ii}{2 * \i + 1}
        \pgfmathtruncatemacro{\jj}{2 * \j}
        \node [tw-black] at (\ii,\jj) {} ;
        \node [tw-black] at (\jj,\ii) {} ;
      }
    }
  \end{scope}
\end{tikzpicture}
}
    \hfill
    \subfloat[\Hav on size 4.]{\label{fig:hav-draw4} \newcommand{\scala}{0.5}
\begin{tikzpicture}[>=stealth',scale=\scala,every node/.style={scale=2*\scala}]

  \havcoordinate{0}{18}{0}{18}
  \foreach \j in {5,...,8} { \node[hav-empty] at  (1-\j)  {};}
  \foreach \j in {4,...,8} { \node[hav-empty] at (1--\j)  {};}
  \foreach \j in {4,...,9} { \node[hav-empty] at  (2-\j)  {};}
  \foreach \j in {3,...,9} { \node[hav-empty] at (2--\j)  {};}
  \foreach \j in {4,...,9} { \node[hav-empty] at  (3-\j)  {};}
  \foreach \j in {4,...,8} { \node[hav-empty] at (3--\j)  {};}
  \foreach \j in {5,...,8} { \node[hav-empty] at  (4-\j)  {};}

  \foreach \j in {5,...,7} { \node[hav-black] at  (1-\j)  {};} \node[hav-white] at  (1-8)  {};
  \foreach \j in {4,...,8} { \node[hav-white] at (1--\j)  {};}
  \foreach \j in {4,...,9} { \node[hav-black] at  (2-\j)  {};}
  \foreach \j in {3,...,5} { \node[hav-white] at (2--\j)  {};}
  \foreach \j in {7,...,9} { \node[hav-black] at (2--\j)  {};} \node[hav-black] at  (2--6)  {};
  \foreach \j in {4,...,9} { \node[hav-white] at  (3-\j)  {};}
  \foreach \j in {4,...,8} { \node[hav-black] at (3--\j)  {};}
  \foreach \j in {6,...,8} { \node[hav-white] at  (4-\j)  {};} \node[hav-black] at  (4-5)  {};
\end{tikzpicture}
} \\
    \subfloat[\Twixt on size 12.]{\label{fig:twixt-draw12} \begin{tikzpicture}[scale=0.43]
  \begin{scope}[xshift=-0cm,yshift=-0cm,yscale=-1]
    \twboard{12}{12}

    \foreach \i in {1,...,4} {
      \pgfmathtruncatemacro{\ii}{2 * \i}
      \pgfmathtruncatemacro{\iii}{2 * \i + 1}
      \pgfmathtruncatemacro{\iv}{2 * \i + 2}
      \draw (\ii,  1) -- (\iii,  3) -- (\iv,  1);
      \draw (12, \ii) -- (10, \iii) -- (12, \iv);

      \pgfmathtruncatemacro{\jj}{13 - 2 * \i}
      \pgfmathtruncatemacro{\jjj}{13 - 2 * \i - 1}
      \pgfmathtruncatemacro{\jv}{13 - 2 * \i - 2}
      \draw (\jj, 12) -- (\jjj, 10) -- (\jv, 12);
      \draw ( 1, \jj) -- ( 3, \jjj) -- ( 1, \jv);
    }
    \foreach \i in {2,...,11} {
      \node [tw-white] at ( \i, 1) {} ;
      \node [tw-white] at ( \i,12) {} ;
      \node [tw-black] at ( 1, \i) {} ;
      \node [tw-black] at (12, \i) {} ;
    }

    \foreach \i in {1,...,5} {
      \foreach \j in {1,...,5} {
        \pgfmathtruncatemacro{\ii}{2 * \i}
        \pgfmathtruncatemacro{\jj}{2 * \j}
        \node [tw-white] at (\ii,\jj) {} ;
      }
    }
    \foreach \i in {1,...,5} {
      \foreach \j in {1,...,5} {
        \pgfmathtruncatemacro{\ii}{2 * \i + 1}
        \pgfmathtruncatemacro{\jj}{2 * \j + 1}
        \node [tw-white] at (\ii,\jj) {} ;
      }
    }
    \foreach \i in {1,...,5} {
      \foreach \j in {1,...,5} {
        \pgfmathtruncatemacro{\ii}{2 * \i + 1}
        \pgfmathtruncatemacro{\jj}{2 * \j}
        \node [tw-black] at (\ii,\jj) {} ;
        \node [tw-black] at (\jj,\ii) {} ;
      }
    }
  \end{scope}
\end{tikzpicture}
}
    \hfill
    \subfloat[\Hav on size 6.]{\label{fig:hav-draw6} \newcommand{\scala}{0.5}
\begin{tikzpicture}[>=stealth',scale=\scala,every node/.style={scale=2*\scala}]

  \havcoordinate{0}{18}{0}{18}
  \foreach \j in {4,...,9}  { \node[hav-empty] at  (0-\j)  {} ;}
  \foreach \j in {3,...,9}  { \node[hav-empty] at (0--\j)  {} ;}
  \foreach \j in {3,...,10} { \node[hav-empty] at  (1-\j)  {} ;}
  \foreach \j in {2,...,10} { \node[hav-empty] at (1--\j)  {} ;}
  \foreach \j in {2,...,11} { \node[hav-empty] at  (2-\j)  {} ;}
  \foreach \j in {1,...,11} { \node[hav-empty] at (2--\j)  {} ;}
  \foreach \j in {2,...,11} { \node[hav-empty] at  (3-\j)  {} ;}
  \foreach \j in {2,...,10} { \node[hav-empty] at (3--\j)  {} ;}
  \foreach \j in {3,...,10} { \node[hav-empty] at  (4-\j)  {} ;}
  \foreach \j in {3,...,9}  { \node[hav-empty] at (4--\j)  {} ;}
  \foreach \j in {4,...,9}  { \node[hav-empty] at  (5-\j)  {} ;}

  \foreach \j in {4,...,8}  { \node[hav-black] at  (0-\j)  {} ;} \node[hav-white] at  (0-9)  {} ;
  \foreach \j in {3,...,9}  { \node[hav-white] at (0--\j)  {} ;}
  \foreach \j in {3,...,10} { \node[hav-black] at  (1-\j)  {} ;}
  \foreach \j in {2,...,10} { \node[hav-white] at (1--\j)  {} ;}
  \foreach \j in {2,...,11} { \node[hav-black] at  (2-\j)  {} ;}
  \foreach \j in {1,...,5}  { \node[hav-white] at (2--\j)  {} ;}
  \foreach \j in {7,...,11} { \node[hav-black] at (2--\j)  {} ;} \node[hav-black] at  (2--6)  {} ;
  \foreach \j in {2,...,11} { \node[hav-white] at  (3-\j)  {} ;}
  \foreach \j in {2,...,10} { \node[hav-black] at (3--\j)  {} ;}
  \foreach \j in {3,...,10} { \node[hav-white] at  (4-\j)  {} ;}
  \foreach \j in {3,...,9}  { \node[hav-black] at (4--\j)  {} ;}
  \foreach \j in {5,...,9}  { \node[hav-white] at  (5-\j)  {} ;} \node[hav-black] at  (5-4)  {} ;
\end{tikzpicture}
}
    \caption{Drawn \twixt and \Hav positions.
    The \twixt positions would also be drawn under the \gamename{twixtpp} variant.}
    \label{fig:draws}
  \end{figure}
\end{proof}

\begin{lemma}
  \Twixt (resp.~\Hav) is not a theoretical second-player win.
\end{lemma}
\begin{proof}
  Zugzwangs are not possible in \twixt (resp.~\Hav), and the game is symmetrical, therefore, the strategy-stealing argument applies to \twixt (resp.~\Hav).
\end{proof}

Although the initial positions of \twixt or \Hav are not second-player wins, draws are possible for boards sufficiently large, and so we cannot conclude that either is a first-player win.

\subsection{Slither}
We now present some observations on and properties of \slither.
Some of the observations made in this section have independently been pointed out earlier in the abstract game community, especially on BoardGameGeek.\footnote{\url{http://boardgamegeek.com/thread/692652/what-if-there-no-legal-move}}
However, we prove for the first time that the standard \slither variant cannot end in a draw, thereby settling an open-problem often raised in this community.

The concept of \emph{zugzwang} appears in \gamename{chess} and denotes a position in which the current player has no desirable move and would rather pass and have the opponent act.
A \emph{mutual zugzwang} is a position in which both players would rather have the opponent play.
Although zugzwangs are virtually unheard of in typical connection games, where additional moves can never hurt you, things are different in \slither.

\begin{proposition}
  Zugzwangs and mutual zugzwangs can occur in \slither.
\end{proposition}
\begin{proof}
  In Figure~\ref{fig:zugzwang}, if it is White (resp. Black) to play, only one move is available, moving stone on C2 (resp. C4) to $a$ and placing a stone on $c$ or equivalently moving to $c$ and placing on $a$.
  Then Black (resp. White) wins by placing a stone on C2 (resp. moving stone B5 to C4 and placing a stone on $b$).
\end{proof}
\begin{figure}[t]
  \centering
  \subfloat[Mutual zugzwang: no good move for Black and no good move for White.]{
    \label{fig:zugzwang}
    \hspace{6mm}
    \begin{tikzpicture}
      \sboard{5}{5}
      \sedges{5}{5}
      \sposition{1-1,1-4,2-4,2-5,3-1,3-2,4-4,4-5,5-1,5-3,5-4}{1-2,1-3,1-5,2-1,2-2,3-4,3-5,4-1,4-2,5-2,5-5}

      \node (A) at (0.5,-0.1) {A};
      \node (B) at (1,  -0.1) {B};
      \node (C) at (1.5,-0.1) {C};
      \node (D) at (2,  -0.1) {D};
      \node (E) at (2.5,-0.1) {E};

      \node (1) at (0,0.5) {1};
      \node (2) at (0,1) {2};
      \node (3) at (0,1.5) {3};
      \node (4) at (0,2) {4};
      \node (5) at (0,2.5) {5};

      \node (u) at (2-3) {$a$};
      \node (u) at (3-3) {$b$};
      \node (u) at (4-3) {$c$};
    \end{tikzpicture}
    \hspace{8mm}
  }
  \hfill
  \subfloat[No move allowed for Black, only White has legal moves.]{
    \label{fig:nomove}
    \hspace{6mm}
    \begin{tikzpicture}
      \sboard{4}{5}
      \sedges{4}{5}

      \sposition{1-2,1-3,2-1,2-2,2-5,3-1,3-4,3-5,4-3,4-4}{1-1,1-4,1-5,2-4,3-2,4-1,4-2,4-5}

      \node (A) at (0.5,-0.1) {A};
      \node (B) at (1,  -0.1) {B};
      \node (C) at (1.5,-0.1) {C};
      \node (D) at (2,  -0.1) {D};

      \node (1) at (0,0.5) {1};
      \node (2) at (0,1) {2};
      \node (3) at (0,1.5) {3};
      \node (4) at (0,2) {4};
      \node (5) at (0,2.5) {5};
    \end{tikzpicture}
    \hspace{4mm}
  }
  \caption{\gamename{slither} positions with a shortage of moves.}
\end{figure}
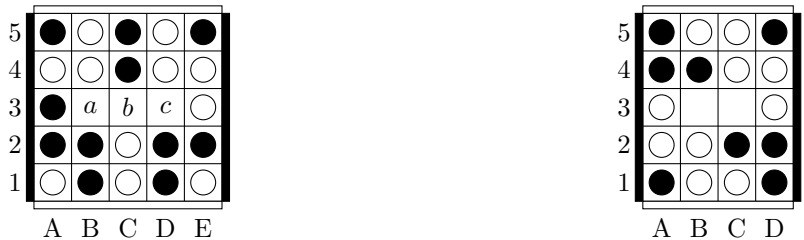

As we have seen, the strategy-stealing argument can be applied to many other games including \twixt, \Hav, and games of the \gamename{connect}($m, n, k$) family~\cite{HsiehT2007}.
Unfortunately, the argument cannot be applied to \slither.

\begin{proposition}
  Nash's strategy-stealing argument does not apply in \slither.
\end{proposition}
\begin{proof}
  Consider the White zugzwang position in Figure~\ref{fig:zugzwang}.
  Had the A4 square not been occupied with a White stone, White would have a winning move: move B4 to $b$ and place a stone on $c$.
  Since there are positions where having one too many stones on the board can make a player lose the game, Nash's strategy-stealing argument does not apply.
\end{proof}

Therefore, there is no theoretical indication yet that \slither is \emph{not} a second-player win on an empty board.
However, in practice it is a huge advantage to play first, so much that if the swap rule is used, it is recommended to swap no matter where the first move is played, including corner locations.
The \slither-specific intuition behind this practical advice is that the game is dynamic and a player can bring back a stone from a corner towards the center, moving it closer every turn.

Draws would occur in \slither if there were positions where no player has a legal move and yet no player connects their own sides.
The \slither community had indeed identified non-terminal positions in which one of the players had no legal moves, but the opponent always had at least one move possible.
In that case, the player with no legal moves would simply skip their turn.

\begin{proposition}
  There exist positions in which a player has no legal moves.
\end{proposition}
\begin{proof}
  For instance, Black has no legal move in the position Figure~\ref{fig:nomove}.
\end{proof}

The designer of \slither has long claimed that the game did not admit any draw.
Since there were no formal proofs, many members of the community were left unconvinced and attempted to find counter-examples.
They would submit positions on forums dedicated to \slither, and Corey Clark or some \slither players would point out that the counter-example was not valid, usually because the diagonal rule was not respected or because a legal move actually existed for one of the players.
Before settling this question formally, let us point out that draws can actually occur when the board topology is not restricted to be rectangular.

\begin{proposition}
  Draws are possible when \slither is played on a cylinder or on a torus.
\end{proposition}
\begin{proof}
  In Figure~\ref{fig:cylinder} (resp.~Figure~\ref{fig:torus}), if black (resp.~black and white) sides are connected, then both players have no legal moves.
\end{proof}

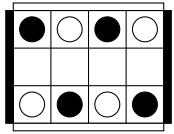
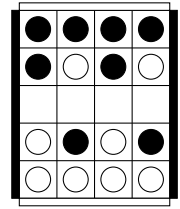
\begin{figure}
  \centering
  \subfloat[Cylinder board: left and right edges are connected.]{
    \hspace{7mm}
    \begin{tikzpicture}\label{fig:cylinder}
      \sboard{4}{3}
      \sedges{4}{3}
      \sposition{1-1,2-3,3-1,4-3}{1-3,2-1,3-3,4-1}
    \end{tikzpicture}
    \hspace{7mm}
  }
  \hfill
  \subfloat[Torus board: left and right edges are connected and top and bottom edges are connected.]{
    \hspace{9mm}
    \begin{tikzpicture}\label{fig:torus}
      \sboard{4}{5}
      \sedges{4}{5}
      \sposition{1-1,1-2,2-1,2-4,3-1,3-2,4-1,4-4}{1-4,1-5,2-2,2-5,3-4,3-5,4-2,4-5}
    \end{tikzpicture}
    \hspace{5mm}
  }
  \caption{Drawn \slither positions on non-planar boards.}
  \label{fig:cylindertorus}
\end{figure}

Draws are possible on some exotic boards.
So, there is probably no fundamental reasons why the following result is true, and the proof, which we defer to the appendix, consists of a large case analysis on the consequences of forbidding diagonal configurations and the possibility of moving stones.

\begin{theorem}
  Draws are not possible in \slither on rectangular boards.
\end{theorem}
\begin{proof}
  This theorem relies on two properties of \slither.
  First, if a \slither board is filled, then at least one player has a winning group.
  Second, if a board is not filled, then at least one player has a legal move.
  As as a result, the game always continues, possibly with a player passing, until one player has a winning group
  Since each move adds a stone to an empty intersection, these properties ensure that the game is bound to end after a finite number of turns with a winner.
  The Appendix proves the two essential properties mentioned above.
\end{proof}





\section{Discussion and Conclusion}
Unlike \gamename{chess}, \gamename{checkers}, and \gamename{go}, the game of \hex is not played on a grid board but uses a hexagonal paving.
The simplest connection game on a grid, using \hex rules on a checkered board where each cell is connected to the 4 vertically and horizontally adjacent ones, easily leads to draws in almost all cases and is not interesting for players.
Designing an interesting connection game played on the usual grid board is challenging and \gamename{twixt} and \slither can be seen as two different solutions.
Other solutions include Jo{\~a}o Pedro Neto's \gamename{gonnect} and Mark Steere's \gamename{crossway}, and \citeauthor{Browne2005}'s book provides many more examples~\cite{Browne2005}.
Each comes with its own set of mechanics and trade-offs.

\gamename{Twixt} eschews adjacent connections in favor of Knight's move connections.
In a sense, this changes the topology of the board fundamentally as, contrary to \hex and most other connection games, the graph of all possible direct connections is not planar anymore.
Another consequence is that \gamename{twixt} allows draws, but this difference is mainly of theoretical interest as draws remain extremely rare in competitive play.

\Sl keeps the grid topology intact but addresses draws by adding forbidden patterns and stone movement.
Forbidding patterns alone would not be sufficient to prevent draws: for instance Figure~\ref{fig:zugzwang},~\ref{fig:nomove}, and~\ref{fig:slither-hex} would all be drawn if it was not for the possibility of moving existing pieces.
Not only does moving pieces significantly changes the dynamics of the game, but it also adds the theoretical possibility of zugzwangs.
Again, this remains mainly of theoretical interest as zugzwangs virtually never seem to happen in practice.

\gamename{Havannah} keeps the convenient hexagonal paving of \hex but explores a new dimension in the realm of connection game design by adding the \emph{ring} winning condition.
A board could in principle support a ring or more from both players, except for the fact that the game ends as soon as the first one is completed.
This leads to winning races where both players attempt to satisfy their own ring patterns, possibly non-intersecting, while disrupting the opponent's strategy by more immediate threats.

The hexagonal versus grid paving issue is of some importance in game design.
It could seem, at first sight, that this difference would make reductions from \hex easier for \Hav than for \gamename{twixt} and \slither.
On the contrary, we only presented a \hex-based reduction for \gamename{twixt} and \slither, and in both reductions, the size of the resulting game is only linearly larger than the size of the input instance.

As a matter of fact, the ring winning condition in \Hav makes obtaining a \pspace-hardness proof by direct reduction from \hex quite difficult.
Attempts by other authors and ourselves at finding a \Hex-based reduction centered around the bridge or the fork winning conditions proved unfruitful.\footnote{Fabien and Olivier Teytaud, personal communication and~\cite{TeytaudT2010}.}
Instead, disregarding bridges and forks, a sequence of ring threats allowed us to reduce from \gamename{generalized geography}, the problem used in Reisch's proof for \hex.
On the other hand, if the \Hav rules were changed to drop the ring winning condition, a direct simulation of \hex on the \Hav board would be possible, just like it is for \gamename{y}, and \pspace-completeness of this new game would be straightforward.



In terms of future work, settling the complexity of other notable connection games remains a natural direction.
In particular, \gamename{gonnect} and \gamename{lines of action} are good candidates.
In \gamename{lines of action}, each player starts with two groups of pieces and tries to connect all their pieces by moving these pieces and possibly capturing opponent pieces~\cite{WinandsBS2010}.
While the goal of \gamename{lines of action} clearly makes it a connection game, the mechanics distinguishes it from more classical connection games as no pieces are added to the board and existing pieces can be moved or removed.
\gamename{Gonnect} is best described as a mix between \hex and \gamename{go}.
Each player is trying to connect the opposite sides of a grid board, but captures are possible following the rules of \gamename{go}.
Both these games are quite different from the usual connection games and from the games studied in this paper in that the game length is not polynomially bounded.
As a result, these games may very well not be \pspace-complete.

Future work could also examine \Hav, \slither, and \gamename{twixt} under a closer eye with the goal of obtaining theoretical results of practical relevance.
For instance, a seemingly reasonable heuristic in \slither consists of playing the first move of a shortest sequence leading to victory, assuming the opponent passes.
We leave as an open question whether such a shortest sequence can be efficiently computed.
In fact, determining if such a sequence exists at all might not necessarily be easy and is reminiscent of \citet{MazzoniW1997} \np-completeness result on single-player \gamename{twixt}.
Another possible bridge towards actual game playing and solving is provided by parameterized complexity.
Our \fpt result on \SHex indicates that solving algorithms need not depend exponentially on the board size, but only on the search depth.
Could this result be turned into a practical algorithm for solving \hex endgames on large board sizes?

We find it remarkable that both \gamename{twixt} and \slither \pspace-hardness proofs are direct reductions from \hex.
It is quite rare indeed that an actually played game constitutes the basis of a computational complexity reduction.
This is a new testimony to the special place that \hex occupies among strategy games, it combines being a relatively popular abstract game and a theoretical problem fundamental enough to reduce from.

\section*{Acknowledgments}
The authors wish to thank D\'aniel Marx and Radu Curticapean for useful discussions and pointing out the paper \cite{Frick01}.
The first author was supported by the ERC grant PARAMTIGHT: ``Parameterized complexity and the search for tight complexity results'' (no.\,280152).
The third author was supported by the Australian Research Council (project DE\,150101351).

\section*{References}
\bibliographystyle{elsarticle-num-names}
\bibliography{final}

\newpage
\appendix
\setcounter{secnumdepth}{4}

\section{Draws are impossible in \gamename{slither}}
\label{app:nodraws}
Stating that \gamename{slither} does not feature any draw actually corresponds to the following three more elementary statements: each filled \gamename{slither} board has winning groups for at least one player, each filled \gamename{slither} board has winning groups for no more than one player, and each non-filled \gamename{slither} board has at least one legal move for one of the two players.
The first two statements can be obtained rather directly from the equivalent ones in \gamename{hex}, thanks to the diagonal rule.
The third statement is much more involved and requires a careful case analysis on non-filled boards.

\begin{lemma}
  If a \gamename{slither} board is filled, then exactly one of the two players wins.
\end{lemma}
\begin{proof}
  Recall that \gamename{hex} can be played on a rectangular board provided we add a link between each pair of king-adjacent squares along one specified diagonal direction, as in Figure~\ref{fig:rect-hex}.
  The forbidden configuration rule ensures that this king-adjacent diagonal connection is respected in \gamename{slither}, although it is indirect.
  Therefore, any filled $m\times n$ \gamename{slither} board can be mapped onto an equivalent $m\times n$ \gamename{hex} board such that any pair of \gamename{slither} squares is connected if and only if the corresponding pair of \gamename{hex} cells is connected.
  Since any filled \gamename{hex} board is won by exactly one player, we have the desired \gamename{slither} result.

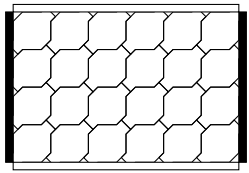
\begin{figure}
  \centering
  \begin{tikzpicture}
    \hboard{6}{4} \sposition{}{}
    \sedges{6}{4}
  \end{tikzpicture}
  \caption{$6 \times 4$ \gamename{hex} board represented on a rectangular \gamename{slither} board.}
  \label{fig:rect-hex}
\end{figure}

\end{proof}

\begin{theorem}
  On a rectangular board, as long as there is at least one empty square, at least one of the two players has a legal move.
\end{theorem}
\begin{proof}
  We adopt the following proof technique.\footnote{No part of the argument will rely upon the color of the board edge.}
We assume for a contradiction that we have a non-filled position with no legal moves for any players.
  We start from an empty square and make deductions concerning its surrounding so as to constrain the occupancy of the nearby squares.
  Each constraint is deduced based on the established occupancies and from the no legal moves assumption or from the diagonal rule.
  We may perform a split case analysis on squares that are not constrained enough to have a definite status.
  In each case, however, we finally arrive at a position with a legal move that cannot be prevented by adding any further constraints.

  As we add constraints to forbid legal moves by either player, we liberally extend the size of the pattern around the empty square.
  If any such extension was not possible because we would have reached the limit of the board, then it would not be possible to forbid the desired legal move and our case would be proved.
  We can therefore disregard the possibility of inadvertently meeting an edge of the board as we extend our patterns, at least for the sake of this argument.

  In addition to the regular three types of squares, \begin{tikzpicture}[scale=0.6, every node/.style={transform shape}] \sboard{1}{1} \sposition{1-1}{} \end{tikzpicture} white stone, \begin{tikzpicture}[scale=0.6, every node/.style={transform shape}] \sboard{1}{1} \sposition{}{1-1} \end{tikzpicture} black stone, and \begin{tikzpicture}[scale=0.6] \sboard{1}{1} \end{tikzpicture} empty, we add the following ones: \begin{tikzpicture}[scale=0.6] \sboard{1}{1} \possb{1}{1} \possw{1}{1} \end{tikzpicture} \emph{no constraints yet}, \begin{tikzpicture}[scale=0.6] \sboard{1}{1} \possb{1}{1} \end{tikzpicture} \emph{cannot hold a white stone}, and \begin{tikzpicture}[scale=0.6] \sboard{1}{1} \possw{1}{1} \end{tikzpicture} \emph{cannot hold a black stone}.

  Consider a rectangular board.
  If there is at least an empty square on the board, then there is at least an empty square $s$ such that one of the following 4 conditions on the bottom and left neighbors of $s$ is satisfied.
  Either $s$ is in the bottom-left corner (Figure~\ref{fig:corner}), or $s$ is on the bottom edge and its left neighbor is occupied (Figure~\ref{fig:edge}), or $s$ has three neighboring stones of different colors (Figure~\ref{fig:21}), or these stones are all of the same color (Figure~\ref{fig:case3}).
  We can assume w.l.o.g.~that a majority of the bottom-left neighboring stones are white.

  The first two cases are treated in Section~\ref{corner-edge}, the third case is treated in Section~\ref{two-one}, and the last case in Section~\ref{three-nil}.
  In each case, we arrive to the conclusion that at least one player can move.
\end{proof}

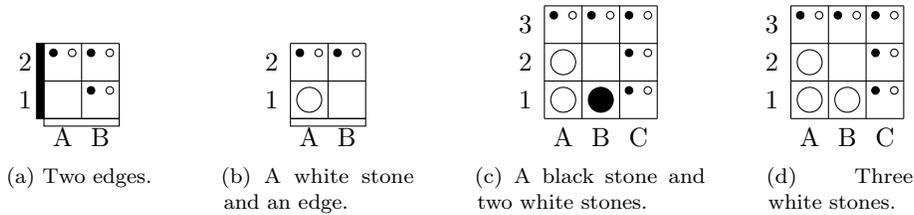
\begin{figure}
\subfloat[Two edges.]{
\label{fig:corner}
\begin{tikzpicture}
\sboard{2}{2} \sposition{}{}
\draw [fill] ({\s / 2 - \bor},{\s / 2}) rectangle ({\s / 2},{(2 + \s)/ 2}) {};
\draw ({\s / 2},{\s / 2 - \bor}) rectangle ({(2 + \s)/ 2},{\s / 2}) {};

\possb{1}{2} \possw{1}{2}
\possb{2}{1} \possw{2}{1}
\possb{2}{2} \possw{2}{2}

\node (A) at (0.5,0) {A}; \node (B) at (1  ,0) {B};
\node (1) at (0,0.5) {1}; \node (2) at (0,1  ) {2};
\end{tikzpicture}
\hspace{3mm}
}
\hfill
\subfloat[A white stone and an edge.]{
\label{fig:edge}
\hspace{2mm}
\begin{tikzpicture}
\sboard{2}{2} \sposition{1-1}{}
\draw ({\s / 2},{\s / 2 - \bor}) rectangle ({(2 + \s)/ 2},{\s / 2}) {};

\possb{1}{2} \possw{1}{2}
\possb{2}{2} \possw{2}{2}

\node (A) at (0.5,0) {A}; \node (B) at (1  ,0) {B};
\node (1) at (0,0.5) {1}; \node (2) at (0,1  ) {2};
\end{tikzpicture}
\hspace{4mm}
}
\hfill
\subfloat[A black stone and two white stones.]{
\label{fig:21}
\hspace{2mm}
\begin{tikzpicture}
\sboard{3}{3}
\sposition{1-1,1-2}{2-1}

\possb{1}{3} \possw{1}{3}
\possb{2}{3} \possw{2}{3}
\possb{3}{1} \possw{3}{1}
\possb{3}{2} \possw{3}{2}
\possb{3}{3} \possw{3}{3}

\node (A) at (0.5,0) {A}; \node (B) at (1  ,0) {B}; \node (C) at (1.5,0) {C};
\node (1) at (0,0.5) {1}; \node (2) at (0,1  ) {2}; \node (3) at (0,1.5) {3};
\end{tikzpicture}
\hspace{4mm}
}
\hfill
\subfloat[Three white stones.]{
\label{fig:case3}
\hspace{-4mm}
\begin{tikzpicture}
\sboard{3}{3}
\sposition{1-1,1-2,2-1}{}

\possb{1}{3} \possw{1}{3}
\possb{2}{3} \possw{2}{3}
\possb{3}{1} \possw{3}{1}
\possb{3}{2} \possw{3}{2}
\possb{3}{3} \possw{3}{3}

\node (A) at (0.5,0) {A}; \node (B) at (1  ,0) {B}; \node (C) at (1.5,0) {C};
\node (1) at (0,0.5) {1}; \node (2) at (0,1  ) {2}; \node (3) at (0,1.5) {3};
\end{tikzpicture}
}
\caption{Case analysis for the bottom left surroundings of the empty square.}
\label{fig:split}
\end{figure}

\section{A square in a corner or on the edge of the board}\label{corner-edge}
If there is an empty square in a corner, as depicted in Figure~\ref{fig:corner}, then placing a stone on that very square is a legal move for at least one player.

If there is an empty square on an edge, we start from the situation in Figure~\ref{fig:edge} and use the following reasoning to constrain the surrounding and obtain Figure~\ref{fig:bottom}.
C2 needs a white stone to forbid White's move B1, and C1 cannot be white.
A2 needs a black stone to forbid Black's move B1.
B2 needs to be empty to forbid White's and Black's move B1.
C1 cannot be empty to forbid White's move C1.
A3 needs a white stone to forbid White's move A1B1-B2, and B3 cannot be white.
C3 needs a black stone to forbid Black's move C1B1-B2, and B3 cannot be black.
Similarly, A4 needs a black stone, C4 needs a white stone, and B4 needs to be empty, so as to forbid White's move A1B2-B3 and Black's move C1B2-B3.

But then, C3B2-B1 is a legal move for Black.
\begin{figure}
\centering
\begin{tikzpicture}
\sboard{3}{4} \sposition{1-1,1-3,3-2,3-4}{1-2,1-4,3-1,3-3}
\draw ({\s / 2},{\s / 2 - \bor}) rectangle ({(3 + \s)/ 2},{\s / 2}) {};

\node (A) at (0.5,0) {A}; \node (B) at (1  ,0) {B}; \node (C) at (1.5,0) {C};
\node (1) at (0,0.5) {1}; \node (2) at (0,1  ) {2}; \node (3) at (0,1.5) {3}; \node (4) at (0,2  ) {4};
\end{tikzpicture}
\caption{The case in Figure~\ref{fig:edge} with a few deducible constraints filled in.}
\label{fig:bottom}
\end{figure}
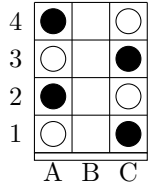

\section{A square with two white stones and a black stone}\label{two-one}
We start from the situation in Figure~\ref{fig:21} and use the following reasoning to constrain the surrounding and obtain Figure~\ref{fig:case21}.
The C2 square cannot contain a white stone, otherwise B2 is a legal move for White.
Similarly, B3 cannot contain a black stone, otherwise B2 is a legal move for Black.
To forbid White's move B2, there should be a white stone on C1 or on C3 (Figure~\ref{fig:21split}).

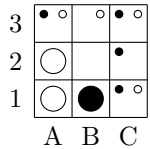
\begin{figure}
\centering
\begin{tikzpicture}
\sboard{3}{3} \sposition{1-1,1-2}{2-1}

\possb{1}{3} \possw{1}{3}
\possw{2}{3}
\possb{3}{1} \possw{3}{1}
\possb{3}{2}
\possb{3}{3} \possw{3}{3}
\node (A) at (0.5,0) {A}; \node (B) at (1  ,0) {B}; \node (C) at (1.5,0) {C};
\node (1) at (0,0.5) {1}; \node (2) at (0,1  ) {2}; \node (3) at (0,1.5) {3};
\end{tikzpicture}
\caption{The case in Figure~\ref{fig:21} with a few deducible constraints filled in.}
\label{fig:case21}
\end{figure}

\begin{figure}
\centering
\subfloat[Assume C1 is white.]{
\label{fig:case21a}
\hspace{19mm}
\begin{tikzpicture}
\sboard{3}{3} \sposition{1-1,1-2,3-1}{2-1}

\possb{1}{3} \possw{1}{3}
\possw{2}{3}
\possb{3}{3} \possw{3}{3}

\node (A) at (0.5,0) {A}; \node (B) at (1  ,0) {B}; \node (C) at (1.5,0) {C};
\node (1) at (0,0.5) {1}; \node (2) at (0,1  ) {2}; \node (3) at (0,1.5) {3};
\end{tikzpicture}
\hspace{19mm}
}
\hfill
\subfloat[Assume C3 is white.]{
\label{fig:case21b}
\hspace{12mm}
\begin{tikzpicture}
\sboard{3}{3} \sposition{1-1,1-2,3-3}{1-3,2-1}

\possb{3}{1} \possw{3}{1}
\possb{3}{2}

\node (A) at (0.5,0) {A}; \node (B) at (1  ,0) {B}; \node (C) at (1.5,0) {C};
\node (1) at (0,0.5) {1}; \node (2) at (0,1  ) {2}; \node (3) at (0,1.5) {3};
\end{tikzpicture}
\hspace{12mm}
}
\caption{Case analysis for Figure~\ref{fig:case21}, either C1 is white or C3 is white.}
\label{fig:21split}
\end{figure}
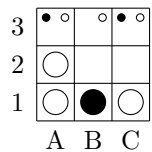
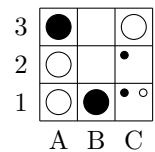

\subsection{Figure~\ref{fig:case21a}}
C2 cannot contain a black stone due to the diagonal rule (since B2 is empty by assumption), and has to be empty.
Now, we distinguish two subcases: either B3 is white, or it is empty (Figure~\ref{fig:21asplit}).

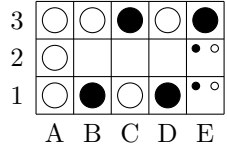
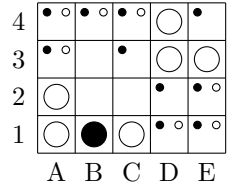
\begin{figure}
\centering
\subfloat[Assume B3 is white.]{
\label{fig:case21aa}
\centering
\hspace{10mm}
\begin{tikzpicture}
\sboard{5}{3} \sposition{1-1,1-2,3-1,1-3,2-3,4-3}{2-1,3-3,4-1,5-3}

\possb{5}{1} \possw{5}{1}
\possb{5}{2} \possw{5}{2}

\node (A) at (0.5,0) {A}; \node (B) at (1  ,0) {B}; \node (C) at (1.5,0) {C}; \node (D) at (2  ,0) {D}; \node (E) at (2.5,0) {E};
\node (1) at (0,0.5) {1}; \node (2) at (0,1  ) {2}; \node (3) at (0,1.5) {3};
\end{tikzpicture}
\hspace{10mm}
}
\hfill
\subfloat[Assume B3 is empty.]{
\label{fig:case21ab}
\centering
\hspace{10mm}
\begin{tikzpicture}
\sboard{5}{4} \sposition{1-1,1-2,3-1,4-3,4-4,5-3}{2-1}

\possb{1}{3} \possw{1}{3}
\possb{1}{4} \possw{1}{4}
\possb{2}{4} \possw{2}{4}
\possb{3}{3}
\possb{3}{4} \possw{3}{4}
\possb{4}{1} \possw{4}{1}
\possb{4}{2}
\possb{5}{1} \possw{5}{1}
\possb{5}{2} \possw{5}{2}
\possb{5}{4}

\node (A) at (0.5,0) {A}; \node (B) at (1  ,0) {B}; \node (C) at (1.5,0) {C}; \node (D) at (2  ,0) {D}; \node (E) at (2.5,0) {E};
\node (1) at (0,0.5) {1}; \node (2) at (0,1  ) {2}; \node (3) at (0,1.5) {3}; \node (4) at (0,2  ) {4};
\end{tikzpicture}
\hspace{10mm}
}
\caption{Case analysis for Figure~\ref{fig:case21a}: B3 is either white or empty.}
\label{fig:21asplit}
\end{figure}

\subsubsection{Figure~\ref{fig:case21aa}}
A3 contains a white stone, by the diagonal rule.
C3 needs a black stone to forbid Black's move B2.
D1 needs a black stone to forbid Black's move B1B2-C2.
D3 needs a white stone to forbid White's move C1B2-C2.
D2 needs to be empty, by the diagonal rule on stones C1 and C3.
E3 needs a black stone to forbid Black's move B1C2-D2.

But, in that situation, D3C2-B2 is a legal move for White.

\subsubsection{Figure~\ref{fig:case21ab}}
To forbid White's move C2, D3 needs a white stone and C3 and D2 cannot contain a white stone.
To forbid White's move D3C2-B2, there should be white stones on E3 and D4 and E4 should not contain a white stone.
Now, we distinguish two subcases: either C3 is black, or it is empty (Figure~\ref{fig:21absplit}).

\begin{figure}
\centering
\subfloat[Assume C3 is black.]{
\label{fig:case21aba}
\centering
\hspace{10mm}
\begin{tikzpicture}
\sboard{5}{4} \sposition{1-1,1-2,3-1,4-3,4-4,5-3}{2-1,3-3,4-1}

\possw{1}{3}
\possb{1}{4} \possw{1}{4}
\possb{2}{4} \possw{2}{4}
\possb{3}{4} \possw{3}{4}
\possb{5}{1} \possw{5}{1}
\possb{5}{2} \possw{5}{2}
\possb{5}{4}

\node (A) at (0.5,0) {A}; \node (B) at (1  ,0) {B}; \node (C) at (1.5,0) {C}; \node (D) at (2  ,0) {D}; \node (E) at (2.5,0) {E};
\node (1) at (0,0.5) {1}; \node (2) at (0,1  ) {2}; \node (3) at (0,1.5) {3}; \node (4) at (0,2  ) {4};
\end{tikzpicture}
\hspace{10mm}
}
\hfill
\subfloat[Assume C3 is empty.]{
\label{fig:case21abb}
\centering
\hspace{10mm}
\begin{tikzpicture}
\sboard{5}{4} \sposition{1-1,1-2,2-4,3-1,4-3,4-4,5-3}{2-1,4-2}

\possb{1}{3} \possw{1}{3}
\possb{1}{4} \possw{1}{4}
\possb{3}{4} \possw{3}{4}
\possb{4}{1} \possw{4}{1}
\possb{5}{1} \possw{5}{1}
\possb{5}{2} \possw{5}{2}
\possb{5}{4}

\node (A) at (0.5,0) {A}; \node (B) at (1  ,0) {B}; \node (C) at (1.5,0) {C}; \node (D) at (2  ,0) {D}; \node (E) at (2.5,0) {E};
\node (1) at (0,0.5) {1}; \node (2) at (0,1  ) {2}; \node (3) at (0,1.5) {3}; \node (4) at (0,2  ) {4};
\end{tikzpicture}
\hspace{10mm}
}
\caption{Case analysis for Figure~\ref{fig:case21ab}: C3 is either black or empty.}
\label{fig:21absplit}
\end{figure}
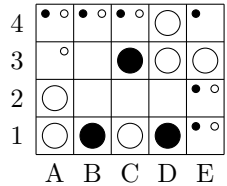
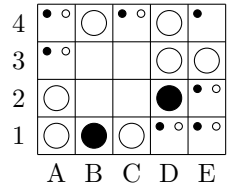

\paragraph{Figure~\ref{fig:case21aba}}
A3 cannot contain a black stone to forbid Black's move B3.
D1 needs a black stone to forbid Black's move C3B2-C2.
D2 cannot contain a black stone by the diagonal rule, and has to be empty.

But then, move B1C2-D2 is legal for Black.

\paragraph{Figure~\ref{fig:case21abb}}
B4 needs a white stone to forbid White's move C3.
D2 needs a black stone to forbid Black's move C3.

But then, B1C2-C3 is legal for Black.

\subsection{Figure~\ref{fig:case21b}}
A3 needs a black stone to forbid Black's move B2.
B3 cannot contain a white stone by the diagonal rule (with A2) and  has to be empty.
Now, we distinguish two subcases: either C2 is black, or it is empty (Figure~\ref{fig:21bsplit}).

\begin{figure}
\centering
\subfloat[Assume C2 is black.]{
\label{fig:case21ba}
\hspace{14mm}
\begin{tikzpicture}
\sboard{3}{3} \sposition{1-1,1-2,3-3}{2-1,1-3,3-1,3-2}

\node (A) at (0.5,0) {A}; \node (B) at (1  ,0) {B}; \node (C) at (1.5,0) {C};
\node (1) at (0,0.5) {1}; \node (2) at (0,1  ) {2}; \node (3) at (0,1.5) {3};
\end{tikzpicture}
\hspace{14mm}
}
\hfill
\subfloat[Assume C2 is empty.]{
\label{fig:case21bb}
\hspace{16mm}
\begin{tikzpicture}
\sboard{4}{3} \sposition{1-1,1-2,3-3,4-1}{2-1,1-3}

\possb{3}{1}
\possb{4}{2}
\possb{4}{3} \possw{4}{3}

\node (A) at (0.5,0) {A}; \node (B) at (1  ,0) {B}; \node (C) at (1.5,0) {C}; \node (D) at (2.0,0) {D};
\node (1) at (0,0.5) {1}; \node (2) at (0,1  ) {2}; \node (3) at (0,1.5) {3};
\end{tikzpicture}
\hspace{16mm}
}
\caption{Case analysis for Figure~\ref{fig:case21b}: C2 is either black or empty.}
\label{fig:21bsplit}
\end{figure}
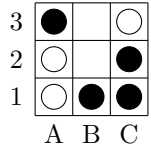
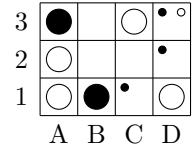

\subsubsection{Figure~\ref{fig:case21ba}}
To forbid White's move C3B2-B3 and Black's move A3B2-B3, there should be at least one white stone and one black stone in $\{$A4, C4$\}$.
So, we distinguish further between $\{$A4 black, C4 white$\}$ and $\{$A4 white, C4 black$\}$ (Figure~\ref{fig:21basplit}).

\begin{figure}
\subfloat[Assume A4 is black and C4 is white.]{
\label{fig:case21baa}
\hspace{14mm}
\begin{tikzpicture}
\sboard{4}{5}
\sposition{1-2,1-3,3-4,3-5}{2-2,1-4,1-5,3-1,3-2,3-3,4-3}

\possb{1}{1} \possw{1}{1}
\possb{2}{1} \possw{2}{1}
\possb{2}{5} \possw{2}{5}
\possb{4}{1} \possw{4}{1}
\possb{4}{2} \possw{4}{2}
\possb{4}{4} \possw{4}{4}
\possb{4}{5} \possw{4}{5}

\node (A) at (0.5,0) {A}; \node (B) at (1  ,0) {B}; \node (C) at (1.5,0) {C}; \node (D) at (2  ,0) {D};
\node (0) at (0,0.5) {0}; \node (1) at (0,1  ) {1}; \node (2) at (0,1.5) {2}; \node (3) at (0,2  ) {3}; \node (4) at (0,2.5) {4};
\end{tikzpicture}
\hspace{14mm}
}
\hfill
\subfloat[Assume A4 is white and C4 is black.]{
\label{fig:case21bab}
\hspace{15mm}
\begin{tikzpicture}
\sboard{3}{5}
\sposition{1-1,1-2,3-3,1-4,1-5}{2-1,1-3,3-1,3-2,3-4,3-5}

\possb{2}{5} \possw{2}{5}

\node (A) at (0.5,0) {A}; \node (B) at (1  ,0) {B}; \node (C) at (1.5,0) {C};
\node (1) at (0,0.5) {1}; \node (2) at (0,1  ) {2}; \node (3) at (0,1.5) {3}; \node (4) at (0,2  ) {4}; \node (5) at (0,2.5) {5};
\end{tikzpicture}
\hspace{16mm}
}
\caption{Case analysis for Figure~\ref{fig:case21ba}: the contents of A4 and C4 is white.}
\label{fig:21basplit}
\end{figure}
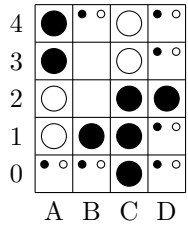
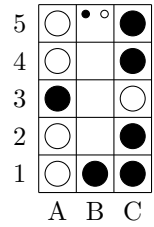

\paragraph{Figure~\ref{fig:case21baa}.}
D2 needs a black stone to forbid Black's move C2B2-B3, and C0 needs a black stone to forbid Black's move C1B2-B3.

But then, B1B2-B3 is a legal move for Black.

\paragraph{Figure~\ref{fig:case21bab}.}
By the diagonal rule, square B4 has to be empty.
C5 (as well as D4) needs a black stone to forbid Black's move C4B3-B2.
Symmetrically, A5 needs a white stone to forbid White's move A4B3-B2.

But then, C3B2-B4 is a legal move for White.

\subsubsection{Figure~\ref{fig:case21bb}.}
To forbid White's move C2, there should be a white stone on D1 but no white stones on C1 nor D2.
We distinguish two subcases: C1 contains a black stone, or it is empty (Figure~\ref{fig:21bbsplit}).

\begin{figure}
\subfloat[Assume C1 is black.]{
\label{fig:case21bba}
\hspace{10mm}
\begin{tikzpicture}
\sboard{5}{3}
\sposition{1-1,1-2,3-3,4-1,5-3}{2-1,1-3,3-1,5-1,4-3}

\node (A) at (0.5,0) {A}; \node (B) at (1  ,0) {B}; \node (C) at (1.5,0) {C}; \node (D) at (2  ,0) {D}; \node (E) at (2.5,0) {E};
\node (1) at (0,0.5) {1}; \node (2) at (0,1  ) {2}; \node (3) at (0,1.5) {3};
\end{tikzpicture}
\hspace{10mm}
}
\hfill
\subfloat[Assume C1 is empty.]{
\label{fig:case21bbb}
\hspace{10mm}
\begin{tikzpicture}
\sboard{5}{4}
\sposition{1-1,1-2,1-3,2-1,3-4,4-1,4-2,5-2}{1-4,2-2}

\possb{4}{3}
\possb{4}{4} \possw{4}{4}
\possb{5}{1} \possw{5}{1}
\possb{5}{3} \possw{5}{3}
\possb{5}{4} \possw{5}{4}

\node (A) at (0.5,0) {A}; \node (B) at (1  ,0) {B}; \node (C) at (1.5,0) {C}; \node (D) at (2  ,0) {D}; \node (E) at (2.5,0) {E};
\node (0) at (0,0.5) {0}; \node (1) at (0,1  ) {1}; \node (2) at (0,1.5) {2}; \node (3) at (0,2  ) {3};
\end{tikzpicture}
\hspace{11mm}
}
\caption{Case analysis for Figure~\ref{fig:case21bb}: C1 is either black or empty.}
\label{fig:21bbsplit}
\end{figure}
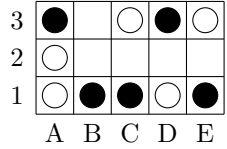
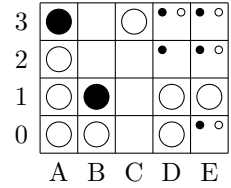

\paragraph{Figure~\ref{fig:case21bba}.}
D3 needs a black stone to forbid Black's move C2 (since B3 is empty).
Then, by the diagonal rule, square D2 can only be empty.
E3 needs a white stone to forbid White's move C3D2-C2.
E1 needs a black stone to forbid Black's move D3C2-D2.

But then, D1C2-B2 is a legal move for White.

\paragraph{Figure~\ref{fig:case21bbb}.}
The only way to forbid White's move D1C2-B2 is to add two white stones on D0 and on E1.
To forbid White's move C1, there should be a white stone on B0 (and a white stone on A0, by the diagonal rule), and no white stones on C0.
C0 cannot contain a black stone because of the diagonal rule

But then, C0 is a legal move for White.

\section{A square with three white stones}\label{three-nil}
We start from the situation in Figure~\ref{fig:case3}.
To forbid White's move B2, there should be a white stone on C3, but no white stones on B3 nor C2.
To forbid Black's move B2, there should be a black stone on C1 or A3, say C1 w.l.o.g.~(see Figure~\ref{fig:case3bis}).

\begin{figure}
\centering
\begin{tikzpicture}
\sboard{3}{3}
\sposition{1-1,1-2,2-1,3-3}{3-1}

\possb{1}{3} \possw{1}{3}
\possb{2}{3}
\possb{3}{2}

\node (A) at (0.5,0) {A}; \node (B) at (1  ,0) {B}; \node (C) at (1.5,0) {C};
\node (1) at (0,0.5) {1}; \node (2) at (0,1  ) {2}; \node (3) at (0,1.5) {3};
\end{tikzpicture}
\caption{The case in Figure~\ref{fig:case3} with a few deducible constraints filled in.}
\label{fig:case3bis}
\end{figure}
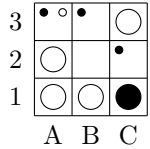

Therefore, B3 and C2 are empty or contain a black stone.
They cannot both contain a black stone since B2 is empty.
We thus distinguish three cases: B3 and C2 are empty, B3 contains a black stone, and C2 contains a black stone (Figure~\ref{fig:case3split}).

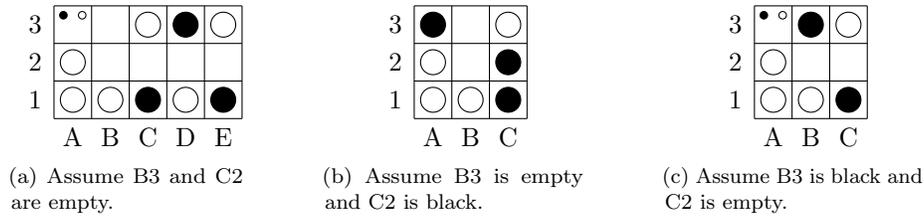
\begin{figure}
\subfloat[Assume B3 and C2 are empty.]{
\label{fig:case3a}
\begin{tikzpicture}
\sboard{5}{3}
\sposition{1-1,1-2,2-1,3-3,4-1,5-3}{3-1,4-3,5-1}

\possb{1}{3} \possw{1}{3}

\node (A) at (0.5,0) {A}; \node (B) at (1  ,0) {B}; \node (C) at (1.5,0) {C}; \node (D) at (2  ,0) {D}; \node (E) at (2.5,0) {E};
\node (1) at (0,0.5) {1}; \node (2) at (0,1  ) {2}; \node (3) at (0,1.5) {3};
\end{tikzpicture}}
\hfill
\subfloat[Assume B3 is empty and C2 is black.]{
\label{fig:case3b}
\hspace{5mm}
\begin{tikzpicture}
\sboard{3}{3}
\sposition{1-1,1-2,2-1,3-3}{1-3,3-1,3-2}

\node (A) at (0.5,0) {A}; \node (B) at (1  ,0) {B}; \node (C) at (1.5,0) {C};
\node (1) at (0,0.5) {1}; \node (2) at (0,1  ) {2}; \node (3) at (0,1.5) {3};
\end{tikzpicture}
\hspace{5mm}
}
\hfill
\subfloat[Assume B3 is black and C2 is empty.]{
\label{fig:case3c}
\hspace{5mm}
\begin{tikzpicture}
\sboard{3}{3}
\sposition{1-1,1-2,2-1,3-3}{2-3,3-1}

\possb{1}{3} \possw{1}{3}
\node (A) at (0.5,0) {A}; \node (B) at (1  ,0) {B}; \node (C) at (1.5,0) {C};
\node (1) at (0,0.5) {1}; \node (2) at (0,1  ) {2}; \node (3) at (0,1.5) {3};
\end{tikzpicture}
\hspace{5mm}
}
\caption{Case analysis for Fig~\ref{fig:case3bis}: the contents of B3 and C2.}
\label{fig:case3split}
\end{figure}

\subsection{Figure~\ref{fig:case3a}}
D3 needs a black stone to forbid Black's move C2.
D1 needs a white stone to forbid White's move C3B2-C2.
E3 needs a white stone to forbid White's move C3B2-D2.
E1 needs a black stone to forbid Black's move D3C2-D2.

But then, D1C2-B2 is legal for White.

\subsection{Figure~\ref{fig:case3b}}
A3 needs a black stone to forbid Black's move B2.
This case is equivalent to the case of Figure~\ref{fig:case21ba} under color and spatial symmetry.

\subsection{Figure~\ref{fig:case3c}}
Let us consider cases for the contents of A3.
If A3 contains a black stone, then we obtain a position equivalent to Figure~\ref{fig:case21ba} under color and spatial symmetry.

If A3 is empty or white, then a similar proof to Figure~\ref{fig:case3a} still holds.
Indeed, D3 needs a black stone to forbid Black's move B3B2-C2 and D1 needs a white stone to forbid White's move C3B2-C2.
The same way, E3 needs a white stone to forbid White's move C3B2-D2, and then E1 needs a black stone to forbid Black's move C1B2-D2.

But then, D1C2-B2 is legal for White.

\end{document}